\documentclass[aps,twocolumn,superscriptaddress]{revtex4-2}

\usepackage{amsmath}
\usepackage{siunitx}
\usepackage[inline]{enumitem}

\usepackage{xcolor}  % For corrections
\usepackage{graphicx}
\usepackage{amsfonts}
\usepackage{mathtools}
\usepackage{booktabs}
\usepackage{xcolor}
\usepackage{braket}
\usepackage{MnSymbol}
\usepackage{tikz}
\usepackage{cancel}
\usepackage{setspace}
\usepackage{upgreek}
\usepackage{amsthm}
\makeatletter
\newtheorem*{rep@theorem}{\rep@title}
\newcommand{\newreptheorem}[2]{%
\newenvironment{rep#1}[1]{%
 \def\rep@title{#2 \ref*{##1}, repeated}%
 \begin{rep@theorem}}%
 {\end{rep@theorem}}}
\makeatother

\newtheorem{lemma}{Lemma}
\newreptheorem{lemma}{Lemma}

\newtheorem{definition}{Definition}

\newtheorem{proposition}{Proposition}
\newreptheorem{proposition}{Proposition}

\newreptheorem{corollary}{Corollary}

\def\norm#1{ {|\hspace{-.015in}|#1|\hspace{-.015in}|} }
\def\norm#1{ {\left \Vert #1 \right \Vert} }
\def\smallnorm#1{ {\Vert  #1 \Vert } }
\def\abs#1{ {|#1|}}
\def\bigabs#1{ {\left|#1 \right|}}

\newcommand{\expect}[1]{\langle #1 \rangle}
\newcommand{\vecket}[1]{|{#1} \rrangle}

\newcommand{\vecbra}[1]{\llangle {#1}|}
\newcommand\numeq[1]%
  {\stackrel{\scriptscriptstyle(\mkern-1.5mu#1\mkern-1.5mu)}{=}}
\newcommand\numleq[1]%
  {\stackrel{\scriptscriptstyle(\mkern-1.5mu#1\mkern-1.5mu)}{\leq}}
\newcommand\numgeq[1]%
  {\stackrel{\scriptscriptstyle(\mkern-1.5mu#1\mkern-1.5mu)}{\geq}}
\begin{document}
\title{Noise robustness of problem-to-simulator mappings for quantum many-body physics}
\author{Rahul Trivedi}
\email{rahul.trivedi@mpq.mpg.de}
\author{J. Ignacio Cirac}
\address{Max Planck Institute of Quantum Optics, Hans Kopfermann Str. 1, Garching, Germany - 85748.}
\address{Munich Center for Quantum Science and Technology, Schellingstr. 4, M\"unchen, Germany - 80799.}
\begin{abstract}
    Simulating quantum dynamics on digital or analog quantum simulators often requires ``problem-to-simulator" mappings such as trotterization, floquet-magnus expansion or perturbative expansions. When the simulator is noiseless, it is well understood that these problem-to-simulator mappings can be made as accurate as desired at the expense of simulator run-time. However, precisely because the simulator has to be run for a longer time to increase its accuracy, it is expected that noise in the quantum simulator catastrophically effects the simulator output. We show that, contrary to this expectation, these mappings remain stable to noise when considering the task of simulating dynamics of local observables in quantum lattice models. Specifically, we prove that in all of these mappings, local observables can be determined to a system-size independent, precision that scales sublinearly with the noise-rate in the simulator. Our results provide theoretical evidence that quantum simulators can be used for solving problems in many-body physics without or with modest error correction.
\end{abstract}
\maketitle
\raggedbottom

Quantum simulators and quantum computers can potentially solve problems arising in both high and low-energy physics which are considered challenging on classical computers \cite{daley2022practical, bauer2023quantum, cao2019quantum, altman2021quantum}. Recent technological progress has resulted in quantum devices with $\sim 100$ qubits with error rates reduced to less than $1\%$ in several platforms \cite{arute2019quantum, wienand2024emergence, ebadi2021quantum, bourgund2025formation, bluvstein2022quantum}. Furthermore, the first few demonstrations of quantum error correction with more than 50 physical qubits have been achieved \cite{bluvstein2024logical, google2023suppressing, bourgund2025formation}. Nevertheless, achieving fault-tolerance at a scale large enough to solve ``classically intractable" problems still remains a technological challenge and, in the near term, available quantum devices will still be operated with only modest error correction \cite{beverland2022assessing}.

In the worst case, even a small but constant rate of errors can accumulate and spoil the simulator output. Nevertheless, growing theoretical evidence shows that problems in many-body physics might avoid this worst-case scenario \cite{granet2025dilution, cubitt2015stability, trivedi2024quantum, schiffer2024quantum, kashyap2025accuracy, bachmann2012automorphic, hastings2005quasiadiabatic, heyl2019quantum}. In particular, it has been shown that quantum simulation of physically relevant observables will have an error $f(\gamma)$ that depend only on the noise rate $\gamma$ and not on the system size $N$. Furthermore, previous work has also provided evidence how reduction in the noise-rate $\gamma$, either by hardware improvements or via error correction, can make it exponentially harder for classical algorithms to obtain the same precision $f(\gamma)$ achievable by the noisy simulator thus hinting towards a ``quantum advantage" with respect to the hardware-limited precision with noisy quantum simulators \cite{trivedi2024quantum, kashyap2025accuracy}.
\begin{figure}[b]
    \centering
    \includegraphics[width=1.0\linewidth]{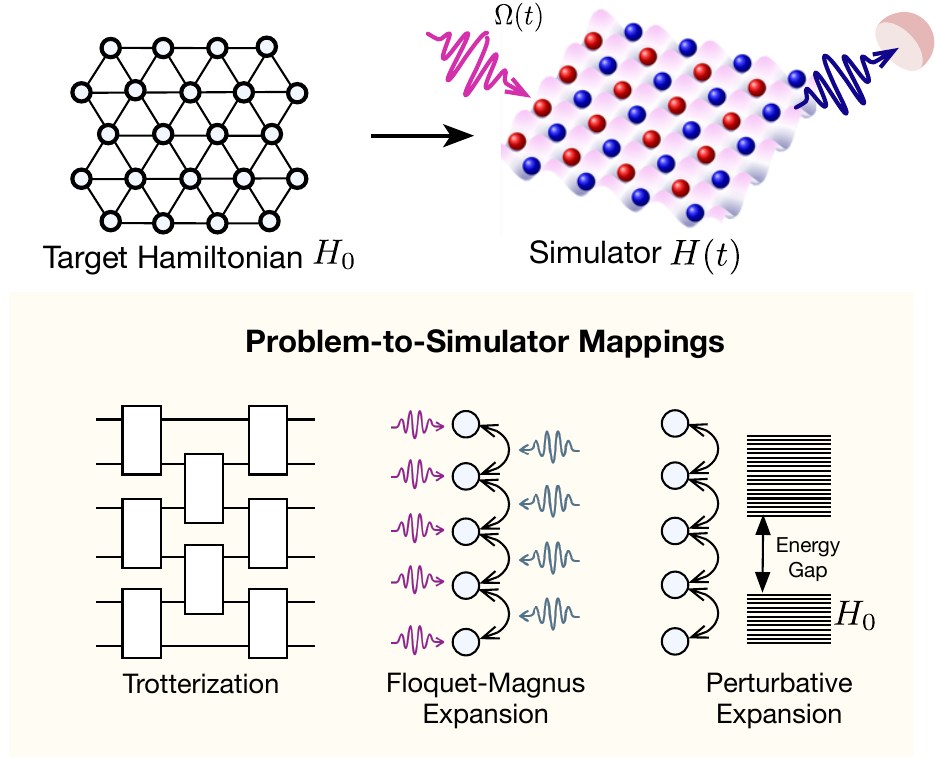}
    \caption{A target many-body Hamiltonian $H_0$ is not often natively available on the quantum simulator. To approximately map this Hamiltonian onto the simulator, several problem-to-simulator mappings can be employed, with the most common strategies being Trotterization, Floquet-Magnus expansions and perturbative expansions.}
    \label{fig:placeholder}
\end{figure}

However, the available theoretical results still make the assumption that the quantum simulator can natively implement the Hamiltonian of the target many-body problem. This is typically not the case in most available quantum simulators, which only implement a restricted families of Hamiltonians \cite{ebadi2021quantum, arute2019quantum}. On the other hand, models of interest in low and high-energy physics are typically more complicated---for instance, several problems in condensed matter physics often have Hamiltonians complex interactions which are not directly available on analog simulators \cite{vanderstraeten2018quasiparticles, binder1980phase, rademaker2020slow, panchenko2014parisi, sachdev1993gapless}. Similarly, many-body problems in High energy physics, such as simulating lattice gauge theories, often have 4-local Hamiltonians which are again not easily available in quantum simulators \cite{kogut1975hamiltonian, bodwin1987lattice, kogut1979introduction, kogut1983lattice, chandrasekharan1997quantum, zohar2015formulation}. To simulate these problems, we require the use of a ``problem-to-simulator" mapping such as Trotterization \cite{childs2021theory, gong2024complexity, suzuki1993general, mizuta2025trotterization, csahinouglu2021hamiltonian, burak2021hamiltonian, tran2020destructive, su2021nearly, tong2021provably, childs2019nearly, lloyd1996universal, hatano2005finding}, Floquet-Magnus \cite{abanin2017effective, abanin2017rigorous, hung2016quantum, mori2016rigorous, zohar2017digital} or perturbative expansions \cite{abanin2017rigorous,zohar2015formulation, zohar2013quantum, bravyi2008quantum}. In the absence of any errors, it is well understood that these mappings can be tuned to be as accurate as desired. However, increasing the accuracy of these mappings requires a corresponding overhead in the simulation time---for example, Trotterization error can be reduced by making the trotter time-step smaller consequently increasing the total circuit depth. Due to this overhead, noise in the simulator would have a longer time to accumulate and corrupt the output of the simulator. It could thus be possible that, even when the original many-body problem is robust to perturbation, the problem-to-simulator mapping could result in a catastrophic propagation of errors.

In this paper, we show that contrary to this expectation, quantum simulation of local observables in dynamics of geometrically local lattice models remains robust to errors. If the target Hamiltonian was natively available on the simulator, then in the presence of noise rate $\gamma$, it follows straightforwardly from the Lieb-Robinson bounds that local observables would incur a simulation error $O(\gamma \tau^{d + 1})$, where $d$ is the dimensionality of the lattice and $\tau$ is the target evolution time. We establish that when using the problem-to-simulator mappings such as Trotterization, Floquet-Magnus expansion and perturbative expansions, local observables instead incur an error $O(\gamma^\alpha \tau^{d + 1})$ where the exponent $\alpha \in (0,  1]$ depends on the specifics of the problem-to-simulator mapping (e.g.~order of the Trotter-formula or Floquet Magnus expansion). Importantly, this error remains system-size independent despite the simulation-time overhead associated with the problem-to-simulator mapping. We remark that the results that we show hold for \emph{any} local observable and in particular for observables that are expressible as a sum of local observables.

\emph{Setup}. We focus on the problem of simulating a local observable $O$ in a target $N-$qubit many-body geometrically-local Hamiltonian $H_0$ on a $d-$dimensional lattice after evolution for time $\tau$ i.e., to use the simulator to compute $\mathcal{O}_\text{target} =\text{Tr}[ O\exp(-iH_0 \tau) \rho(0)\exp(iH_0 \tau)]$ where $\rho(0)$ is the initial state of the qubits. To simulate the dynamics of $H_0$, the quantum simulator, in the absence of noise, implements a possibly time-dependent geometrically-local simulator Hamiltonian $H(t)$ which is evolved for a time $t_\text{sim}$. Unless the Hamiltonian $H_0$ is natively available on the simulator, typically $t_\text{sim} \gg \tau$. We will model the noise in the quantum simulator via the master equation
\[
\dot{\rho}(t) = \mathcal{L}_\gamma (t) \rho(t) \text{ with }\mathcal{L}_\gamma(t) = -i[H(t), \cdot] + \gamma \sum_\alpha \mathcal{N}_\alpha,
\]
where $\mathcal{N}_\alpha$ are geometrically local Lindbladians with $\smallnorm{\mathcal{N}_\alpha}_\diamond \leq 1$ and $\gamma$ can be interpreted as the noise rate. {On the simulator, we will measure the expected value of the observable $O$ to obtain $\mathcal{O}_\text{sim} = \text{Tr}[O\mathcal{E}_\gamma(t_\text{sim})(\rho(0))]$ where $\mathcal{E}_\gamma(t) = \mathcal{T}\exp(\int_0^t \mathcal{L}_\gamma(s) ds)$. In the following results, we will characterize the error  $\Delta \mathcal{O} = \abs{\mathcal{O}_\text{target}-\mathcal{O}_\text{sim}}$ ---in particular, we will show that
\begin{align}\label{eq:target_scaling}
    \Delta \mathcal{O}\leq O(\gamma^\alpha \tau^{d + 1}),
\end{align}
independent of the system size and for $\alpha \in (0, 1]$.
}

\emph{Results}. We first consider simulating a target geometrically local Hamiltonian $H_0$ with $p^\text{th}$ order Trotter formula i.e., the time-dependent simulator Hamiltonian $H(t)$ is chosen to correspond to the trotter circuit implementing $\exp(-iH_0\tau)$. Recall that given a geometrically local Hamiltonian $H_0$ in $d$ dimension, we can rewrite it as
\[
H_0 = \sum_{j = 1}^K H_j,
\]
where $H_1, H_2 \dots H_K$ are themselves geometrically local commuting Hamiltonians and $K$ is independent of the system size $N$. For instance, for a 1D nearest neighbour Hamiltonian $H = \sum_{x = 1}^{N - 1} h_{x, x + 1}$ where $h_{x, x+1}$ acts only on sites $x$ and $x + 1$, we can use $H_1 = \sum_{x} h_{2x - 1, 2x}$ and $H_2 = \sum_{x}h_{2x, 2x + 1}$. A similar decomposition can be made in 2 and higher dimensions. A $p^\text{th}$ order trotter formula with $s_P$ stages for $\exp(-i\varepsilon H_0)$ is then specified by a set of real numbers $\alpha_{i, j}$ for $i\in\{1, 2 \dots s_P\}$ and $j \in \{1, 2\dots K\}$ \cite{childs2019nearly}:
\[
S_p(\varepsilon) = \prod_{i = 1}^{s_P} \exp(-i \alpha_{i, 1} \varepsilon H_1)  \dots  \exp(-i \alpha_{i,K}\varepsilon H_K),
\]
where the real numbers $\alpha_{i, j}$ are chosen such that
\[
\exp(-iH_0 \varepsilon) - S_p(\varepsilon) = O(\varepsilon^{p+ 1}).
\]
To simulate the dynamics upto time $\tau$ using the $p^\text{th}$ order Trotter formula, we simply divide up the total evolution time into $T$ time-steps each of time-width $\varepsilon = \tau /T$ and implement the unitary $U = S_p(\varepsilon)^T$. This unitary can also be considered to be a time-dependent Hamiltonian $H(t)$ simulating the original time-independent lattice model --- assuming that the unitary is implemented using a fixed universal gate-set, the simulation time of the simulator then scales as $t_\text{sim} = \Theta (T) = \Theta(\tau/\varepsilon)$. 
\begin{figure*}
    \centering
    \includegraphics[width=1.0\linewidth]{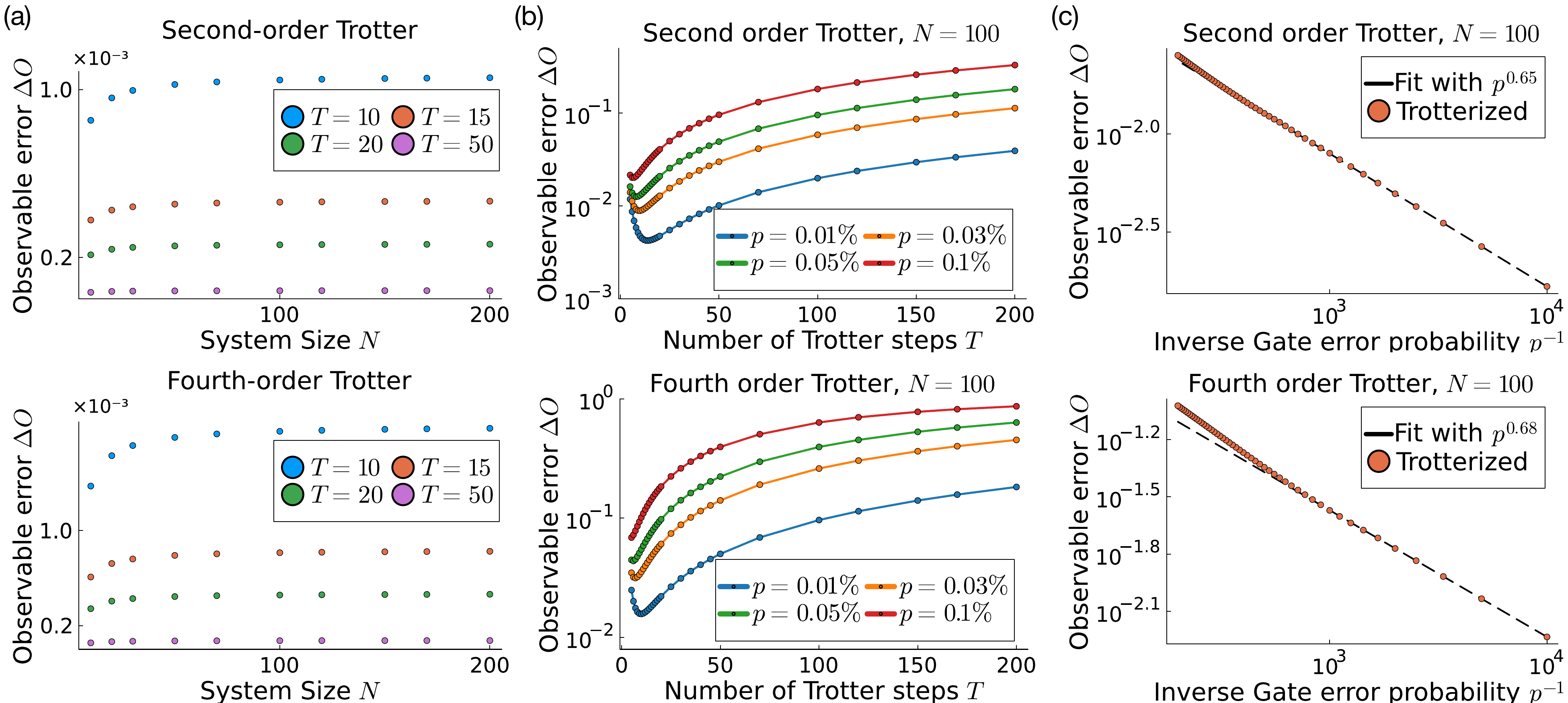}
    \caption{\textbf{Robustness of Trotterization}. For the numerical example, we consider a 1D chain of $N$ fermionic modes with annihilation operators $a_1, a_2 \dots a_N$, and a Gaussian Hamiltonian $H = h\sum_i a_i^\dagger a_i + ig \sum_i (a_i + \text{h.c.})(a_{i + 1} + \text{h.c.})$. The fermionic modes are initially in the state $\ket{0}^{\otimes N}$ and the observable $O = \sum_i a_i^\dagger a_i / N$ is measured at time $\tau$. We choose $h = 1.0, g = 0.5$ and $\tau = 1.0$ for all the calculations. (a) Observable error, in the absence of noise, due to Trotterization as a function of system size for different number of Trotter step $T$ for second and fourth-order Trotterization. We observe that the Trotter error in the local observable $O$ is independent of the system-size. (b) Observable error in the presence single-fermion depolarizing noise in the Trotterized circuit, with $p$ being the probability of depolarization per qubit per gate, as a function of the number of Trotter step. In the presence of noise, there is an optimal number of Trotter steps which depends on the noise-rate and minimizes the observable error. (c) The observable error at this optimal $T$ as a function of the noise probability $p$. }
    \label{fig:trotter}
\end{figure*}

Several rigorous analyses of Trotter errors and their scalings are available including those specialized to geometrically local lattice models \cite{childs2019nearly, childs2021theory}, those analyzing Trotter errors within low-energy subspaces \cite{mizuta2025trotterization, gong2024complexity, csahinouglu2021hamiltonian, burak2021hamiltonian} as well as connecting Trotter error to operator scrambling \cite{feng2025trotterization}. However, most analyses have focussed on characterizing the accuracy of the full many-body state prepared by the Trotterized evolution. For geometrically local Hamiltonians in particular, to ensure that the trace-norm Trotter error in the full state is $<\delta$, we require $t_\text{sim} = \Theta(T) =  \Theta(\tau (N\tau/\delta)^{1/p})$ \cite{childs2019nearly}. We first show that using the Lieb-Robinson bounds, this error bound can be improved to be uniform in $N$ if only accurate local observables are required.
\begin{lemma}[Trotter error for local observables]\label{lemma:trotter_local_obs_main}
  Suppose $O_X$ is an observable supported on $X$, and for any initial state $\rho(0)$, $\rho(\tau) = \exp(-iH_0\tau) \rho(0) \exp(iH_0\tau)$ is the state in the target evolution and $\hat{\rho}_{T}(\tau) = [S_p(\varepsilon)]^T \rho(0) [S_p^\dagger(\varepsilon)]^T$, with $\varepsilon = \tau / T$, is the state in the trotterized evolution with $T$ trotter steps, then
  \[
   \abs{\textnormal{Tr}(O\rho(\tau)) - \textnormal{Tr}(O\hat{\rho}_T(\tau))}  \leq O\bigg(\abs{X}^2\frac{\tau^{d+p+1}}{T^{p}}\bigg),
  \]
  where $d$ is the dimensionality of the lattice in which $H_0$ is embedded.
\end{lemma}
\noindent Thus, on a noiseless simulator, to simulate only local observables to an accuracy $\delta$, we require $T = \Theta(\tau(\tau^{d + 1}/\delta)^{1/p})$ independent of the system size. The proof of this lemma, provided in the supplement, uses the local-error sum representation for the Trotter error developed in Ref.~\cite{childs2019nearly} together with Lieb-Robinson bounds \cite{hastings2006spectral}.

This Trotter error bound allows us to establish that Trotterization is robust as a problem-to-simulator mapping when considering the impact of noise on quantum simulators. We make this precise in the following proposition.

\begin{proposition}[Noise-robustness of Trotterization]\label{prop:trotter}
    {When using the $p^\textnormal{th}$ order Trotterization with $T$ Trotter steps and with noise at rate $\gamma>0$, the observable error $\Delta \mathcal{O}$ satisfies
    \begin{align}\label{eq:prop_1_breakup}
    \Delta \mathcal{O} \leq O\bigg(\frac{\tau^{d +  p + 1}}{T^p}\bigg) + O\big(\gamma T^{d + 1}\big).
    \end{align}
    This error bound is minimized when $T =\tau \Theta(\gamma^{-1/(p+d+1)})$, choosing which yields that $\Delta \mathcal{O}$ satisfies Eq.~\eqref{eq:target_scaling} with $\alpha = p/(p + d+ 1)$ and that $t_\textnormal{sim} \leq O(\tau \gamma^{-1/(p+d+1)})$.}
\end{proposition}
\noindent Proposition \ref{prop:trotter} establishes that, for geometrically local observables, the error due to noise in the Trotterized circuit depend polynomially on the noise rate $\gamma$ and the target evolution time $\tau$, but is independent of the system size. Moreoever, in the presence of noise, there is an optimal number of Trotter steps $T$ which depend on the noise rate $\gamma$ but is independent of the system size $N$ at which the minimum observable error is achieved --- at very small values of $T$, while the error accumulation due to noise is small, the large Trotter error degrades the observable accuracy and at very large values of $T$, while the Trotter error is small, the accumulation of errors due to noise degrades the observable accuracy. As a consequence of the locality of the observables and using the improved Trotterization error bounds from lemma \ref{lemma:trotter_local_obs_main}, we are able to show that the optimal number of Trotter steps $T$ is system-size independent and consequently the total simulation error is also system-size independent.

We illustrate these findings numerically by analyzing Trotterization of a 1D nearest-neighbour Gaussian fermion model in Fig.~\ref{fig:trotter}. Figure \ref{fig:trotter}(a) shows the observable error, in the absence of noise, as a function of system size $N$ for different Trotter time-steps $T$ --- as predicted by lemma \ref{lemma:trotter_local_obs_main}, we see that the Trotter error for local observables is independent of the system size. In Fig.~\ref{fig:trotter}(b), we consider noise in the Trotterized evolution and show the dependence of the observable error on the number of Trotter steps---as predicted by Eq.~\eqref{eq:prop_1_breakup}, in the presence of noise, there is a noise-dependent optimal number of Trotter steps that minimize the observable error. Finally, in Fig.~\ref{fig:trotter}(c), we show the variation of this optimal observable error with the noise rate in the Trotterized circuit---consistent with proposition \ref{prop:trotter}, we find that the observable error scales as a positive power of the noise rate.

Another consequence of lemma \ref{lemma:trotter_local_obs_main} is that it predicts lower fault tolerance overhead when considering only local observables. Our next result holds in the more restricted setting of local stochastic noise models \cite{aliferis2005quantum, aharonov1997fault, gottesman2024surviving}, although we expect it to be generalizable to more complex noise models. For this noise model, we show that if the physical noise rate is less than the fault tolerance threshold then while using code concatenation with the Trotterized circuit, the number of levels of concatentation needed to obtain a target local observable precision depends only on the evolution time and not on the system size.
\begin{proposition}[Reduced fault tolerance overhead]
If the physical noise rate is less than fault tolerance threshold, then the number of levels $L$ of concatenation needed to simulate local observables to precision $\delta$ is given by
\[
L = \Theta\bigg[\log\bigg(\frac{p + d + 1}{p} \frac{\log({\tau^{d + 1}}/{\delta})}{\log(\alpha)}\bigg) \bigg],
\]
where $\alpha$ is the ratio of the fault tolerance threshold to the physical noise rate.
\end{proposition}

Next, we consider implementing a target Hamiltonian $H_\text{target}$ as the effective Floquet Hamiltonian corresponding to a time-dependent simulator Hamiltonian $H(t)$. More specifically, we pick a stroboscopic period $\uptau$ for the simulator Hamiltonian, which will then be a periodic function of time with period $\uptau$. The effective Floquet Hamiltonian $H_F$ implemented by the simulator is then a Hamiltonian that generates the same unitary for each stroboscopic period $T$ i.e., $\exp(-iH_F \uptau) = \mathcal{T}\exp(-i\int_0^\uptau H(s) ds)$. For sufficiently small $\uptau$, $H_F$ can be systematically expressed as a power-series expansion in $T$ via the Floquet-Magnus expansion \cite{abanin2017effective, blanes2009magnus}
\[
H_F = \sum_{p = 0}^\infty \uptau^p V_F^{(p)},
\]
where the first few terms of this expansion are given by
\begin{subequations}\label{eq:floquet_magnus_expansion}
\begin{align}
&V_F^{(0)} = \int_0^\uptau H(s) \frac{ds}{\uptau}, \\
&V_F^{(1)} = \int_{\leq 0 \leq s_2 \leq s_1 \leq \uptau}[H(s_1), H(s_2)]\frac{ds_1 ds_2}{2i \uptau^2},  \\
&V_F^{(2)} = -\int_{0\leq s_3\leq s_2 \leq s_1\leq \uptau} \bigg([H(s_1),[H(s_2), H(s_3)]]+\nonumber\\
&\qquad \qquad \qquad \qquad  [H(s_3),[H(s_2), H(s_1)]]\bigg)\frac{ds_1 ds_2 ds_3}{6\uptau^3}
\end{align}
\end{subequations}
To implement the target Hamiltonian $H_0$ on the simulator, we will design the simulator Hamiltonian $H(t)$ such that, for some $p$, the Floquet Hamiltonian $H_F^{(p)} = V_F^{(0)} + \uptau V_F^{(1)} + \dots \uptau^p V_F^{(p)}$ upto the $p^\text{th}$ order satisfies
\begin{align}\label{eq:floquet_target_relation}
    H_F^{(p)} = \uptau^p H_0 + \uptau^{p + 1}V,
\end{align}
i.e., $H_F^{(p)}$, to order $\uptau^p$, matches with $H_0$ with error between $H_F^{(p)}$ and $\uptau^p H_0$ being captured by $\uptau^{p + 1}V$, where $V$ is another Hamiltonian. As illustrated in Eq.~\eqref{eq:floquet_magnus_expansion}, $V_F^{(p)}$ is expressible in terms of commutators of $H(t)$ nested $p$ times. Since $H(t)$ is geometrically local, this implies that $V$ is also geometrically local but with locality possibly growing with $p$. Furthermore, Eq.~\eqref{eq:floquet_target_relation} implies that to simulate $e^{-iH_0 \tau}$, the simulator would then have to be evolved for time $t_\text{sim} = \tau / \uptau^p$.

We emphasize that to ensure that Eq.~\eqref{eq:floquet_target_relation} holds, in general we would need to scale different terms in $H(t)$ with different powers of $\uptau$. For instance, suppose the simulator implements a time-dependent 1D transverse-field Ising model $H(t) = f(t) \sum_i \sigma^x_i + g(t) \sum_i \sigma^z_i \sigma^z_{i + 1}$ and our target Hamiltonian is
$H_0 = h\sum_i \sigma^x_i + J_z\sum_i \sigma^z_i \sigma^z_{i + 1} + J_y \sum_i \sigma^y_i \sigma^y_{i + 1}$. This target Hamiltonian can be implemented via a second-order Floquet Magnus expansion (Eq.~\eqref{eq:floquet_target_relation} with $p = 2$) --- in particular, choosing $f(t) = h \uptau^2 +F_1 \cos(2\pi t / \uptau), g(t) = (J_x + J_y)\uptau^2 + G_1 \cos(4\pi t / \uptau)$
with $F_1 G_1^2 = 4\pi^2 J_y $, we obtain that $H_F^{(2)} = \uptau^2 H_0 + O(\uptau^3)$. Here, it is critical to scale the coefficients in $f(t)$ and $g(t)$ with powers of $\uptau$, otherwise in the limit of $\uptau\to 0$, the dynamics due to $H(t)$ would be identical to the dynamics due to $H_F^{(0)} = V_F^{(0)}$. With the correctly rescaled simulator Hamiltonian, evolving $H(t)$ to time $\tau/\uptau^2$ effectively simulates $H_0$ for time $\tau$.

For Hamiltonians implemented via the Floquet-Magnus expansion, we establish the following proposition for simulating local observables in the presence of noise.
\begin{proposition}[Noise-robustness of Floquet-Magnus expansion]\label{prop:fm_expansion}
    {When using the Floquet-Magnus expansion upto $p^\text{th}$ order and with noise at rate $\gamma > 0$, the observable error $\Delta \mathcal{O}$ satisfies
    \begin{align}\label{eq:final_bound_fm_unopt}
    \Delta \mathcal{O} \leq O(\tau^{d + 1} \uptau) + O\bigg(\tau^{d + 1}\frac{\gamma}{\uptau^{p(d + 1)}} \bigg).
    \end{align}
    This error bound is minimized on choosing $\uptau = \Theta(\gamma^{1/(p(d + 1) + 1)})$, which yields that $\Delta \mathcal{O}$ satisfies Eq.~\eqref{eq:target_scaling} with $\alpha = 1/(p(d + 1) + 1)$ and that $t_\textnormal{sim} \leq O(\tau \gamma^{-p/(p(d + 1) + 1)})$.}
\end{proposition}
\noindent The proof of this proposition builds upon Refs.~\cite{abanin2017effective, abanin2017rigorous} which rigorously analyzed heating time-scales in periodically driven many-body systems. Similar to proposition \ref{prop:trotter} for Trotterization, proposition \ref{prop:fm_expansion} shows that even while using Floquet-Magnus expansion to implement a target Hamiltonian, local observables incur an error that depends only on the noise rate $\gamma$ and the target evolution time $\tau$ and not on the system size $N$. Furthermore, there is an optimal stroboscopic period $\uptau$ dependent only on the noise rate $\gamma$ which, in the presence of noise, yields the minimum error in the measured observable. Physically, this arises due to the fact that in the limit of large $\uptau$, the Floquet Hamiltonian $H_F$ is only a crude approximation to the target Hamiltonian while in the limit of small $\uptau$, the required simulation time overhead is large thus causing an accumulation of errors due to noise. As consequence of the locality of the observable, we are able to show that the optimal stroboscopic period is, in-fact, system-size independent thus yielding a total simulation error that is also system-size independent.

These scalings are numerically illustrated in Fig.~\ref{fig:floquet} by analyzing a Gaussian fermion model. We consider a 1D chain of fermionic modes with annihilation operators $a_1, a_2 \dots a_N$ and a simulator Hamiltonian
\begin{subequations}\label{eq:floquet_magnus_demo}
\begin{align}
H(t) = h(t) \mu + ig(t) Q, 
\end{align}
where
\begin{align}
\mu = \sum_i a_i^\dagger a_i \text{ and } Q = \sum_i (a_i + a_i^\dagger)(a_{i + 1} + a_{i + 1}^\dagger).
\end{align}
\end{subequations}
We choose $h(t) = h_0 + h_1 \cos (2\pi t/\uptau) , g(t) = g_0 + g_1 \sin (2\pi t / \uptau)$. To the zeroth order, the simulator implements $H_0 = h_0 \mu + g_0 Q$ with $t_\text{sim} = \tau$. Furthermore, if $h_0 = g_0 = 0$, then to the first order, the simulator implements $H_0 = i h_1 g_1 [\mu, Q]$ wtih $t_\text{sim} = \tau / \uptau$. The target observable that we simulate is $O = \mu / N$. Figure \ref{fig:floquet}(a) shows the observable error, in the absence of noise, as a function of the system size $N$ --- consistent with Eq.~\eqref{eq:final_bound_fm_unopt}, we find that the observable error does not depend on the system size. Figure \ref{fig:floquet}(b) shows the observable error as a function of the stroboscopic period $\uptau$ in the presence of noise --- when using zeroth order Floquet-Magnus expansion, since there is no simulation-time overhead (i.e., $t_\text{sim}  =\tau$), we observe that the observable error decreases on decreasing $\uptau$ and saturates to a noise-rate dependent value as $\uptau \to 0$. When using first order Floquet-Magnus expansion, since $t_\text{sim} = \tau / \uptau$, there is an optimal $\uptau$ at which the observable error in the presence of noise is minimized. Finally, Fig.~\ref{fig:floquet}(c) shows the optimal observable error as a function of the noise-rate $\gamma$---as predicted by proposition \ref{prop:fm_expansion}, this error scales as $\gamma^\alpha$, where the exponent $\alpha$ depends on the order of the Floquet-Magnus expansion used to simulate the target Hamiltonian.

Finally, we consider a problem-to-simulator mapping, commonly proposed for quantum simulation of lattice gauge theories, that employs the Schriffer Wolf transformation to effectively implement the target Hamiltonian $H_0$ as a Hamiltonian $M$ projected within the low-energy eigen-spaces of a geometrically local Gauge Hamiltonian $P$. More specifically, on the simulator we will implement a Hamiltonian $H$ given by
\[
H = \uptau M + P,
\]
where $M$ is a geometrically local Hamiltonian and $\uptau$ is a small parameter that controls the relative strength of the $P$ and $M$. We will assume that $P = \sum_\alpha P_\alpha$ is a sum of commuting orthogonal projectors i.e., $P_\alpha =  P_\alpha^\dagger = P_\alpha^2, [P_\alpha, P_\beta] = 0$. By performing a Schriffer Wolf transformation upto order $p$, we can obtain an effective geometrically local Hamiltonian $\tilde{M}^{(p)}$ and a anti-Hermitian matrix $\Omega$ such that, at a fixed system size $N$,
\[
\uptau M + P = e^{\Omega}(\uptau \tilde{M}^{(p)} + P)e^{-\Omega} + O(\uptau^{p + 1}),
\]
and $\Omega \sim O(\uptau)$. As with the case of Floquet-Magnus expansion, we will assume that $M$ and $P$ are designed such that
\begin{align}
    \tilde{M}^{(p)} = \uptau^p H_0 + \uptau^{p + 1}V,
\end{align}
where $V$ is a geometrically local Hamiltonian. For the simulator to reproduce the dynamics of $H_0$ for time $\tau$, it would then have to be run for $t_\text{sim} = \tau / \uptau^{p+ 1}$. We now establish the following proposition for local observable which parallels proposition \ref{prop:fm_expansion} for the Floquet-Magnus expansion. The proof of this proposition, provided in the supplement, closely parallels the proof of proposition 3, and builds upon the tools developed in Ref.~\cite{abanin2017rigorous}.
\begin{proposition}
    {Suppose $O$ is a local observable such that $[O, P] = 0$. Using perturbative expansion upto $p^\textnormal{th}$ order and with noise at rate $\gamma > 0$, the observable error $\Delta \mathcal{O}$ satisfies
    \[
    \Delta \mathcal{O} \leq O(\tau^{d + 1} \uptau) + O\bigg(\tau^{d + 1}\frac{\gamma}{\uptau^{p(d + 1) + 1}} \bigg).
    \]
    This error bound is minimized on choosing $\uptau = \Theta(\gamma^{1/(p(d + 1) + 2)})$, which yields that $\Delta \mathcal{O}$ satisfies Eq.~\eqref{eq:target_scaling} with $\alpha = 1/(p(d+ 1) + 2)$ and that $t_\textnormal{sim} \leq O(\tau \gamma^{-(p + 1)/(p(d + 1)+2)})$.}
\end{proposition}

\begin{figure*}
    \centering
    \includegraphics[width=1.0\linewidth]{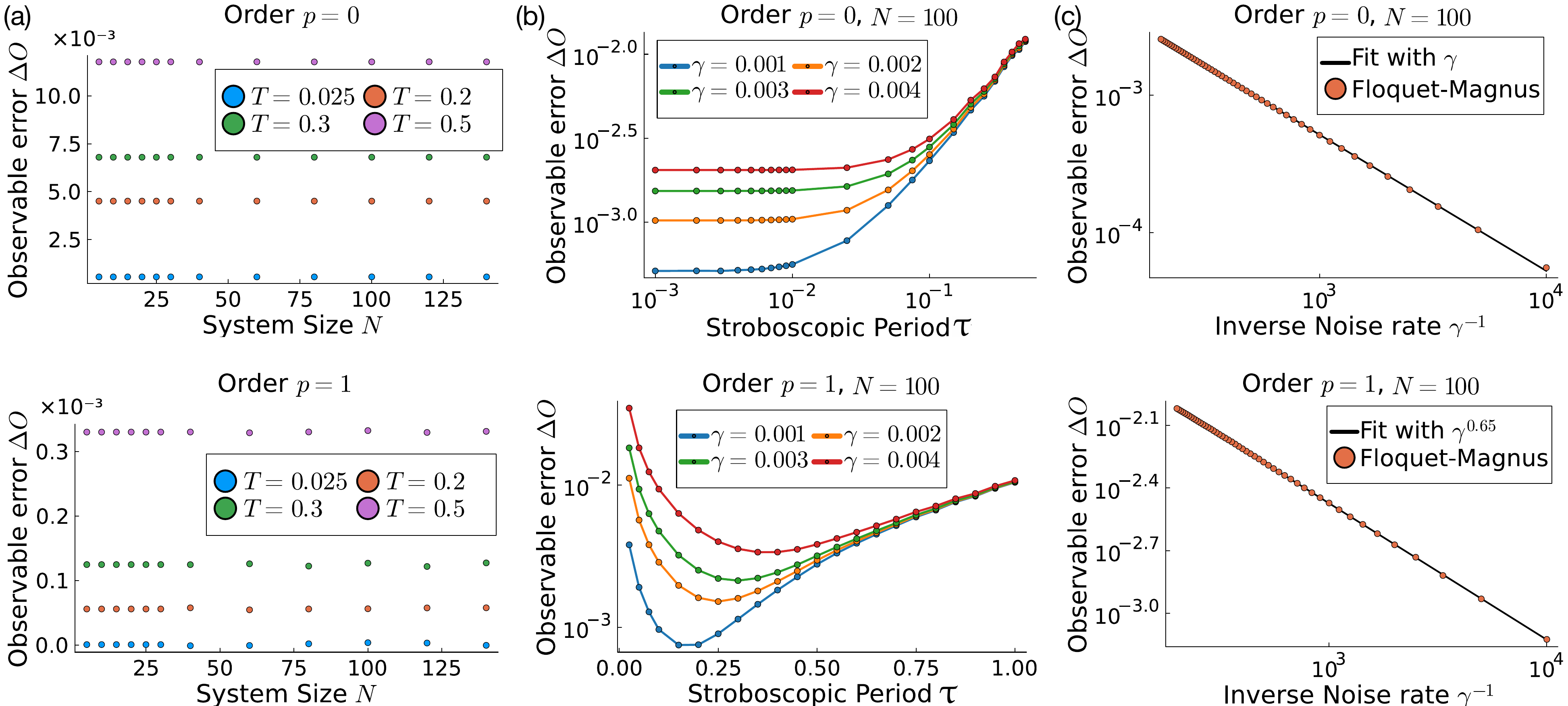}
    \caption{\textbf{Robustness of Floquet-Magnus expansion}. Numerical analysis of the model in Eq.~\eqref{eq:floquet_magnus_demo}. The fermionic modes are initially in the state $\ket{0}^{\otimes N}$ and the observable $O = \sum_i a_i^\dagger a_i / N$ is measured at time $\tau$. For $p = 0$, we choose $h_0 = g_0 = 1, h_1 = g_1 = 1/2, \tau = 1.5$ and for $p = 1$, we choose  $h_0 = g_0 = 0, h_1 = g_1 = 1, \tau = 5$. (a) Observable error, in the absence of noise, as a function of system size for different number stroboscopic periods $\uptau$ for $p = 0, 1$. We observe that $\Delta \mathcal{O}$ is independent of the system-size. (b) Observable error in the presence single-fermion depolarizing noise in the simulator, with $\gamma$ being the noise rate, as a function the stroboscopic period $\uptau$. In the presence of noise, there is an optimal stroboscopic period $\uptau$ which depends $\gamma$ and minimizes the observable error. (c) The observable error at this optimal $\uptau$ as a function of the noise rate $\gamma$. }
    \label{fig:floquet}
\end{figure*}

{\emph{Discussion}. Propositions 1, 3 and 4 established that the error in local observables $\Delta \mathcal{O} \leq O(\gamma^\alpha \tau^{d + 1})$ with $\alpha$ depending on the dimensionality of the lattice $d$ and the order $p$ of the problem-to-simulator mapping. For Trotterization, $\alpha \to 1$ as $p$ is increased---this is consistent with the previous results that showed that higher-order Trotter formulae result in asymptotically lower gate counts and thus lower simulation-time overheads \cite{childs2019nearly}. Nevertheless, higher order Trotterization would also be accompanied by a constant-factor overhead which would grow with $p$ and, at any fixed noise-rate $\gamma$, might make it unfavorable to use very high-order Trotter formulae. On the other hand, for Floquet-Magnus and perturbative expansions, $\alpha \to 0$ as $p \to \infty$---this is a consequence of these mappings effectively requiring higher-order interactions to correctly mimic the physics of the target model, and thus incurring a longer simulation time overhead.

Furthermore, these robustness results have an important implication on quantum advantage as formulated in Refs.~\cite{trivedi2024quantum, kashyap2025accuracy}. Since $\Delta \mathcal{O} \leq O(\gamma^\alpha \tau^{d + 1})$, to maintain a fixed $O(1)$ precision in the target observable, we can reliably simulate $\tau \leq O(\gamma^{-\alpha/(d + 1)})$. Furthermore, for all of the mappings that we consider, the simulator run-time $t_\text{sim} \leq O(\text{poly}(\gamma^{-1}))$. For classical algorithms to simulate local observables for such evolution times, we will thus need a run-time $\exp(O(\tau^d)) \sim \exp(\text{poly}(\gamma^{-1}))$. Thus, as the noise-rate $\gamma$ is reduced, either by improvement in the physical hardware of the simulator or via quantum error correction, classical algorithms are expected to require exponentially more time to reach the precision that can be reached by the quantum simulator.

We remark that this conclusion strongly depends on the dependence of the observable error $\Delta \mathcal{O}$ on $\gamma$ and $\tau$---even if $\Delta O$ is independent of the system size (i.e., it is noise robust), it could be that classical algorithms are not exponentially slower than the quantum simulator as $\gamma \to 0$. For instance, if for a problem-to-simulator mapping, the observables perturb by $\Delta \mathcal{O} \leq O(\gamma \exp(O(\tau))$ or $O(\log^{-1}(\gamma^{-1}) \text{poly}(\tau))$, then to maintain a fixed $O(1)$ precision we can only reliably go upto $\tau \leq O(\text{polylog}(\gamma^{-1}))$ which can be simulated on a classical computer in time superpolynomial in $\gamma^{-1}$. For such a mapping, the classical algorithm would thus not be exponentially slower than quantum simulators as the noise-rate $\gamma$ is reduced.}

\emph{Conclusion}. In conclusion, this paper provides theoretical evidence that even while using complex problem-to-simulator mappings on noisy quantum simulators, we could expect local observables to be reliably simulated even if the full quantum state is not. The worst-case accumulation of errors with system size, which is expected when looking at full state fidelities, is avoided when considering local observables in geometrically local models due to their light cone. However, our numerical simulations of Gaussian fermion model, which are likely not the worst-case when it comes to error accumulation, indicate that to obtain local observable errors $\sim 1\%$ we still need noise rates $\sim 0.1\%$. While these noise rates are not scalably available at the level of the physical qubits, with recent progress in experimental quantum error correction \cite{bluvstein2024logical, google2025quantum}, we expect that such error rates could be achievable in the logical qubits with one or a few rounds of error correction. The results in this paper would then provide theoretical evidence that local observables in such ``pre-fault tolerant" experiments could be trusted.

Our work opens up several questions---extending the results of this paper to both models with longer-range interactions and noise-models that are either space or time-correlated would immediately improve their applicability to physical experiments. {Furthermore, it is reasonable to expect that noise process, on average, incur an observable error much smaller than the worst case \cite{cai2023stochastic} and we can aim at understanding the average-case impact of noise on problem-to-simulator mappings.}  Finally, while in this paper we provide upper-bounds on the observable errors and focus on their asymptotic scalings with system size, evolution time and noise-rate, it would be useful to develop classical numerical tools which can efficiently compute numerically useful upper bounds on observable errors. Some progress in this direction has recently been made in this direction by combining tensor network methods with lower-bounding techniques from convex optimization theory \cite{mishra2024classically} in the context of Hamiltonian ground state problems. Finally, developing similar stability results for the simulation of more complex many-body systems, such as open quantum systems or quantum field theories, could also improve our understanding of the utility of pre-fault tolerant quantum hardware.

\begin{acknowledgements}
    We thank Andrew Childs for discussions that inspired the project. We thank the Institute of Advanced Studies, Princeton for their hospitality where part of this work was done. R.T. acknowledges funding from the European Union's Horizon Europe research and innovation program under grant agreement number 101221560 (ToNQS). JIC acknowledges funding from the Federal Ministry of Education and Research Germany (BMBF) via the project FermiQP (13N15889) and THEQUCO, which is part of the Munich Quantum Valley, and supported by the Bavarian state government with funds from the Hightech Agenda Bayern Plus.
\end{acknowledgements}
\bibliography{references}

\begin{thebibliography}{59}%
\makeatletter
\providecommand \@ifxundefined [1]{%
 \@ifx{#1\undefined}
}%
\providecommand \@ifnum [1]{%
 \ifnum #1\expandafter \@firstoftwo
 \else \expandafter \@secondoftwo
 \fi
}%
\providecommand \@ifx [1]{%
 \ifx #1\expandafter \@firstoftwo
 \else \expandafter \@secondoftwo
 \fi
}%
\providecommand \natexlab [1]{#1}%
\providecommand \enquote  [1]{``#1''}%
\providecommand \bibnamefont  [1]{#1}%
\providecommand \bibfnamefont [1]{#1}%
\providecommand \citenamefont [1]{#1}%
\providecommand \href@noop [0]{\@secondoftwo}%
\providecommand \href [0]{\begingroup \@sanitize@url \@href}%
\providecommand \@href[1]{\@@startlink{#1}\@@href}%
\providecommand \@@href[1]{\endgroup#1\@@endlink}%
\providecommand \@sanitize@url [0]{\catcode `\\12\catcode `\$12\catcode
  `\&12\catcode `\#12\catcode `\^12\catcode `\_12\catcode `\%12\relax}%
\providecommand \@@startlink[1]{}%
\providecommand \@@endlink[0]{}%
\providecommand \url  [0]{\begingroup\@sanitize@url \@url }%
\providecommand \@url [1]{\endgroup\@href {#1}{\urlprefix }}%
\providecommand \urlprefix  [0]{URL }%
\providecommand \Eprint [0]{\href }%
\providecommand \doibase [0]{https://doi.org/}%
\providecommand \selectlanguage [0]{\@gobble}%
\providecommand \bibinfo  [0]{\@secondoftwo}%
\providecommand \bibfield  [0]{\@secondoftwo}%
\providecommand \translation [1]{[#1]}%
\providecommand \BibitemOpen [0]{}%
\providecommand \bibitemStop [0]{}%
\providecommand \bibitemNoStop [0]{.\EOS\space}%
\providecommand \EOS [0]{\spacefactor3000\relax}%
\providecommand \BibitemShut  [1]{\csname bibitem#1\endcsname}%
\let\auto@bib@innerbib\@empty
%</preamble>
\bibitem [{\citenamefont {Daley}\ \emph {et~al.}(2022)\citenamefont {Daley},
  \citenamefont {Bloch}, \citenamefont {Kokail}, \citenamefont {Flannigan},
  \citenamefont {Pearson}, \citenamefont {Troyer},\ and\ \citenamefont
  {Zoller}}]{daley2022practical}%
  \BibitemOpen
  \bibfield  {author} {\bibinfo {author} {\bibfnamefont {A.~J.}\ \bibnamefont
  {Daley}}, \bibinfo {author} {\bibfnamefont {I.}~\bibnamefont {Bloch}},
  \bibinfo {author} {\bibfnamefont {C.}~\bibnamefont {Kokail}}, \bibinfo
  {author} {\bibfnamefont {S.}~\bibnamefont {Flannigan}}, \bibinfo {author}
  {\bibfnamefont {N.}~\bibnamefont {Pearson}}, \bibinfo {author} {\bibfnamefont
  {M.}~\bibnamefont {Troyer}},\ and\ \bibinfo {author} {\bibfnamefont
  {P.}~\bibnamefont {Zoller}},\ }\bibfield  {title} {\bibinfo {title}
  {Practical quantum advantage in quantum simulation},\ }\href@noop {}
  {\bibfield  {journal} {\bibinfo  {journal} {Nature}\ }\textbf {\bibinfo
  {volume} {607}},\ \bibinfo {pages} {667} (\bibinfo {year}
  {2022})}\BibitemShut {NoStop}%
\bibitem [{\citenamefont {Bauer}\ \emph {et~al.}(2023)\citenamefont {Bauer},
  \citenamefont {Davoudi}, \citenamefont {Balantekin}, \citenamefont
  {Bhattacharya}, \citenamefont {Carena}, \citenamefont {De~Jong},
  \citenamefont {Draper}, \citenamefont {El-Khadra}, \citenamefont {Gemelke},
  \citenamefont {Hanada} \emph {et~al.}}]{bauer2023quantum}%
  \BibitemOpen
  \bibfield  {author} {\bibinfo {author} {\bibfnamefont {C.~W.}\ \bibnamefont
  {Bauer}}, \bibinfo {author} {\bibfnamefont {Z.}~\bibnamefont {Davoudi}},
  \bibinfo {author} {\bibfnamefont {A.~B.}\ \bibnamefont {Balantekin}},
  \bibinfo {author} {\bibfnamefont {T.}~\bibnamefont {Bhattacharya}}, \bibinfo
  {author} {\bibfnamefont {M.}~\bibnamefont {Carena}}, \bibinfo {author}
  {\bibfnamefont {W.~A.}\ \bibnamefont {De~Jong}}, \bibinfo {author}
  {\bibfnamefont {P.}~\bibnamefont {Draper}}, \bibinfo {author} {\bibfnamefont
  {A.}~\bibnamefont {El-Khadra}}, \bibinfo {author} {\bibfnamefont
  {N.}~\bibnamefont {Gemelke}}, \bibinfo {author} {\bibfnamefont
  {M.}~\bibnamefont {Hanada}}, \emph {et~al.},\ }\bibfield  {title} {\bibinfo
  {title} {Quantum simulation for high-energy physics},\ }\href@noop {}
  {\bibfield  {journal} {\bibinfo  {journal} {PRX quantum}\ }\textbf {\bibinfo
  {volume} {4}},\ \bibinfo {pages} {027001} (\bibinfo {year}
  {2023})}\BibitemShut {NoStop}%
\bibitem [{\citenamefont {Cao}\ \emph {et~al.}(2019)\citenamefont {Cao},
  \citenamefont {Romero}, \citenamefont {Olson}, \citenamefont {Degroote},
  \citenamefont {Johnson}, \citenamefont {Kieferov{\'a}}, \citenamefont
  {Kivlichan}, \citenamefont {Menke}, \citenamefont {Peropadre}, \citenamefont
  {Sawaya} \emph {et~al.}}]{cao2019quantum}%
  \BibitemOpen
  \bibfield  {author} {\bibinfo {author} {\bibfnamefont {Y.}~\bibnamefont
  {Cao}}, \bibinfo {author} {\bibfnamefont {J.}~\bibnamefont {Romero}},
  \bibinfo {author} {\bibfnamefont {J.~P.}\ \bibnamefont {Olson}}, \bibinfo
  {author} {\bibfnamefont {M.}~\bibnamefont {Degroote}}, \bibinfo {author}
  {\bibfnamefont {P.~D.}\ \bibnamefont {Johnson}}, \bibinfo {author}
  {\bibfnamefont {M.}~\bibnamefont {Kieferov{\'a}}}, \bibinfo {author}
  {\bibfnamefont {I.~D.}\ \bibnamefont {Kivlichan}}, \bibinfo {author}
  {\bibfnamefont {T.}~\bibnamefont {Menke}}, \bibinfo {author} {\bibfnamefont
  {B.}~\bibnamefont {Peropadre}}, \bibinfo {author} {\bibfnamefont {N.~P.}\
  \bibnamefont {Sawaya}}, \emph {et~al.},\ }\bibfield  {title} {\bibinfo
  {title} {Quantum chemistry in the age of quantum computing},\ }\href@noop {}
  {\bibfield  {journal} {\bibinfo  {journal} {Chemical reviews}\ }\textbf
  {\bibinfo {volume} {119}},\ \bibinfo {pages} {10856} (\bibinfo {year}
  {2019})}\BibitemShut {NoStop}%
\bibitem [{\citenamefont {Altman}\ \emph {et~al.}(2021)\citenamefont {Altman},
  \citenamefont {Brown}, \citenamefont {Carleo}, \citenamefont {Carr},
  \citenamefont {Demler}, \citenamefont {Chin}, \citenamefont {DeMarco},
  \citenamefont {Economou}, \citenamefont {Eriksson}, \citenamefont {Fu} \emph
  {et~al.}}]{altman2021quantum}%
  \BibitemOpen
  \bibfield  {author} {\bibinfo {author} {\bibfnamefont {E.}~\bibnamefont
  {Altman}}, \bibinfo {author} {\bibfnamefont {K.~R.}\ \bibnamefont {Brown}},
  \bibinfo {author} {\bibfnamefont {G.}~\bibnamefont {Carleo}}, \bibinfo
  {author} {\bibfnamefont {L.~D.}\ \bibnamefont {Carr}}, \bibinfo {author}
  {\bibfnamefont {E.}~\bibnamefont {Demler}}, \bibinfo {author} {\bibfnamefont
  {C.}~\bibnamefont {Chin}}, \bibinfo {author} {\bibfnamefont {B.}~\bibnamefont
  {DeMarco}}, \bibinfo {author} {\bibfnamefont {S.~E.}\ \bibnamefont
  {Economou}}, \bibinfo {author} {\bibfnamefont {M.~A.}\ \bibnamefont
  {Eriksson}}, \bibinfo {author} {\bibfnamefont {K.-M.~C.}\ \bibnamefont {Fu}},
  \emph {et~al.},\ }\bibfield  {title} {\bibinfo {title} {Quantum simulators:
  Architectures and opportunities},\ }\href@noop {} {\bibfield  {journal}
  {\bibinfo  {journal} {PRX quantum}\ }\textbf {\bibinfo {volume} {2}},\
  \bibinfo {pages} {017003} (\bibinfo {year} {2021})}\BibitemShut {NoStop}%
\bibitem [{\citenamefont {Arute}\ \emph {et~al.}(2019)\citenamefont {Arute},
  \citenamefont {Arya}, \citenamefont {Babbush}, \citenamefont {Bacon},
  \citenamefont {Bardin}, \citenamefont {Barends}, \citenamefont {Biswas},
  \citenamefont {Boixo}, \citenamefont {Brandao}, \citenamefont {Buell} \emph
  {et~al.}}]{arute2019quantum}%
  \BibitemOpen
  \bibfield  {author} {\bibinfo {author} {\bibfnamefont {F.}~\bibnamefont
  {Arute}}, \bibinfo {author} {\bibfnamefont {K.}~\bibnamefont {Arya}},
  \bibinfo {author} {\bibfnamefont {R.}~\bibnamefont {Babbush}}, \bibinfo
  {author} {\bibfnamefont {D.}~\bibnamefont {Bacon}}, \bibinfo {author}
  {\bibfnamefont {J.~C.}\ \bibnamefont {Bardin}}, \bibinfo {author}
  {\bibfnamefont {R.}~\bibnamefont {Barends}}, \bibinfo {author} {\bibfnamefont
  {R.}~\bibnamefont {Biswas}}, \bibinfo {author} {\bibfnamefont
  {S.}~\bibnamefont {Boixo}}, \bibinfo {author} {\bibfnamefont {F.~G.}\
  \bibnamefont {Brandao}}, \bibinfo {author} {\bibfnamefont {D.~A.}\
  \bibnamefont {Buell}}, \emph {et~al.},\ }\bibfield  {title} {\bibinfo {title}
  {Quantum supremacy using a programmable superconducting processor},\
  }\href@noop {} {\bibfield  {journal} {\bibinfo  {journal} {Nature}\ }\textbf
  {\bibinfo {volume} {574}},\ \bibinfo {pages} {505} (\bibinfo {year}
  {2019})}\BibitemShut {NoStop}%
\bibitem [{\citenamefont {Wienand}\ \emph {et~al.}(2024)\citenamefont
  {Wienand}, \citenamefont {Karch}, \citenamefont {Impertro}, \citenamefont
  {Schweizer}, \citenamefont {McCulloch}, \citenamefont {Vasseur},
  \citenamefont {Gopalakrishnan}, \citenamefont {Aidelsburger},\ and\
  \citenamefont {Bloch}}]{wienand2024emergence}%
  \BibitemOpen
  \bibfield  {author} {\bibinfo {author} {\bibfnamefont {J.~F.}\ \bibnamefont
  {Wienand}}, \bibinfo {author} {\bibfnamefont {S.}~\bibnamefont {Karch}},
  \bibinfo {author} {\bibfnamefont {A.}~\bibnamefont {Impertro}}, \bibinfo
  {author} {\bibfnamefont {C.}~\bibnamefont {Schweizer}}, \bibinfo {author}
  {\bibfnamefont {E.}~\bibnamefont {McCulloch}}, \bibinfo {author}
  {\bibfnamefont {R.}~\bibnamefont {Vasseur}}, \bibinfo {author} {\bibfnamefont
  {S.}~\bibnamefont {Gopalakrishnan}}, \bibinfo {author} {\bibfnamefont
  {M.}~\bibnamefont {Aidelsburger}},\ and\ \bibinfo {author} {\bibfnamefont
  {I.}~\bibnamefont {Bloch}},\ }\bibfield  {title} {\bibinfo {title} {Emergence
  of fluctuating hydrodynamics in chaotic quantum systems},\ }\href@noop {}
  {\bibfield  {journal} {\bibinfo  {journal} {Nature Physics}\ }\textbf
  {\bibinfo {volume} {20}},\ \bibinfo {pages} {1732} (\bibinfo {year}
  {2024})}\BibitemShut {NoStop}%
\bibitem [{\citenamefont {Ebadi}\ \emph {et~al.}(2021)\citenamefont {Ebadi},
  \citenamefont {Wang}, \citenamefont {Levine}, \citenamefont {Keesling},
  \citenamefont {Semeghini}, \citenamefont {Omran}, \citenamefont {Bluvstein},
  \citenamefont {Samajdar}, \citenamefont {Pichler}, \citenamefont {Ho} \emph
  {et~al.}}]{ebadi2021quantum}%
  \BibitemOpen
  \bibfield  {author} {\bibinfo {author} {\bibfnamefont {S.}~\bibnamefont
  {Ebadi}}, \bibinfo {author} {\bibfnamefont {T.~T.}\ \bibnamefont {Wang}},
  \bibinfo {author} {\bibfnamefont {H.}~\bibnamefont {Levine}}, \bibinfo
  {author} {\bibfnamefont {A.}~\bibnamefont {Keesling}}, \bibinfo {author}
  {\bibfnamefont {G.}~\bibnamefont {Semeghini}}, \bibinfo {author}
  {\bibfnamefont {A.}~\bibnamefont {Omran}}, \bibinfo {author} {\bibfnamefont
  {D.}~\bibnamefont {Bluvstein}}, \bibinfo {author} {\bibfnamefont
  {R.}~\bibnamefont {Samajdar}}, \bibinfo {author} {\bibfnamefont
  {H.}~\bibnamefont {Pichler}}, \bibinfo {author} {\bibfnamefont {W.~W.}\
  \bibnamefont {Ho}}, \emph {et~al.},\ }\bibfield  {title} {\bibinfo {title}
  {Quantum phases of matter on a 256-atom programmable quantum simulator},\
  }\href@noop {} {\bibfield  {journal} {\bibinfo  {journal} {Nature}\ }\textbf
  {\bibinfo {volume} {595}},\ \bibinfo {pages} {227} (\bibinfo {year}
  {2021})}\BibitemShut {NoStop}%
\bibitem [{\citenamefont {Bourgund}\ \emph {et~al.}(2025)\citenamefont
  {Bourgund}, \citenamefont {Chalopin}, \citenamefont {Bojovi{\'c}},
  \citenamefont {Schl{\"o}mer}, \citenamefont {Wang}, \citenamefont {Franz},
  \citenamefont {Hirthe}, \citenamefont {Bohrdt}, \citenamefont {Grusdt},
  \citenamefont {Bloch} \emph {et~al.}}]{bourgund2025formation}%
  \BibitemOpen
  \bibfield  {author} {\bibinfo {author} {\bibfnamefont {D.}~\bibnamefont
  {Bourgund}}, \bibinfo {author} {\bibfnamefont {T.}~\bibnamefont {Chalopin}},
  \bibinfo {author} {\bibfnamefont {P.}~\bibnamefont {Bojovi{\'c}}}, \bibinfo
  {author} {\bibfnamefont {H.}~\bibnamefont {Schl{\"o}mer}}, \bibinfo {author}
  {\bibfnamefont {S.}~\bibnamefont {Wang}}, \bibinfo {author} {\bibfnamefont
  {T.}~\bibnamefont {Franz}}, \bibinfo {author} {\bibfnamefont
  {S.}~\bibnamefont {Hirthe}}, \bibinfo {author} {\bibfnamefont
  {A.}~\bibnamefont {Bohrdt}}, \bibinfo {author} {\bibfnamefont
  {F.}~\bibnamefont {Grusdt}}, \bibinfo {author} {\bibfnamefont
  {I.}~\bibnamefont {Bloch}}, \emph {et~al.},\ }\bibfield  {title} {\bibinfo
  {title} {Formation of individual stripes in a mixed-dimensional cold-atom
  fermi--hubbard system},\ }\href@noop {} {\bibfield  {journal} {\bibinfo
  {journal} {Nature}\ }\textbf {\bibinfo {volume} {637}},\ \bibinfo {pages}
  {57} (\bibinfo {year} {2025})}\BibitemShut {NoStop}%
\bibitem [{\citenamefont {Bluvstein}\ \emph {et~al.}(2022)\citenamefont
  {Bluvstein}, \citenamefont {Levine}, \citenamefont {Semeghini}, \citenamefont
  {Wang}, \citenamefont {Ebadi}, \citenamefont {Kalinowski}, \citenamefont
  {Keesling}, \citenamefont {Maskara}, \citenamefont {Pichler}, \citenamefont
  {Greiner} \emph {et~al.}}]{bluvstein2022quantum}%
  \BibitemOpen
  \bibfield  {author} {\bibinfo {author} {\bibfnamefont {D.}~\bibnamefont
  {Bluvstein}}, \bibinfo {author} {\bibfnamefont {H.}~\bibnamefont {Levine}},
  \bibinfo {author} {\bibfnamefont {G.}~\bibnamefont {Semeghini}}, \bibinfo
  {author} {\bibfnamefont {T.~T.}\ \bibnamefont {Wang}}, \bibinfo {author}
  {\bibfnamefont {S.}~\bibnamefont {Ebadi}}, \bibinfo {author} {\bibfnamefont
  {M.}~\bibnamefont {Kalinowski}}, \bibinfo {author} {\bibfnamefont
  {A.}~\bibnamefont {Keesling}}, \bibinfo {author} {\bibfnamefont
  {N.}~\bibnamefont {Maskara}}, \bibinfo {author} {\bibfnamefont
  {H.}~\bibnamefont {Pichler}}, \bibinfo {author} {\bibfnamefont
  {M.}~\bibnamefont {Greiner}}, \emph {et~al.},\ }\bibfield  {title} {\bibinfo
  {title} {A quantum processor based on coherent transport of entangled atom
  arrays},\ }\href@noop {} {\bibfield  {journal} {\bibinfo  {journal} {Nature}\
  }\textbf {\bibinfo {volume} {604}},\ \bibinfo {pages} {451} (\bibinfo {year}
  {2022})}\BibitemShut {NoStop}%
\bibitem [{\citenamefont {Bluvstein}\ \emph {et~al.}(2024)\citenamefont
  {Bluvstein}, \citenamefont {Evered}, \citenamefont {Geim}, \citenamefont
  {Li}, \citenamefont {Zhou}, \citenamefont {Manovitz}, \citenamefont {Ebadi},
  \citenamefont {Cain}, \citenamefont {Kalinowski}, \citenamefont {Hangleiter}
  \emph {et~al.}}]{bluvstein2024logical}%
  \BibitemOpen
  \bibfield  {author} {\bibinfo {author} {\bibfnamefont {D.}~\bibnamefont
  {Bluvstein}}, \bibinfo {author} {\bibfnamefont {S.~J.}\ \bibnamefont
  {Evered}}, \bibinfo {author} {\bibfnamefont {A.~A.}\ \bibnamefont {Geim}},
  \bibinfo {author} {\bibfnamefont {S.~H.}\ \bibnamefont {Li}}, \bibinfo
  {author} {\bibfnamefont {H.}~\bibnamefont {Zhou}}, \bibinfo {author}
  {\bibfnamefont {T.}~\bibnamefont {Manovitz}}, \bibinfo {author}
  {\bibfnamefont {S.}~\bibnamefont {Ebadi}}, \bibinfo {author} {\bibfnamefont
  {M.}~\bibnamefont {Cain}}, \bibinfo {author} {\bibfnamefont {M.}~\bibnamefont
  {Kalinowski}}, \bibinfo {author} {\bibfnamefont {D.}~\bibnamefont
  {Hangleiter}}, \emph {et~al.},\ }\bibfield  {title} {\bibinfo {title}
  {Logical quantum processor based on reconfigurable atom arrays},\ }\href@noop
  {} {\bibfield  {journal} {\bibinfo  {journal} {Nature}\ }\textbf {\bibinfo
  {volume} {626}},\ \bibinfo {pages} {58} (\bibinfo {year} {2024})}\BibitemShut
  {NoStop}%
\bibitem [{goo(2023)}]{google2023suppressing}%
  \BibitemOpen
  \bibfield  {title} {\bibinfo {title} {Suppressing quantum errors by scaling a
  surface code logical qubit},\ }\href@noop {} {\bibfield  {journal} {\bibinfo
  {journal} {Nature}\ }\textbf {\bibinfo {volume} {614}},\ \bibinfo {pages}
  {676} (\bibinfo {year} {2023})}\BibitemShut {NoStop}%
\bibitem [{\citenamefont {Beverland}\ \emph {et~al.}(2022)\citenamefont
  {Beverland}, \citenamefont {Murali}, \citenamefont {Troyer}, \citenamefont
  {Svore}, \citenamefont {Hoefler}, \citenamefont {Kliuchnikov}, \citenamefont
  {Low}, \citenamefont {Soeken}, \citenamefont {Sundaram},\ and\ \citenamefont
  {Vaschillo}}]{beverland2022assessing}%
  \BibitemOpen
  \bibfield  {author} {\bibinfo {author} {\bibfnamefont {M.~E.}\ \bibnamefont
  {Beverland}}, \bibinfo {author} {\bibfnamefont {P.}~\bibnamefont {Murali}},
  \bibinfo {author} {\bibfnamefont {M.}~\bibnamefont {Troyer}}, \bibinfo
  {author} {\bibfnamefont {K.~M.}\ \bibnamefont {Svore}}, \bibinfo {author}
  {\bibfnamefont {T.}~\bibnamefont {Hoefler}}, \bibinfo {author} {\bibfnamefont
  {V.}~\bibnamefont {Kliuchnikov}}, \bibinfo {author} {\bibfnamefont {G.~H.}\
  \bibnamefont {Low}}, \bibinfo {author} {\bibfnamefont {M.}~\bibnamefont
  {Soeken}}, \bibinfo {author} {\bibfnamefont {A.}~\bibnamefont {Sundaram}},\
  and\ \bibinfo {author} {\bibfnamefont {A.}~\bibnamefont {Vaschillo}},\
  }\bibfield  {title} {\bibinfo {title} {Assessing requirements to scale to
  practical quantum advantage},\ }\href@noop {} {\bibfield  {journal} {\bibinfo
   {journal} {arXiv preprint arXiv:2211.07629}\ } (\bibinfo {year}
  {2022})}\BibitemShut {NoStop}%
\bibitem [{\citenamefont {Granet}\ and\ \citenamefont
  {Dreyer}(2025)}]{granet2025dilution}%
  \BibitemOpen
  \bibfield  {author} {\bibinfo {author} {\bibfnamefont {E.}~\bibnamefont
  {Granet}}\ and\ \bibinfo {author} {\bibfnamefont {H.}~\bibnamefont
  {Dreyer}},\ }\bibfield  {title} {\bibinfo {title} {Dilution of error in
  digital hamiltonian simulation},\ }\href@noop {} {\bibfield  {journal}
  {\bibinfo  {journal} {PRX Quantum}\ }\textbf {\bibinfo {volume} {6}},\
  \bibinfo {pages} {010333} (\bibinfo {year} {2025})}\BibitemShut {NoStop}%
\bibitem [{\citenamefont {Cubitt}\ \emph {et~al.}(2015)\citenamefont {Cubitt},
  \citenamefont {Lucia}, \citenamefont {Michalakis},\ and\ \citenamefont
  {Perez-Garcia}}]{cubitt2015stability}%
  \BibitemOpen
  \bibfield  {author} {\bibinfo {author} {\bibfnamefont {T.~S.}\ \bibnamefont
  {Cubitt}}, \bibinfo {author} {\bibfnamefont {A.}~\bibnamefont {Lucia}},
  \bibinfo {author} {\bibfnamefont {S.}~\bibnamefont {Michalakis}},\ and\
  \bibinfo {author} {\bibfnamefont {D.}~\bibnamefont {Perez-Garcia}},\
  }\bibfield  {title} {\bibinfo {title} {Stability of local quantum dissipative
  systems},\ }\href@noop {} {\bibfield  {journal} {\bibinfo  {journal}
  {Communications in Mathematical Physics}\ }\textbf {\bibinfo {volume}
  {337}},\ \bibinfo {pages} {1275} (\bibinfo {year} {2015})}\BibitemShut
  {NoStop}%
\bibitem [{\citenamefont {Trivedi}\ \emph {et~al.}(2024)\citenamefont
  {Trivedi}, \citenamefont {Franco~Rubio},\ and\ \citenamefont
  {Cirac}}]{trivedi2024quantum}%
  \BibitemOpen
  \bibfield  {author} {\bibinfo {author} {\bibfnamefont {R.}~\bibnamefont
  {Trivedi}}, \bibinfo {author} {\bibfnamefont {A.}~\bibnamefont
  {Franco~Rubio}},\ and\ \bibinfo {author} {\bibfnamefont {J.~I.}\ \bibnamefont
  {Cirac}},\ }\bibfield  {title} {\bibinfo {title} {Quantum advantage and
  stability to errors in analogue quantum simulators},\ }\href@noop {}
  {\bibfield  {journal} {\bibinfo  {journal} {Nature Communications}\ }\textbf
  {\bibinfo {volume} {15}},\ \bibinfo {pages} {6507} (\bibinfo {year}
  {2024})}\BibitemShut {NoStop}%
\bibitem [{\citenamefont {Schiffer}\ \emph {et~al.}(2024)\citenamefont
  {Schiffer}, \citenamefont {Rubio}, \citenamefont {Trivedi},\ and\
  \citenamefont {Cirac}}]{schiffer2024quantum}%
  \BibitemOpen
  \bibfield  {author} {\bibinfo {author} {\bibfnamefont {B.~F.}\ \bibnamefont
  {Schiffer}}, \bibinfo {author} {\bibfnamefont {A.~F.}\ \bibnamefont {Rubio}},
  \bibinfo {author} {\bibfnamefont {R.}~\bibnamefont {Trivedi}},\ and\ \bibinfo
  {author} {\bibfnamefont {J.~I.}\ \bibnamefont {Cirac}},\ }\bibfield  {title}
  {\bibinfo {title} {The quantum adiabatic algorithm suppresses the
  proliferation of errors},\ }\href@noop {} {\bibfield  {journal} {\bibinfo
  {journal} {arXiv preprint arXiv:2404.15397}\ } (\bibinfo {year}
  {2024})}\BibitemShut {NoStop}%
\bibitem [{\citenamefont {Kashyap}\ \emph {et~al.}(2025)\citenamefont
  {Kashyap}, \citenamefont {Styliaris}, \citenamefont {Mouradian},
  \citenamefont {Cirac},\ and\ \citenamefont {Trivedi}}]{kashyap2025accuracy}%
  \BibitemOpen
  \bibfield  {author} {\bibinfo {author} {\bibfnamefont {V.}~\bibnamefont
  {Kashyap}}, \bibinfo {author} {\bibfnamefont {G.}~\bibnamefont {Styliaris}},
  \bibinfo {author} {\bibfnamefont {S.}~\bibnamefont {Mouradian}}, \bibinfo
  {author} {\bibfnamefont {J.~I.}\ \bibnamefont {Cirac}},\ and\ \bibinfo
  {author} {\bibfnamefont {R.}~\bibnamefont {Trivedi}},\ }\bibfield  {title}
  {\bibinfo {title} {Accuracy guarantees and quantum advantage in analog open
  quantum simulation with and without noise},\ }\href@noop {} {\bibfield
  {journal} {\bibinfo  {journal} {Physical Review X}\ }\textbf {\bibinfo
  {volume} {15}},\ \bibinfo {pages} {021017} (\bibinfo {year}
  {2025})}\BibitemShut {NoStop}%
\bibitem [{\citenamefont {Bachmann}\ \emph {et~al.}(2012)\citenamefont
  {Bachmann}, \citenamefont {Michalakis}, \citenamefont {Nachtergaele},\ and\
  \citenamefont {Sims}}]{bachmann2012automorphic}%
  \BibitemOpen
  \bibfield  {author} {\bibinfo {author} {\bibfnamefont {S.}~\bibnamefont
  {Bachmann}}, \bibinfo {author} {\bibfnamefont {S.}~\bibnamefont
  {Michalakis}}, \bibinfo {author} {\bibfnamefont {B.}~\bibnamefont
  {Nachtergaele}},\ and\ \bibinfo {author} {\bibfnamefont {R.}~\bibnamefont
  {Sims}},\ }\bibfield  {title} {\bibinfo {title} {Automorphic equivalence
  within gapped phases of quantum lattice systems},\ }\href@noop {} {\bibfield
  {journal} {\bibinfo  {journal} {Communications in Mathematical Physics}\
  }\textbf {\bibinfo {volume} {309}},\ \bibinfo {pages} {835} (\bibinfo {year}
  {2012})}\BibitemShut {NoStop}%
\bibitem [{\citenamefont {Hastings}\ and\ \citenamefont
  {Wen}(2005)}]{hastings2005quasiadiabatic}%
  \BibitemOpen
  \bibfield  {author} {\bibinfo {author} {\bibfnamefont {M.~B.}\ \bibnamefont
  {Hastings}}\ and\ \bibinfo {author} {\bibfnamefont {X.-G.}\ \bibnamefont
  {Wen}},\ }\bibfield  {title} {\bibinfo {title} {Quasiadiabatic continuation
  of quantum states: The stability of topological ground-state degeneracy and
  emergent gauge invariance},\ }\href@noop {} {\bibfield  {journal} {\bibinfo
  {journal} {Physical Review B—Condensed Matter and Materials Physics}\
  }\textbf {\bibinfo {volume} {72}},\ \bibinfo {pages} {045141} (\bibinfo
  {year} {2005})}\BibitemShut {NoStop}%
\bibitem [{\citenamefont {Heyl}\ \emph {et~al.}(2019)\citenamefont {Heyl},
  \citenamefont {Hauke},\ and\ \citenamefont {Zoller}}]{heyl2019quantum}%
  \BibitemOpen
  \bibfield  {author} {\bibinfo {author} {\bibfnamefont {M.}~\bibnamefont
  {Heyl}}, \bibinfo {author} {\bibfnamefont {P.}~\bibnamefont {Hauke}},\ and\
  \bibinfo {author} {\bibfnamefont {P.}~\bibnamefont {Zoller}},\ }\bibfield
  {title} {\bibinfo {title} {Quantum localization bounds trotter errors in
  digital quantum simulation},\ }\href@noop {} {\bibfield  {journal} {\bibinfo
  {journal} {Science advances}\ }\textbf {\bibinfo {volume} {5}},\ \bibinfo
  {pages} {eaau8342} (\bibinfo {year} {2019})}\BibitemShut {NoStop}%
\bibitem [{\citenamefont {Vanderstraeten}\ \emph {et~al.}(2018)\citenamefont
  {Vanderstraeten}, \citenamefont {Van~Damme}, \citenamefont {B{\"u}chler},\
  and\ \citenamefont {Verstraete}}]{vanderstraeten2018quasiparticles}%
  \BibitemOpen
  \bibfield  {author} {\bibinfo {author} {\bibfnamefont {L.}~\bibnamefont
  {Vanderstraeten}}, \bibinfo {author} {\bibfnamefont {M.}~\bibnamefont
  {Van~Damme}}, \bibinfo {author} {\bibfnamefont {H.~P.}\ \bibnamefont
  {B{\"u}chler}},\ and\ \bibinfo {author} {\bibfnamefont {F.}~\bibnamefont
  {Verstraete}},\ }\bibfield  {title} {\bibinfo {title} {Quasiparticles in
  quantum spin chains with long-range interactions},\ }\href@noop {} {\bibfield
   {journal} {\bibinfo  {journal} {Physical Review Letters}\ }\textbf {\bibinfo
  {volume} {121}},\ \bibinfo {pages} {090603} (\bibinfo {year}
  {2018})}\BibitemShut {NoStop}%
\bibitem [{\citenamefont {Binder}\ and\ \citenamefont
  {Landau}(1980)}]{binder1980phase}%
  \BibitemOpen
  \bibfield  {author} {\bibinfo {author} {\bibfnamefont {K.}~\bibnamefont
  {Binder}}\ and\ \bibinfo {author} {\bibfnamefont {D.}~\bibnamefont
  {Landau}},\ }\bibfield  {title} {\bibinfo {title} {Phase diagrams and
  critical behavior in ising square lattices with nearest-and
  next-nearest-neighbor interactions},\ }\href@noop {} {\bibfield  {journal}
  {\bibinfo  {journal} {Physical Review B}\ }\textbf {\bibinfo {volume} {21}},\
  \bibinfo {pages} {1941} (\bibinfo {year} {1980})}\BibitemShut {NoStop}%
\bibitem [{\citenamefont {Rademaker}\ and\ \citenamefont
  {Abanin}(2020)}]{rademaker2020slow}%
  \BibitemOpen
  \bibfield  {author} {\bibinfo {author} {\bibfnamefont {L.}~\bibnamefont
  {Rademaker}}\ and\ \bibinfo {author} {\bibfnamefont {D.~A.}\ \bibnamefont
  {Abanin}},\ }\bibfield  {title} {\bibinfo {title} {Slow nonthermalizing
  dynamics in a quantum spin glass},\ }\href@noop {} {\bibfield  {journal}
  {\bibinfo  {journal} {Physical review letters}\ }\textbf {\bibinfo {volume}
  {125}},\ \bibinfo {pages} {260405} (\bibinfo {year} {2020})}\BibitemShut
  {NoStop}%
\bibitem [{\citenamefont {Panchenko}(2014)}]{panchenko2014parisi}%
  \BibitemOpen
  \bibfield  {author} {\bibinfo {author} {\bibfnamefont {D.}~\bibnamefont
  {Panchenko}},\ }\bibfield  {title} {\bibinfo {title} {The parisi formula for
  mixed p-spin models},\ }\href@noop {} {\  (\bibinfo {year}
  {2014})}\BibitemShut {NoStop}%
\bibitem [{\citenamefont {Sachdev}\ and\ \citenamefont
  {Ye}(1993)}]{sachdev1993gapless}%
  \BibitemOpen
  \bibfield  {author} {\bibinfo {author} {\bibfnamefont {S.}~\bibnamefont
  {Sachdev}}\ and\ \bibinfo {author} {\bibfnamefont {J.}~\bibnamefont {Ye}},\
  }\bibfield  {title} {\bibinfo {title} {Gapless spin-fluid ground state in a
  random quantum heisenberg magnet},\ }\href@noop {} {\bibfield  {journal}
  {\bibinfo  {journal} {Physical review letters}\ }\textbf {\bibinfo {volume}
  {70}},\ \bibinfo {pages} {3339} (\bibinfo {year} {1993})}\BibitemShut
  {NoStop}%
\bibitem [{\citenamefont {Kogut}\ and\ \citenamefont
  {Susskind}(1975)}]{kogut1975hamiltonian}%
  \BibitemOpen
  \bibfield  {author} {\bibinfo {author} {\bibfnamefont {J.}~\bibnamefont
  {Kogut}}\ and\ \bibinfo {author} {\bibfnamefont {L.}~\bibnamefont
  {Susskind}},\ }\bibfield  {title} {\bibinfo {title} {Hamiltonian formulation
  of wilson's lattice gauge theories},\ }\href@noop {} {\bibfield  {journal}
  {\bibinfo  {journal} {Physical Review D}\ }\textbf {\bibinfo {volume} {11}},\
  \bibinfo {pages} {395} (\bibinfo {year} {1975})}\BibitemShut {NoStop}%
\bibitem [{\citenamefont {Bodwin}\ and\ \citenamefont
  {Kovacs}(1987)}]{bodwin1987lattice}%
  \BibitemOpen
  \bibfield  {author} {\bibinfo {author} {\bibfnamefont {G.~T.}\ \bibnamefont
  {Bodwin}}\ and\ \bibinfo {author} {\bibfnamefont {E.~V.}\ \bibnamefont
  {Kovacs}},\ }\bibfield  {title} {\bibinfo {title} {Lattice fermions in the
  schwinger model},\ }\href@noop {} {\bibfield  {journal} {\bibinfo  {journal}
  {Physical Review D}\ }\textbf {\bibinfo {volume} {35}},\ \bibinfo {pages}
  {3198} (\bibinfo {year} {1987})}\BibitemShut {NoStop}%
\bibitem [{\citenamefont {Kogut}(1979)}]{kogut1979introduction}%
  \BibitemOpen
  \bibfield  {author} {\bibinfo {author} {\bibfnamefont {J.~B.}\ \bibnamefont
  {Kogut}},\ }\bibfield  {title} {\bibinfo {title} {An introduction to lattice
  gauge theory and spin systems},\ }\href@noop {} {\bibfield  {journal}
  {\bibinfo  {journal} {Reviews of Modern Physics}\ }\textbf {\bibinfo {volume}
  {51}},\ \bibinfo {pages} {659} (\bibinfo {year} {1979})}\BibitemShut
  {NoStop}%
\bibitem [{\citenamefont {Kogut}(1983)}]{kogut1983lattice}%
  \BibitemOpen
  \bibfield  {author} {\bibinfo {author} {\bibfnamefont {J.~B.}\ \bibnamefont
  {Kogut}},\ }\bibfield  {title} {\bibinfo {title} {The lattice gauge theory
  approach to quantum chromodynamics},\ }\href@noop {} {\bibfield  {journal}
  {\bibinfo  {journal} {Reviews of Modern Physics}\ }\textbf {\bibinfo {volume}
  {55}},\ \bibinfo {pages} {775} (\bibinfo {year} {1983})}\BibitemShut
  {NoStop}%
\bibitem [{\citenamefont {Chandrasekharan}\ and\ \citenamefont
  {Wiese}(1997)}]{chandrasekharan1997quantum}%
  \BibitemOpen
  \bibfield  {author} {\bibinfo {author} {\bibfnamefont {S.}~\bibnamefont
  {Chandrasekharan}}\ and\ \bibinfo {author} {\bibfnamefont {U.-J.}\
  \bibnamefont {Wiese}},\ }\bibfield  {title} {\bibinfo {title} {Quantum link
  models: A discrete approach to gauge theories},\ }\href@noop {} {\bibfield
  {journal} {\bibinfo  {journal} {Nuclear Physics B}\ }\textbf {\bibinfo
  {volume} {492}},\ \bibinfo {pages} {455} (\bibinfo {year}
  {1997})}\BibitemShut {NoStop}%
\bibitem [{\citenamefont {Zohar}\ and\ \citenamefont
  {Burrello}(2015)}]{zohar2015formulation}%
  \BibitemOpen
  \bibfield  {author} {\bibinfo {author} {\bibfnamefont {E.}~\bibnamefont
  {Zohar}}\ and\ \bibinfo {author} {\bibfnamefont {M.}~\bibnamefont
  {Burrello}},\ }\bibfield  {title} {\bibinfo {title} {Formulation of lattice
  gauge theories for quantum simulations},\ }\href@noop {} {\bibfield
  {journal} {\bibinfo  {journal} {Physical Review D}\ }\textbf {\bibinfo
  {volume} {91}},\ \bibinfo {pages} {054506} (\bibinfo {year}
  {2015})}\BibitemShut {NoStop}%
\bibitem [{\citenamefont {Childs}\ \emph {et~al.}(2021)\citenamefont {Childs},
  \citenamefont {Su}, \citenamefont {Tran}, \citenamefont {Wiebe},\ and\
  \citenamefont {Zhu}}]{childs2021theory}%
  \BibitemOpen
  \bibfield  {author} {\bibinfo {author} {\bibfnamefont {A.~M.}\ \bibnamefont
  {Childs}}, \bibinfo {author} {\bibfnamefont {Y.}~\bibnamefont {Su}}, \bibinfo
  {author} {\bibfnamefont {M.~C.}\ \bibnamefont {Tran}}, \bibinfo {author}
  {\bibfnamefont {N.}~\bibnamefont {Wiebe}},\ and\ \bibinfo {author}
  {\bibfnamefont {S.}~\bibnamefont {Zhu}},\ }\bibfield  {title} {\bibinfo
  {title} {Theory of trotter error with commutator scaling},\ }\href@noop {}
  {\bibfield  {journal} {\bibinfo  {journal} {Physical Review X}\ }\textbf
  {\bibinfo {volume} {11}},\ \bibinfo {pages} {011020} (\bibinfo {year}
  {2021})}\BibitemShut {NoStop}%
\bibitem [{\citenamefont {Gong}\ \emph {et~al.}(2024)\citenamefont {Gong},
  \citenamefont {Zhou},\ and\ \citenamefont {Li}}]{gong2024complexity}%
  \BibitemOpen
  \bibfield  {author} {\bibinfo {author} {\bibfnamefont {W.}~\bibnamefont
  {Gong}}, \bibinfo {author} {\bibfnamefont {S.}~\bibnamefont {Zhou}},\ and\
  \bibinfo {author} {\bibfnamefont {T.}~\bibnamefont {Li}},\ }\bibfield
  {title} {\bibinfo {title} {Complexity of digital quantum simulation in the
  low-energy subspace: Applications and a lower bound},\ }\href@noop {}
  {\bibfield  {journal} {\bibinfo  {journal} {Quantum}\ }\textbf {\bibinfo
  {volume} {8}},\ \bibinfo {pages} {1409} (\bibinfo {year} {2024})}\BibitemShut
  {NoStop}%
\bibitem [{\citenamefont {Suzuki}(1993)}]{suzuki1993general}%
  \BibitemOpen
  \bibfield  {author} {\bibinfo {author} {\bibfnamefont {M.}~\bibnamefont
  {Suzuki}},\ }\bibfield  {title} {\bibinfo {title} {General decomposition
  theory of ordered exponentials},\ }\href@noop {} {\bibfield  {journal}
  {\bibinfo  {journal} {Proceedings of the Japan Academy, Series B}\ }\textbf
  {\bibinfo {volume} {69}},\ \bibinfo {pages} {161} (\bibinfo {year}
  {1993})}\BibitemShut {NoStop}%
\bibitem [{\citenamefont {Mizuta}\ and\ \citenamefont
  {Kuwahara}(2025)}]{mizuta2025trotterization}%
  \BibitemOpen
  \bibfield  {author} {\bibinfo {author} {\bibfnamefont {K.}~\bibnamefont
  {Mizuta}}\ and\ \bibinfo {author} {\bibfnamefont {T.}~\bibnamefont
  {Kuwahara}},\ }\bibfield  {title} {\bibinfo {title} {Trotterization is
  substantially efficient for low-energy states},\ }\href@noop {} {\bibfield
  {journal} {\bibinfo  {journal} {arXiv preprint arXiv:2504.20746}\ } (\bibinfo
  {year} {2025})}\BibitemShut {NoStop}%
\bibitem [{\citenamefont {{\c{S}}ahino{\u{g}}lu}\ and\ \citenamefont
  {Somma}(2021)}]{csahinouglu2021hamiltonian}%
  \BibitemOpen
  \bibfield  {author} {\bibinfo {author} {\bibfnamefont {B.}~\bibnamefont
  {{\c{S}}ahino{\u{g}}lu}}\ and\ \bibinfo {author} {\bibfnamefont {R.~D.}\
  \bibnamefont {Somma}},\ }\bibfield  {title} {\bibinfo {title} {Hamiltonian
  simulation in the low-energy subspace},\ }\href@noop {} {\bibfield  {journal}
  {\bibinfo  {journal} {npj Quantum Information}\ }\textbf {\bibinfo {volume}
  {7}},\ \bibinfo {pages} {119} (\bibinfo {year} {2021})}\BibitemShut {NoStop}%
\bibitem [{\citenamefont {Burak}\ and\ \citenamefont
  {Somma}(2021)}]{burak2021hamiltonian}%
  \BibitemOpen
  \bibfield  {author} {\bibinfo {author} {\bibfnamefont {{\c{S}}.}~\bibnamefont
  {Burak}}\ and\ \bibinfo {author} {\bibfnamefont {R.~D.}\ \bibnamefont
  {Somma}},\ }\bibfield  {title} {\bibinfo {title} {Hamiltonian simulation in
  the low-energy subspace},\ }\href@noop {} {\bibfield  {journal} {\bibinfo
  {journal} {NPJ Quantum Information}\ }\textbf {\bibinfo {volume} {7}}
  (\bibinfo {year} {2021})}\BibitemShut {NoStop}%
\bibitem [{\citenamefont {Tran}\ \emph {et~al.}(2020)\citenamefont {Tran},
  \citenamefont {Chu}, \citenamefont {Su}, \citenamefont {Childs},\ and\
  \citenamefont {Gorshkov}}]{tran2020destructive}%
  \BibitemOpen
  \bibfield  {author} {\bibinfo {author} {\bibfnamefont {M.~C.}\ \bibnamefont
  {Tran}}, \bibinfo {author} {\bibfnamefont {S.-K.}\ \bibnamefont {Chu}},
  \bibinfo {author} {\bibfnamefont {Y.}~\bibnamefont {Su}}, \bibinfo {author}
  {\bibfnamefont {A.~M.}\ \bibnamefont {Childs}},\ and\ \bibinfo {author}
  {\bibfnamefont {A.~V.}\ \bibnamefont {Gorshkov}},\ }\bibfield  {title}
  {\bibinfo {title} {Destructive error interference in product-formula lattice
  simulation},\ }\href@noop {} {\bibfield  {journal} {\bibinfo  {journal}
  {Physical review letters}\ }\textbf {\bibinfo {volume} {124}},\ \bibinfo
  {pages} {220502} (\bibinfo {year} {2020})}\BibitemShut {NoStop}%
\bibitem [{\citenamefont {Su}\ \emph {et~al.}(2021)\citenamefont {Su},
  \citenamefont {Huang},\ and\ \citenamefont {Campbell}}]{su2021nearly}%
  \BibitemOpen
  \bibfield  {author} {\bibinfo {author} {\bibfnamefont {Y.}~\bibnamefont
  {Su}}, \bibinfo {author} {\bibfnamefont {H.-Y.}\ \bibnamefont {Huang}},\ and\
  \bibinfo {author} {\bibfnamefont {E.~T.}\ \bibnamefont {Campbell}},\
  }\bibfield  {title} {\bibinfo {title} {Nearly tight trotterization of
  interacting electrons},\ }\href@noop {} {\bibfield  {journal} {\bibinfo
  {journal} {Quantum}\ }\textbf {\bibinfo {volume} {5}},\ \bibinfo {pages}
  {495} (\bibinfo {year} {2021})}\BibitemShut {NoStop}%
\bibitem [{\citenamefont {Tong}\ \emph {et~al.}(2021)\citenamefont {Tong},
  \citenamefont {Albert}, \citenamefont {McClean}, \citenamefont {Preskill},\
  and\ \citenamefont {Su}}]{tong2021provably}%
  \BibitemOpen
  \bibfield  {author} {\bibinfo {author} {\bibfnamefont {Y.}~\bibnamefont
  {Tong}}, \bibinfo {author} {\bibfnamefont {V.~V.}\ \bibnamefont {Albert}},
  \bibinfo {author} {\bibfnamefont {J.~R.}\ \bibnamefont {McClean}}, \bibinfo
  {author} {\bibfnamefont {J.}~\bibnamefont {Preskill}},\ and\ \bibinfo
  {author} {\bibfnamefont {Y.}~\bibnamefont {Su}},\ }\bibfield  {title}
  {\bibinfo {title} {Provably accurate simulation of gauge theories and bosonic
  systems},\ }\href@noop {} {\bibfield  {journal} {\bibinfo  {journal} {arXiv
  preprint arXiv:2110.06942}\ } (\bibinfo {year} {2021})}\BibitemShut {NoStop}%
\bibitem [{\citenamefont {Childs}\ and\ \citenamefont
  {Su}(2019)}]{childs2019nearly}%
  \BibitemOpen
  \bibfield  {author} {\bibinfo {author} {\bibfnamefont {A.~M.}\ \bibnamefont
  {Childs}}\ and\ \bibinfo {author} {\bibfnamefont {Y.}~\bibnamefont {Su}},\
  }\bibfield  {title} {\bibinfo {title} {Nearly optimal lattice simulation by
  product formulas},\ }\href@noop {} {\bibfield  {journal} {\bibinfo  {journal}
  {Physical review letters}\ }\textbf {\bibinfo {volume} {123}},\ \bibinfo
  {pages} {050503} (\bibinfo {year} {2019})}\BibitemShut {NoStop}%
\bibitem [{\citenamefont {Lloyd}(1996)}]{lloyd1996universal}%
  \BibitemOpen
  \bibfield  {author} {\bibinfo {author} {\bibfnamefont {S.}~\bibnamefont
  {Lloyd}},\ }\bibfield  {title} {\bibinfo {title} {Universal quantum
  simulators},\ }\href@noop {} {\bibfield  {journal} {\bibinfo  {journal}
  {Science}\ }\textbf {\bibinfo {volume} {273}},\ \bibinfo {pages} {1073}
  (\bibinfo {year} {1996})}\BibitemShut {NoStop}%
\bibitem [{\citenamefont {Hatano}\ and\ \citenamefont
  {Suzuki}(2005)}]{hatano2005finding}%
  \BibitemOpen
  \bibfield  {author} {\bibinfo {author} {\bibfnamefont {N.}~\bibnamefont
  {Hatano}}\ and\ \bibinfo {author} {\bibfnamefont {M.}~\bibnamefont
  {Suzuki}},\ }\bibfield  {title} {\bibinfo {title} {Finding exponential
  product formulas of higher orders},\ }in\ \href@noop {} {\emph {\bibinfo
  {booktitle} {Quantum annealing and other optimization methods}}}\ (\bibinfo
  {publisher} {Springer},\ \bibinfo {year} {2005})\ pp.\ \bibinfo {pages}
  {37--68}\BibitemShut {NoStop}%
\bibitem [{\citenamefont {Abanin}\ \emph
  {et~al.}(2017{\natexlab{a}})\citenamefont {Abanin}, \citenamefont {De~Roeck},
  \citenamefont {Ho},\ and\ \citenamefont {Huveneers}}]{abanin2017effective}%
  \BibitemOpen
  \bibfield  {author} {\bibinfo {author} {\bibfnamefont {D.~A.}\ \bibnamefont
  {Abanin}}, \bibinfo {author} {\bibfnamefont {W.}~\bibnamefont {De~Roeck}},
  \bibinfo {author} {\bibfnamefont {W.~W.}\ \bibnamefont {Ho}},\ and\ \bibinfo
  {author} {\bibfnamefont {F.}~\bibnamefont {Huveneers}},\ }\bibfield  {title}
  {\bibinfo {title} {Effective hamiltonians, prethermalization, and slow energy
  absorption in periodically driven many-body systems},\ }\href@noop {}
  {\bibfield  {journal} {\bibinfo  {journal} {Physical Review B}\ }\textbf
  {\bibinfo {volume} {95}},\ \bibinfo {pages} {014112} (\bibinfo {year}
  {2017}{\natexlab{a}})}\BibitemShut {NoStop}%
\bibitem [{\citenamefont {Abanin}\ \emph
  {et~al.}(2017{\natexlab{b}})\citenamefont {Abanin}, \citenamefont {De~Roeck},
  \citenamefont {Ho},\ and\ \citenamefont {Huveneers}}]{abanin2017rigorous}%
  \BibitemOpen
  \bibfield  {author} {\bibinfo {author} {\bibfnamefont {D.}~\bibnamefont
  {Abanin}}, \bibinfo {author} {\bibfnamefont {W.}~\bibnamefont {De~Roeck}},
  \bibinfo {author} {\bibfnamefont {W.~W.}\ \bibnamefont {Ho}},\ and\ \bibinfo
  {author} {\bibfnamefont {F.}~\bibnamefont {Huveneers}},\ }\bibfield  {title}
  {\bibinfo {title} {A rigorous theory of many-body prethermalization for
  periodically driven and closed quantum systems},\ }\href@noop {} {\bibfield
  {journal} {\bibinfo  {journal} {Communications in Mathematical Physics}\
  }\textbf {\bibinfo {volume} {354}},\ \bibinfo {pages} {809} (\bibinfo {year}
  {2017}{\natexlab{b}})}\BibitemShut {NoStop}%
\bibitem [{\citenamefont {Hung}\ \emph {et~al.}(2016)\citenamefont {Hung},
  \citenamefont {Gonz{\'a}lez-Tudela}, \citenamefont {Cirac},\ and\
  \citenamefont {Kimble}}]{hung2016quantum}%
  \BibitemOpen
  \bibfield  {author} {\bibinfo {author} {\bibfnamefont {C.-L.}\ \bibnamefont
  {Hung}}, \bibinfo {author} {\bibfnamefont {A.}~\bibnamefont
  {Gonz{\'a}lez-Tudela}}, \bibinfo {author} {\bibfnamefont {J.~I.}\
  \bibnamefont {Cirac}},\ and\ \bibinfo {author} {\bibfnamefont
  {H.}~\bibnamefont {Kimble}},\ }\bibfield  {title} {\bibinfo {title} {Quantum
  spin dynamics with pairwise-tunable, long-range interactions},\ }\href@noop
  {} {\bibfield  {journal} {\bibinfo  {journal} {Proceedings of the National
  Academy of Sciences}\ }\textbf {\bibinfo {volume} {113}},\ \bibinfo {pages}
  {E4946} (\bibinfo {year} {2016})}\BibitemShut {NoStop}%
\bibitem [{\citenamefont {Mori}\ \emph {et~al.}(2016)\citenamefont {Mori},
  \citenamefont {Kuwahara},\ and\ \citenamefont {Saito}}]{mori2016rigorous}%
  \BibitemOpen
  \bibfield  {author} {\bibinfo {author} {\bibfnamefont {T.}~\bibnamefont
  {Mori}}, \bibinfo {author} {\bibfnamefont {T.}~\bibnamefont {Kuwahara}},\
  and\ \bibinfo {author} {\bibfnamefont {K.}~\bibnamefont {Saito}},\ }\bibfield
   {title} {\bibinfo {title} {Rigorous bound on energy absorption and generic
  relaxation in periodically driven quantum systems},\ }\href@noop {}
  {\bibfield  {journal} {\bibinfo  {journal} {Physical review letters}\
  }\textbf {\bibinfo {volume} {116}},\ \bibinfo {pages} {120401} (\bibinfo
  {year} {2016})}\BibitemShut {NoStop}%
\bibitem [{\citenamefont {Zohar}\ \emph {et~al.}(2017)\citenamefont {Zohar},
  \citenamefont {Farace}, \citenamefont {Reznik},\ and\ \citenamefont
  {Cirac}}]{zohar2017digital}%
  \BibitemOpen
  \bibfield  {author} {\bibinfo {author} {\bibfnamefont {E.}~\bibnamefont
  {Zohar}}, \bibinfo {author} {\bibfnamefont {A.}~\bibnamefont {Farace}},
  \bibinfo {author} {\bibfnamefont {B.}~\bibnamefont {Reznik}},\ and\ \bibinfo
  {author} {\bibfnamefont {J.~I.}\ \bibnamefont {Cirac}},\ }\bibfield  {title}
  {\bibinfo {title} {Digital quantum simulation of z 2 lattice gauge theories
  with dynamical fermionic matter},\ }\href@noop {} {\bibfield  {journal}
  {\bibinfo  {journal} {Physical review letters}\ }\textbf {\bibinfo {volume}
  {118}},\ \bibinfo {pages} {070501} (\bibinfo {year} {2017})}\BibitemShut
  {NoStop}%
\bibitem [{\citenamefont {Zohar}\ \emph {et~al.}(2013)\citenamefont {Zohar},
  \citenamefont {Cirac},\ and\ \citenamefont {Reznik}}]{zohar2013quantum}%
  \BibitemOpen
  \bibfield  {author} {\bibinfo {author} {\bibfnamefont {E.}~\bibnamefont
  {Zohar}}, \bibinfo {author} {\bibfnamefont {J.~I.}\ \bibnamefont {Cirac}},\
  and\ \bibinfo {author} {\bibfnamefont {B.}~\bibnamefont {Reznik}},\
  }\bibfield  {title} {\bibinfo {title} {Quantum simulations of gauge theories
  with ultracold atoms: Local gauge invariance from angular-momentum
  conservation},\ }\href@noop {} {\bibfield  {journal} {\bibinfo  {journal}
  {Physical Review A—Atomic, Molecular, and Optical Physics}\ }\textbf
  {\bibinfo {volume} {88}},\ \bibinfo {pages} {023617} (\bibinfo {year}
  {2013})}\BibitemShut {NoStop}%
\bibitem [{\citenamefont {Bravyi}\ \emph {et~al.}(2008)\citenamefont {Bravyi},
  \citenamefont {DiVincenzo}, \citenamefont {Loss},\ and\ \citenamefont
  {Terhal}}]{bravyi2008quantum}%
  \BibitemOpen
  \bibfield  {author} {\bibinfo {author} {\bibfnamefont {S.}~\bibnamefont
  {Bravyi}}, \bibinfo {author} {\bibfnamefont {D.~P.}\ \bibnamefont
  {DiVincenzo}}, \bibinfo {author} {\bibfnamefont {D.}~\bibnamefont {Loss}},\
  and\ \bibinfo {author} {\bibfnamefont {B.~M.}\ \bibnamefont {Terhal}},\
  }\bibfield  {title} {\bibinfo {title} {Quantum simulation of many-body
  hamiltonians using perturbation theory<? format?> with bounded-strength
  interactions},\ }\href@noop {} {\bibfield  {journal} {\bibinfo  {journal}
  {Physical review letters}\ }\textbf {\bibinfo {volume} {101}},\ \bibinfo
  {pages} {070503} (\bibinfo {year} {2008})}\BibitemShut {NoStop}%
\bibitem [{\citenamefont {Feng}\ \emph {et~al.}(2025)\citenamefont {Feng},
  \citenamefont {Cao},\ and\ \citenamefont {Zhao}}]{feng2025trotterization}%
  \BibitemOpen
  \bibfield  {author} {\bibinfo {author} {\bibfnamefont {T.}~\bibnamefont
  {Feng}}, \bibinfo {author} {\bibfnamefont {Y.}~\bibnamefont {Cao}},\ and\
  \bibinfo {author} {\bibfnamefont {Q.}~\bibnamefont {Zhao}},\ }\bibfield
  {title} {\bibinfo {title} {Trotterization, operator scrambling, and
  entanglement},\ }\href@noop {} {\bibfield  {journal} {\bibinfo  {journal}
  {arXiv preprint arXiv:2506.23345}\ } (\bibinfo {year} {2025})}\BibitemShut
  {NoStop}%
\bibitem [{\citenamefont {Hastings}\ and\ \citenamefont
  {Koma}(2006)}]{hastings2006spectral}%
  \BibitemOpen
  \bibfield  {author} {\bibinfo {author} {\bibfnamefont {M.~B.}\ \bibnamefont
  {Hastings}}\ and\ \bibinfo {author} {\bibfnamefont {T.}~\bibnamefont
  {Koma}},\ }\bibfield  {title} {\bibinfo {title} {Spectral gap and exponential
  decay of correlations},\ }\href@noop {} {\bibfield  {journal} {\bibinfo
  {journal} {Communications in mathematical physics}\ }\textbf {\bibinfo
  {volume} {265}},\ \bibinfo {pages} {781} (\bibinfo {year}
  {2006})}\BibitemShut {NoStop}%
\bibitem [{\citenamefont {Aliferis}\ \emph {et~al.}(2005)\citenamefont
  {Aliferis}, \citenamefont {Gottesman},\ and\ \citenamefont
  {Preskill}}]{aliferis2005quantum}%
  \BibitemOpen
  \bibfield  {author} {\bibinfo {author} {\bibfnamefont {P.}~\bibnamefont
  {Aliferis}}, \bibinfo {author} {\bibfnamefont {D.}~\bibnamefont
  {Gottesman}},\ and\ \bibinfo {author} {\bibfnamefont {J.}~\bibnamefont
  {Preskill}},\ }\bibfield  {title} {\bibinfo {title} {Quantum accuracy
  threshold for concatenated distance-3 codes},\ }\href@noop {} {\bibfield
  {journal} {\bibinfo  {journal} {arXiv preprint quant-ph/0504218}\ } (\bibinfo
  {year} {2005})}\BibitemShut {NoStop}%
\bibitem [{\citenamefont {Aharonov}\ and\ \citenamefont
  {Ben-Or}(1997)}]{aharonov1997fault}%
  \BibitemOpen
  \bibfield  {author} {\bibinfo {author} {\bibfnamefont {D.}~\bibnamefont
  {Aharonov}}\ and\ \bibinfo {author} {\bibfnamefont {M.}~\bibnamefont
  {Ben-Or}},\ }\bibfield  {title} {\bibinfo {title} {Fault-tolerant quantum
  computation with constant error},\ }in\ \href@noop {} {\emph {\bibinfo
  {booktitle} {Proceedings of the twenty-ninth annual ACM symposium on Theory
  of computing}}}\ (\bibinfo {year} {1997})\ pp.\ \bibinfo {pages}
  {176--188}\BibitemShut {NoStop}%
\bibitem [{\citenamefont {Gottesman}(2024)}]{gottesman2024surviving}%
  \BibitemOpen
  \bibfield  {author} {\bibinfo {author} {\bibfnamefont {D.}~\bibnamefont
  {Gottesman}},\ }\bibfield  {title} {\bibinfo {title} {Surviving as a quantum
  computer in a classical world},\ }\href@noop {} {\bibfield  {journal}
  {\bibinfo  {journal} {Textbook manuscript preprint}\ } (\bibinfo {year}
  {2024})}\BibitemShut {NoStop}%
\bibitem [{\citenamefont {Blanes}\ \emph {et~al.}(2009)\citenamefont {Blanes},
  \citenamefont {Casas}, \citenamefont {Oteo},\ and\ \citenamefont
  {Ros}}]{blanes2009magnus}%
  \BibitemOpen
  \bibfield  {author} {\bibinfo {author} {\bibfnamefont {S.}~\bibnamefont
  {Blanes}}, \bibinfo {author} {\bibfnamefont {F.}~\bibnamefont {Casas}},
  \bibinfo {author} {\bibfnamefont {J.-A.}\ \bibnamefont {Oteo}},\ and\
  \bibinfo {author} {\bibfnamefont {J.}~\bibnamefont {Ros}},\ }\bibfield
  {title} {\bibinfo {title} {The magnus expansion and some of its
  applications},\ }\href@noop {} {\bibfield  {journal} {\bibinfo  {journal}
  {Physics reports}\ }\textbf {\bibinfo {volume} {470}},\ \bibinfo {pages}
  {151} (\bibinfo {year} {2009})}\BibitemShut {NoStop}%
\bibitem [{goo(2025)}]{google2025quantum}%
  \BibitemOpen
  \bibfield  {title} {\bibinfo {title} {Quantum error correction below the
  surface code threshold},\ }\href@noop {} {\bibfield  {journal} {\bibinfo
  {journal} {Nature}\ }\textbf {\bibinfo {volume} {638}},\ \bibinfo {pages}
  {920} (\bibinfo {year} {2025})}\BibitemShut {NoStop}%
\bibitem [{\citenamefont {Cai}\ \emph {et~al.}(2023)\citenamefont {Cai},
  \citenamefont {Tong},\ and\ \citenamefont {Preskill}}]{cai2023stochastic}%
  \BibitemOpen
  \bibfield  {author} {\bibinfo {author} {\bibfnamefont {Y.}~\bibnamefont
  {Cai}}, \bibinfo {author} {\bibfnamefont {Y.}~\bibnamefont {Tong}},\ and\
  \bibinfo {author} {\bibfnamefont {J.}~\bibnamefont {Preskill}},\ }\bibfield
  {title} {\bibinfo {title} {Stochastic error cancellation in analog quantum
  simulation},\ }\href@noop {} {\bibfield  {journal} {\bibinfo  {journal}
  {arXiv preprint arXiv:2311.14818}\ } (\bibinfo {year} {2023})}\BibitemShut
  {NoStop}%
\bibitem [{\citenamefont {Mishra}\ \emph {et~al.}(2024)\citenamefont {Mishra},
  \citenamefont {Fr{\'\i}as-P{\'e}rez},\ and\ \citenamefont
  {Trivedi}}]{mishra2024classically}%
  \BibitemOpen
  \bibfield  {author} {\bibinfo {author} {\bibfnamefont {S.~D.}\ \bibnamefont
  {Mishra}}, \bibinfo {author} {\bibfnamefont {M.}~\bibnamefont
  {Fr{\'\i}as-P{\'e}rez}},\ and\ \bibinfo {author} {\bibfnamefont
  {R.}~\bibnamefont {Trivedi}},\ }\bibfield  {title} {\bibinfo {title}
  {Classically computing performance bounds on depolarized quantum circuits},\
  }\href@noop {} {\bibfield  {journal} {\bibinfo  {journal} {PRX Quantum}\
  }\textbf {\bibinfo {volume} {5}},\ \bibinfo {pages} {020317} (\bibinfo {year}
  {2024})}\BibitemShut {NoStop}%
\end{thebibliography}%
\onecolumngrid
\newpage
\onehalfspacing
\section{Preliminaries}
\subsection{Notation}
Given a finite-dimensional Hilbert space $\mathcal{H}$, we will denote by $\text{M}(\mathcal{H})$ the set of bounded linear operators from $\mathcal{H}\to \mathcal{H}$, define $\text{M}_h(\mathcal{H})$ to be the set of the bounded Hermitian operators from $\mathcal{H}\to \mathcal{H}$ and define $\text{D}_1(\mathcal{H})$ as the set of valid density matrices on $\mathcal{H}$. We will typically use the $\dagger$ superscript to indicate the adjoint, or Hermitian conjugate, of an operator or superoperator. However in some cases, more compact expression can obtained by using the following notation --- for some operator or superoperator $X$, we define $X^{(-)}:=X$ and $X^{(+)}:=X^\dagger$. Furthermore we will use $\bar{+}=-$ and $\bar{-}=+$. For example, $X^{(\bar{u})}:=X^\dagger$ for $u=-$.

While dealing with mixed states and their dynamics, it will often be convenient to adopt the vectorized notation, where we map operators on a (finite-dimensional) Hilbert space to state vectors via $\rho = \sum_{i_l, i_r}\rho_{i_l, i_r}\ket{i_l}\!\bra{i_r} \to \vecket{\rho} = \sum_{i_l, i_r}\rho_{i_l, i_r} \ket{i_l, i_r}$. Superoperators, such as Lindladians or channels, will map to ordinary operators in this picture. Given an operator $X \in \text{M}(\mathcal{H})$, we will define $X_l, X_r \in \text{M}(\mathcal{H}\otimes \mathcal{H})$ by $X_l\vecket{\rho} = (X\otimes I)\vecket{\rho} = \vecket{X\rho}$ and $X_r\vecket{\rho} = (I\otimes X^\text{T})\vecket{\rho} =  \vecket{\rho X}$. $X_l (X_r)$ can also be interpreted as a superoperator which left (right) multiplies its argument with $X$ i.e. $X_l(Y) = XY$ and $X_r(Y) = YX$. Given an operator $X \in \text{M}(\mathcal{H})$, we will define $\mathcal{C}_X \in \text{M}(\mathcal{H}\otimes \mathcal{H})$ to be the super-operator corresponding to commutator with $X$ i.e., $\mathcal{C}_X = X_l - X_r$. Note that $\mathcal{C}_X\vecket{\rho} = \vecket{[X, \rho]}$.

\emph{Norms}. $\norm{A}_p$ denotes the Schatten $p$-norm of an operator $A$. We will denote the operator norm, which is also the Schatten-$\infty$ norm, by $\norm{A}$ without an subscript. $\norm{\mathcal{A}}_{p \to q} := \max_{O,\norm{O}_p=1} \norm{\mathcal{A}(O)}_q$ indicates the norm of a superoperator $\mathcal{A}$. We define the completely-bounded norm of a superoperator $\mathcal{A}$ as $\norm{\mathcal{A}}_{cb, p \to q} := \sup_{n\geq 2} \norm{\mathcal{A} \otimes \textnormal{id}_n}_{p \to q}$. The diamond norm is the completely bounded $1\to1$ norm --- for the diamond norm, we will use the standard notation $\norm{\mathcal{A}}_\diamond := \norm{\mathcal{A}}_{cb, 1\to 1}$.

\emph{Asymptotics} For two real-valued functions $f(x)$ and $g(x)$, we will write $f(x) \leq O(g(x))$ to indicate that there exists $C,x_0 \in \mathbb{R}$ such that $f(x) \leq C g(x)$ for all $x>x_0$. Similarly $f(x) \geq \Omega(g(x))$ indicates that there exists $C,x_0 \in \mathbb{R}$ such that $f(x) \geq C g(x)$ for all $x>x_0$. Finally $f(x) = \Theta(g(x))$ indicates $\Omega(g(x)) \leq f(x) \leq O(g(x))$.

\subsection{Lattice models and Lieb Robinson bounds}
We will work with quantum spin models defined on a $d$-dimensional lattice $\mathbb{Z}^d$. For $x, y \in \mathbb{Z}^d$, $d(x,y)$ for two lattice sites $x$ and $y$ denotes the Manhattan distance between $x$ and $y$, i.e., the graph path length to reach $y$ from $x$. The distance between a set of lattice sites $S \subset \mathbb{Z}^d$ and a single site $x$ will be defined as $d(S,x):=\min_{y \in S} d(x,y)$. Similarly, the distance between two sets of lattice sites $S_1$ and $S_2$ will be defined as $d(S_1, S_2) := \min_{x \in S_1, y \in S_2} d(x,y)$. The diameter of set of lattice sites $S$ is written $\textnormal{diam}(S) := \max_{x\in S,y\in S} d(x,y)$.

We now consider $N$ spins located at the vertices of the lattice $\mathbb{Z}^d$. We will work with operators and super-operators defined on the Hilbert space of these spins --- for an operator $M$ or super-operator $\mathcal{M}$, we will denote by $\text{supp}(M), \text{supp}(\mathcal{M}) \subset \mathbb{Z}^d$ the set of vertices where $M$ or $\mathcal{M}$ is not identity. For ease of notation, we will also denote by $\text{diam}(M)= \text{diam}(\text{supp}(M))$ and $\text{diam}(\mathcal{M}) = \text{diam}(\text{supp}(\mathcal{M}))$. An operator $O$ is a geometrically local extensive operator with diameter $a$, coordination number $\mathcal{Z}$ and local strength $J$ if it can be expressed as
\begin{align}\label{eq:geom_local_op}
O = \sum_\alpha o_\alpha \text{ where }\forall \alpha: \text{diam}(\text{supp}(o_\alpha)) \leq a, \abs{\{\alpha' : o_{\alpha'} \cap o_\alpha \neq \emptyset\}} \leq \mathcal{Z} \text{ and } \norm{o_\alpha} \leq J
\end{align}
We will often call such operators $(a, \mathcal{Z}, J)-$geometrically local extensive operators. An operator $O$ will be called $(a, \mathcal{Z}, J)-$geometrically local extensive Hermitian operator if it satisfies Eq.~\ref{eq:geom_local_op} with $o_\alpha = o_\alpha^\dagger$. We remark that if $O_1, O_2 \dots O_M$ are individually $(a, \mathcal{Z},J)$-geometrically local extensive operators, then $[O_1, [O_2, \dots [O_{M - 1},O_M]]]$ is $(Ma, \mathcal{Z}^M, (2J)^M)$-geometrically local extensive operator. A useful result that we will use is the Lieb-Robinson bound.
\begin{lemma}[Lieb-Robinson bounds Ref.~\cite{hastings2006spectral}]\label{lemma:lieb_robinson}
    Suppose $H$ is a $(a,\mathcal{Z},J)$-geometrically local Hermitian operator, and $O_X$ be an operator supported on $X$ and $\mathcal{K}_Y$ be a super-operator supported on $Y$ which maps any operator to a trace-less operator (i.e., $\textnormal{Tr}(\mathcal{K}_Y(X)) = 0 \ \forall X$), then for any state $\Psi$,
    \[
    \textnormal{Tr}({e^{iHt}O_Xe^{-iHt}\mathcal{K}_Y(\Psi)}) \leq \frac{e}{\mathcal{Z}}\abs{X} \norm{O_X}\norm{\mathcal{K}_Y}_{1\to1} \exp\bigg(\frac{1}{a}\big(c_\textnormal{LR}t - d(X, Y)\big)\bigg),
    \]
    where $c_\textnormal{LR} = 4eJ\mathcal{Z}a$ is the Lieb-Robinson velocity corresponding to $H$.
\end{lemma}

We will find it useful to use a ``intensive norm" associated with a representation of an operator as a sum of local terms. Suppose $A = \sum_{\alpha} a_\alpha$ is an operator defined on the Hilbert space of the lattice model, then for $\kappa > 0$,
\begin{align}\label{eq:local_norm}
\smallnorm{A}_\star = \sup_x \sum_{\alpha: x \ni \text{supp}(a_\alpha)} \norm{a_\alpha}.
\end{align}
This norm is motivated from the local norms used in Refs.~\cite{abanin2017effective, abanin2017rigorous}, but does not account for the geometric locality of the individual terms in the operator. Note that this norm is dependent on the specific decomposition of $A$ as $\sum_{\alpha} a_\alpha$ --- for our purposes, since we will be working with only specific decompositions of the operators of interest, this ambiguity will not be a concern. The following lemma summarizes some useful properties of the local norms.
\begin{lemma}[Properties of $\norm{\cdot}_\star$]\label{lemma:local_norm_prop}
Suppose $A = \sum_\alpha a_\alpha$ and $B = \sum_\alpha b_\alpha$ are operators for the Lattice model in $\mathbb{Z}^d$, then
    \begin{enumerate}
        \item[(a)] $\norm{A+B}_\star \leq \norm{A}_\star +\norm{B}_\star$.
        \item[(b)] $\norm{[A, B]}_\star \leq 2 \norm{A}_\star \norm{B}_\star$.
        \item[(c)] Suppose $U$ is a locality-preserving unitary i.e.,$\exists r_0>0 : $ for any operator $A$, $\textnormal{supp}(U^\dagger A U) \subseteq \{x : d(x, \textnormal{supp}(A)) \leq r_0\}$, then $\smallnorm{UA U^\dagger}_\star \leq (2r_0)^d \smallnorm{A}_\star$.
    \end{enumerate}
\end{lemma}
\begin{proof} All the proofs rely on the basic definition of the local norm.

\noindent (a) Since $A + B = \sum_\alpha (a_\alpha+b_\alpha)$, from the definition of the local norm, we obtain that
    \begin{align*}
    \norm{A + B}_\star &= \sup_x \bigg(\sum_{\alpha: x \ni \text{supp}(a_\alpha)} \norm{a_\alpha}  + \sum_{\alpha: x \ni \text{supp}(b_\alpha)} \norm{b_\alpha}\bigg), \nonumber \\
    &\leq \sup_x \sum_{\alpha: x \ni\text{supp}(a_\alpha)}\norm{a_\alpha} + \sup_x \sum_{\alpha: x \ni\text{supp}(b_\alpha)}\norm{b_\alpha},\nonumber\\
    &\leq \norm{A}_\star + \norm{B}_\star.
    \end{align*}
(b) We first note that
\begin{align*}
[A, B] = \sum_{\alpha, \alpha'} [a_\alpha, b_{\alpha'}] = \sum_{\alpha, \alpha': \text{supp}(a_\alpha) \cap \text{supp}(b_{\alpha'}) \neq \emptyset} [a_\alpha, b_{\alpha'}],
\end{align*}
and consequently since $\text{diam}(A\cup B) \leq \text{diam}(A)+\text{diam}(B)$,
\begin{align*}
\norm{[A, B]}_\star &=\sup_x \sum_{\substack{\alpha, \alpha': \text{supp}(a_\alpha) \cap \text{supp}(b_{\alpha'}) \neq \emptyset, \\ x \ni \text{supp}(a_\alpha) \cup \text{supp}(b_{\alpha'})}} \norm{[a_\alpha,b_{\alpha'}]},\nonumber \\
&\leq 2 \sum_{\substack{\alpha, \alpha': \text{supp}(a_\alpha) \cap \text{supp}(b_{\alpha'}) \neq \emptyset, \\ x \ni \text{supp}(a_\alpha) \cup \text{supp}(b_{\alpha'})}} \norm{a_\alpha}\norm{b_{\alpha'}}, \nonumber \\
&\leq 2 \sum_{\alpha: x \ni \text{supp}(a_\alpha)} \sum_{\alpha': x \ni \text{supp}(a_{\alpha'})} \norm{a_\alpha}\norm{b_{\alpha'}} = 2\norm{A}_\star \norm{B}_\star.
\end{align*}
(c) Consider the operator $A' = U A U^\dagger = \sum_{\alpha} a_{\alpha}'$ where $a_\alpha' = U a_\alpha U^\dagger$. By definition, $\text{supp}(a_\alpha') \subseteq \{x: d(x, \text{supp}(a_\alpha)) \leq r_0\}$ and consequently, $\text{diam}(a_\alpha') \leq \text{diam}(a_\alpha) + 2r_0$ Now, 
 \begin{align*}
     \smallnorm{A'}_\star &= \sup_x \sum_{\alpha: x\ni \text{supp}(a_\alpha')} \norm{a_\alpha'}, \nonumber \\
     &\leq \sup_x \sum_{\alpha: d(x, \text{supp}(a_\alpha)) \leq r_0} \norm{a_\alpha}, \nonumber \\
     &\leq \sup_x \sum_{d(y, x) \leq r_0} \sum_{\alpha: y \ni \text{supp}(a_\alpha)} \norm{a_\alpha}, \nonumber \\
     &\leq \sup_x \sum_{d(y, x) \leq r_0} \smallnorm{A}_\star \leq (2r_0)^d e^{2\kappa r_0} \smallnorm{A}_\star.
 \end{align*}
\end{proof}
\noindent We remark that while considering super-operators, an intensive super-operator norm can be defined analogous to Eq.~\eqref{eq:local_norm} while using the diamond norm instead of the operator norm --- for super-operator $\mathcal{N}$, we will denote this norm by $\norm{\mathcal{N}}_{\diamond, \star}$. Since the diamond norm satisfies sub-multiplicity, triangle inequality and unitary invariance, lemma \ref{lemma:local_norm_prop} continues to apply for local super-operator norm as well.

A useful lemma that we will use in our stability analysis is below.
\begin{lemma}\label{lemma:lieb_robinson_perturbation_theory}
Suppose $H(t) = \sum_\alpha h_\alpha(t)$ is a $(a, \mathcal{Z}, J)-$geometrically local extensive operator, $\mathcal{N} = \sum_\alpha \mathcal{N}_\alpha$ is Lindbladian with $\norm{\mathcal{N}}_{\diamond,\star} \leq 1$ and $\mathcal{E}_\delta(t, s)$ is the channel defined by
\[
\mathcal{E}_\delta(t, s) = \mathcal{T}\exp\bigg(\int_s^t \mathcal{L}_\delta(s') ds'\bigg) \text{ where }\mathcal{L}_\delta(s) = -i[H(s), \cdot] + \delta \mathcal{N}.
\]
Then, for an operator $O_X$ supported on $X$ and any state $\Psi$,
\begin{align}
\abs{\textnormal{Tr}[O_X \mathcal{E}_\delta(t, s)(\Psi)] - \textnormal{Tr}[O_X \mathcal{E}_{\delta = 0}(t, s)(\Psi)]}\leq   \delta \norm{O_X}\abs{X}^2 \abs{t - s}\nu_d(c_\textnormal{LR}\abs{t - s}),
\end{align}
where $\nu_d(x)$ is a non-decreasing function of $x$ which also satisfies $\nu_d(x) \leq O(x^d)$ as $x \to \infty$ and $c_\textnormal{LR}$ is the Lieb-Robinson velocity from lemma \ref{lemma:lieb_robinson}.
\end{lemma}
\begin{proof}
    We will denote by $S_\alpha = \text{supp}(h_\alpha)$. We begin by noting that for any state $\Psi$,
    \begin{align}\label{eq:pert_theory_order_1}
    \abs{\text{Tr}[O_X \mathcal{E}_\delta(t, s)(\Psi)] - \text{Tr}[O_X \mathcal{E}_{\delta = 0}(t, s)(\Psi)]} \leq \delta \sum_\alpha \int_s^t \bigabs{\text{Tr}[O_X\mathcal{E}_{\delta=0}(t, s')\mathcal{N}_\alpha \mathcal{E}_\delta(s', s)(\Psi)] } ds'.
    \end{align}
    Furthermore, using lemma \ref{lemma:lieb_robinson}, we obtain that
    \begin{align}\label{eq:lr_upper_bound}
   \bigabs{\text{Tr}[O_X\mathcal{E}_{\delta=0}(t, s')\mathcal{N}_\alpha \mathcal{E}_\delta(s', s)(\Psi)] } \leq e\norm{\mathcal{N}_\alpha}_\diamond\norm{O_X}\abs{X}f_{c_\text{LR}(t - s')}(d(X, S_\alpha)), \text{ where }
    f_{x}(n) = \begin{cases}
    1 & \text{ if } n \leq x, \\
    \exp((x - n)/a), & \text{ otherwise,}
    \end{cases}
    \end{align}
    where $c_\text{LR}$ is the Lieb-Robinson velocity in Eq.~\eqref{lemma:lieb_robinson}. Using Eq.~\eqref{eq:pert_theory_order_1} together with Eq.~\eqref{eq:lr_upper_bound}, we obtain that
    \begin{align}\label{eq:sim_sum_1}
        \abs{\text{Tr}[O_X \mathcal{E}_\delta(t, s)(\Psi)] - \text{Tr}[O_X \mathcal{E}_{\delta = 0}(t, s)(\Psi)]} &\leq e \delta \smallnorm{O_X}\abs{X} \int_{s}^t ds' \bigg(\sum_{\alpha} \smallnorm{\mathcal{N}_\alpha}_\diamond f_{c_\text{LR}(t-s')}(d(X, S_\alpha)) \bigg), \nonumber\\ 
        &\leq e \delta \smallnorm{O_X}\abs{X}\int_s^t ds'\bigg(\sum_{n = 0}^\infty \bigg(\sum_{\alpha : d(S_\alpha, X) = n}\norm{\mathcal{N}_\alpha}_\diamond \bigg)f_{c_\text{LR}(t - s')}(n)\bigg).
    \end{align}
    Next, we note that for any $n \in \{0, 1, 2 \dots \}$,
    \begin{align}\label{eq:sim_sum_2}
    \sum_{\alpha : d(S_\alpha, X) = n} \norm{\mathcal{N}_\alpha}_\diamond \leq \sum_{\substack{x \in X \\ y: d(x, y) = n}} \sum_{\alpha: S_\alpha \ni y} \norm{\mathcal{N}_\alpha}_\diamond \leq\norm{\mathcal{N}}_{\diamond, \star}  \sum_{\substack{x \in X \\ y: d(x, y) = n}} 1 \leq  \abs{X} \sum_{y:d(0, y) = n} 1.
    \end{align}
    Furthermore, note that
    \begin{align}\label{eq:sim_sum_3}
        \sum_{y:d(0, y) = n } 1 = \bigabs{\bigg\{y \in \mathbb{Z}^d : \sum_{i = 1}^d \abs{y_i} = n \bigg\}} \leq 2^d \bigabs{\bigg\{y \in \mathbb{Z}^d : y_i \geq 0 \text{ and } \sum_{i = 1}^d {y_i} = n \bigg\}} = 2^d {n + d - 1 \choose d - 1} \leq \frac{2^d}{(d - 1)!}(n + d - 1)^{d - 1}.
    \end{align}
    Combining Eqs.~\eqref{eq:sim_sum_1}, \eqref{eq:sim_sum_2} and \eqref{eq:sim_sum_3}, we obtain that
    \begin{align}\label{eq:final_error_bound}
   \abs{\text{Tr}[O_X \mathcal{E}_\delta(t, s)(\Psi)] - \text{Tr}[O_X \mathcal{E}_{\delta = 0}(t, s)(\Psi)]} \leq \norm{O_X} \abs{X}^2\norm{\mathcal{N}}_\kappa \int_{s}^t \nu_d(c_\text{LR}(t - s'))ds',
    \end{align}
    where
    \[
    \nu_d(x) =\frac{2^{d + 1}e}{(d - 1)!} \sum_{n = 0}^\infty (n + d - 1)^{d - 1} f_x(n).
    \]
    We observe that, for any fixed $n$, $f_x(n)$ is a non-decreasing function of $x$ and consequently so is $\nu_d(x)$. Using this to upper bound $\int_s^t \nu_d(c_\text{LR}(t - s'))ds' \leq (t - s) \nu_d(c_\text{LR}(t - s))$ in Eq.~\eqref{eq:final_error_bound}, we obtain the lemma statement.
\end{proof}
\section{Noise robustness of trotter formulae}
We consider a target geometrically local Hamiltonian $H_0 = \sum_\alpha h_\alpha$ on a $d$-dimensional lattice $\mathbb{Z}^d$ and a local observable $O_X$ supported on region $X \subseteq \mathbb{Z}^d$. We will assume that $H$ is a $(a_0, \mathcal{Z}, J)-$geometrically local intensive Hamiltonian i.e., for all $\alpha$, $\norm{h_\alpha}\leq J, \text{diam}(h_\alpha) \leq a_0$ and $\abs{\{\alpha': \text{supp}(h_\alpha) \cap \text{supp}(h_{\alpha'})\}} \leq \mathcal{Z}$. In general, we will assume that we can split $H$ as
\[
H_0 = \sum_{j =1}^K H_j \text{ where }H_j = \sum_{\alpha \in\mathcal{A}_j} h_\alpha,
\]
where $K$ is independent of the system size and $[h_\alpha, h_\beta] = 0$ if $\alpha, \beta \in \mathcal{A}_j$ for some $j \in \{1, 2 \dots K\}$ i.e., each of the Hamiltonians $H_j$ are themselves sum of commuting terms. For example, in 1D, a nearest neighbour Hamiltonian $H_0 = \sum_x h_{x, x + 1}$ can be decomposed into $H_0 = H_e + H_o$, where
\[
H_o = \sum_{x} h_{2x + 1, 2x + 2} \text{ and } H_e = \sum_x h_{2x, 2x + 1}.
\]
Similarly, a nearest neighbour Hamiltonian in 2D, $H_0 = \sum_{x, y}\big( h_{x, y; x + 1, y} + h_{x, y; x, y + 1}\big)$ can be decomposed as $H = H_{e}^{X} + H_o^{X} + H_e^{Y} + H_o^{Y}$ where
\[
H_o^X = \sum_{x, y} h_{2x + 1, y; 2x + 2, y}, H_e^X = \sum_{x, y} h_{2x, y; 2x + 1, y}, H_o^Y = \sum_{x, y} h_{x, 2y + 1; x, 2y + 2} \text{ and }H_e^Y = \sum_{x, y} h_{x, 2y; x, 2y + 1}.
\]

\noindent Following the formalism developed in Ref.~\cite{childs2019nearly} but with only slight modification to account for the fact that $H$ is written as a sum of more than 2 operators, we next recap the definition of a trotter formula.
\begin{definition}\label{def:trotter_formula}
    $\{r_i \in \{1, 2 \dots K\}\}_{i \in [1:M]}$ and $\{x_i \in \mathbb{R}\}_{i \in[1:K]}$ specify a $p^\text{th}$ order trotter formula $S(\tau) = \prod_{j = 1}^K \exp(-ix_j  H_{r_j} \tau)$ for the exponential $\exp(-i (H_1 + H_2 \dots H_K) \tau)$ if for all hermitian operators $H_1, H_2 \dots H_K$, 
    \[
\frac{d^k}{d\tau^k} S(\tau)\bigg |_{\tau = 0} - \bigg(-i \sum_{j = 1}^K H_j\bigg)^k = 0 \ \forall \ k \in \{0, 1 \dots p\}.
    \]
\end{definition}
A very useful technical tool that was developed in Ref.~\cite{childs2019nearly} was the local error sum representation, which we recap in the following lemma.
\begin{lemma}[Local-sum representation of Trotter Error, Ref.~\cite{childs2019nearly}]\label{lemma:local_error_sum}
    Suppose $\{\alpha_i \in [1:K]\}_{i \in [1:M]}$ and $\{x_i \in \mathbb{R}\}_{i \in[1:K]}$ specify a $p^\text{th}$ order trotter formula $\mathcal{S}(\tau) = \prod_{j = 1}^K \exp(-ix_j  H_{r_j} \tau)$ for the exponential of $H = \sum_{j = 1}^K H_j$, then
    \[
    S(\tau) - \exp(-iH \tau) =  \frac{1}{(p - 1)!} \int_{0 \leq \tau_2 \leq \tau_1 \leq \tau}\exp(-iH(\tau - \tau_1)) J^{(p)}(\tau_2) (\tau_1 - \tau_2)^{p - 1} d\tau_1 d\tau_2,
    \]
    where
    \begin{align}\label{eq:expression_J}
    J(t) = -i\sum_{j = 2}^M x_j \bigg[\bigg(\prod_{j' = 1}^{j - 1} \exp(-i x_{j'} H_{\alpha_{j'}} t)\bigg) H_{\alpha_j}\bigg(\prod_{j' = 1}^{j - 1} \exp(-i x_{j'}H_{\alpha_{j'}} t)\bigg)^\dagger - H_{\alpha_j}\bigg].
    \end{align}
\end{lemma}
\begin{proof}
We begin by defining
\begin{align}\label{eq:def_J_deriv}
\mathcal{J}(\tau)  = \bigg(\frac{d}{d\tau} \mathcal{S}(\tau)\bigg)\mathcal{S}^\dagger(\tau) + iH.
\end{align}
Note that, from the definition of the Trotter formula [Definition \ref{def:trotter_formula}], it follows that $H = \sum_{j = 1}^K x_j H_{\alpha_j}$. Using this together with Eq.~\ref{eq:def_J_deriv}, we then obtain the expression in Eq.~\ref{eq:expression_J}. Next, we note that since by definition $S^{(k)}(0) = (-iH)^k$ for $k \in \{0, 1, 2 \dots p\}$, $\mathcal{R}(\tau) = \mathcal{S}'(\tau) + iH \mathcal{S}(\tau)$ satisfies $\mathcal{R}^{(k)}(0) = 0$ for $k \in \{0, 1, 2 \dots p- 1\}$ and since $\mathcal{J}(\tau) = \mathcal{R}(\tau) \mathcal{S}^\dagger(\tau)$, it also follows that $\mathcal{J}^{(k)}(0) = 0$ for $k \in \{0, 1, 2 \dots p - 1\}$. Therefore, from the Taylor's remainder theorem, it follows that
\[
\mathcal{J}(\tau) = \frac{1}{(p - 1)!} \int_0^\tau \mathcal{J}^{(p)}(s) (\tau - s)^{p - 1} ds. 
\]
Finally, integrating Eq.~\ref{eq:def_J_deriv}, we then obtain that
\begin{align}
\mathcal{S}(\tau) - \exp(-iH \tau) &= \int_0^{\tau}  \exp(-iH(\tau - \tau_1)) \mathcal{J}(\tau_1) \mathcal{S}(\tau_1) d\tau_1, \nonumber\\
 &= \frac{1}{(p - 1)!} \int_{0 \leq \tau_2 \leq \tau_1 \leq \tau} \exp(-iH(\tau - \tau_1)) \mathcal{J}^{(p)}(\tau_2) (\tau_1 - \tau_2)^{p - 1} d\tau_1 d\tau_2,
\end{align}
which proves the lemma.
\end{proof}
Using lemma \ref{lemma:local_error_sum} together with the Lieb-Robinson bounds (lemma \ref{lemma:lieb_robinson}), we can establish the following lemma which upper bounds the trotterization error for local observables. Previously, Ref.~\cite{childs2019nearly} established that, for dynamics of geometrically local models on $N$ qubits and for evolution time $\tau$, a $p^\text{th}$ order trotter formula approximates the full many-body state obtained in the geometrically local model trace-norm error $O(N\tau \varepsilon^p)$. In the lemma below we show that the error in local observables, instead of the full many-body state, can be upper bounded by $O(\tau^{d + 1}\varepsilon^p)$ which is uniform in the system size $N$.

\begin{replemma}{lemma:trotter_local_obs_main}[Trotter error for local observables, Formal]\label{lemma:trotter_local_obs}
 If $S(\varepsilon)$ is a $p^\text{th}$ order trotter formula (Definition \ref{def:trotter_formula}), then for an observable $O_X$ supported on $X$ and any initial state $\rho(0)$, the $p^\textnormal{th}$ trotterization error $\Delta^{(p)} = \bigabs{\textnormal{Tr}[O_X S^T(\varepsilon)\rho(0)S^{\dagger T}(\varepsilon)] - \textnormal{Tr}[O_X \exp(-iH\tau) \rho(0) \exp(iH\tau)]}$ satisfies
  \[
   \Delta^{(p)} \leq \frac{a_0^{(d - 1)(p - 1)}}{\mathcal{Z}^2} C^{(p)}\abs{X}^2 \norm{O_X} c_\textnormal{LR}\tau \nu_d(c_\textnormal{LR}\tau) (c_\textnormal{LR}\varepsilon)^{p} \leq O(\tau^{d+1} \varepsilon^{p}),
  \]
  where $\varepsilon$ is the trotter time-step, $T = \tau / \varepsilon$ is the total number of time-steps in the trotterized evolution, $c_\textnormal{LR}$ is the Lieb-Robinson velocity appearing in lemma \ref{lemma:lieb_robinson}, $C^{(p)}$ is a constant that depends only on the trotter formula but is independent of the system-size $N$, evolution time $\tau$ and $\varepsilon$ and $\nu_d(l) \leq O(l^d)$ as $l \to \infty$.
\end{replemma}
\begin{proof}
    We will combine Lieb-Robinson bounds with the local error sum representation. Denoting by $\mathcal{S}(\varepsilon) = S(\varepsilon) (\cdot) S^\dagger(\varepsilon) \equiv S_l(\varepsilon)S_r^\dagger(\varepsilon)$, we note that $\mathcal{S}(\varepsilon)$ is simply a trotterization of unitary $\exp(-i\mathcal{C}_H \varepsilon)$. Consequently, we can apply lemma \ref{lemma:local_error_sum} in the vectorized notation. We begin with the telescoping sum 
    \begin{align}\label{lemma:telescope_sum}
    &\text{Tr}(O_X \exp(-iH\tau)\rho(0)\exp(-iH\tau)) - \text{Tr}(O_X S^T(\varepsilon)\rho(0)S^{\dagger T}(\varepsilon)) \nonumber\\
    &\qquad =\vecbra{O_X} \big(\exp(-i\mathcal{C}_H \tau) - \mathcal{S}^T(\varepsilon)\big) \vecket{\rho(0)} = \sum_{n = 1}^T \vecbra{O_X}\exp(-i(n - 1)\varepsilon \mathcal{C}_H ) \big(\exp(-i\varepsilon \mathcal{C}_H \big) - \mathcal{S}(\varepsilon)\big) \mathcal{S}^{T - n - 1}(\varepsilon) \vecket{\rho}
    \end{align}
    We now use lemma \ref{lemma:local_error_sum} to obtain
    \begin{align}
    &\vecbra{O}\big(\exp(-i\mathcal{C}_Ht) - \mathcal{S}^T(\varepsilon)\big)\vecket{\rho}  \nonumber\\
    &\qquad = \frac{1}{(p - 1)!}\sum_{n = 1}^T \int_0^{\varepsilon} \int_0^{\tau_1} (\tau_1 - \tau_2)^{p - 1}\vecbra{O_X}\exp(-i\mathcal{C}_{H} (n \varepsilon - \tau_1))  \mathcal{J}^{(p)}(\tau_2)\mathcal{S}(\tau_1) \mathcal{S}^{T - n - 1}(\varepsilon) \vecket{\rho}d\tau ds,
    \end{align}
    from which we obtain that
    \begin{align}\label{eq:observable_remainder_ub}
        &\abs{\vecbra{O_X}\exp(-i\mathcal{C}_Ht) - \mathcal{S}^T(\varepsilon)\vecket{\rho(0)}} \leq 
        \frac{1}{(p - 1)!}\sum_{n = 1}^T \int_0^\varepsilon \int_0^{\tau_1} (\tau_1 - \tau_2)^{p - 1} \norm{\vecbra{O_X} \exp(-i\mathcal{C}_H(n\varepsilon - \tau_1)) \mathcal{J}^{(p)}(\tau_2)} d\tau_1 d\tau_2.
    \end{align}
    Next, we observe that $\mathcal{J}^{(p)}$ can itself be expressed as a sum of local commutators. To see this, we first rewrite $\mathcal{J}(t)$ in the form
    \[
    \mathcal{J}(t) = -i\sum_{j = 2}^M x_j \textnormal{Ad}_{\exp(-ix_{1} \mathcal{C}_{H_{r_{1}}}t)} \textnormal{Ad}_{\exp(-ix_{2} \mathcal{C}_{H_{r_{2}}}t)} \dots \textnormal{Ad}_{\exp(-ix_{j - 1} \mathcal{C}_{H_{r_{j - 1}}}t)} (\mathcal{C}_{H_{r_j}}),
    \]
    where, for a superoperator $\mathcal{U}$, we define $\text{Ad}_\mathcal{U}(\mathcal{X}) = \mathcal{U}\mathcal{X} \mathcal{U}^\dagger$. We then obtain that
    \begin{align}\label{eq:deriv_J_t_adj}
    &\frac{d^p}{dt^p} \mathcal{J}(t) = (-i)^{p + 1}\sum_{j = 2}^M x_j \sum_{\substack{n_1, n_2 \dots n_{j - 1} \\ n_1 + n_2 \dots n_{j- 1} = p} }\frac{(j - 1)!}{n_1 ! n_2 ! \dots n_{j - 1}!} x_1^{n_1} x_2^{n_2} \dots x_{n_{j - 1}}^{n_{j - 1}}\times \nonumber\\
    &\qquad \qquad \bigg(\text{ad}_{\mathcal{C}_{H_{\alpha_1}}}^{n_1} \textnormal{Ad}_{\exp(-ix_{1} \mathcal{C}_{H_{r_{1}}}t)}\text{ad}_{\mathcal{C}_{H_{r_1}}}^{n_2}\textnormal{Ad}_{\exp(-ix_{2} \mathcal{C}_{H_{r_{2}}}t)} \dots \text{ad}^{n_{j - 1}}_{\mathcal{C}_{H_{r_{j-1}}}}\textnormal{Ad}_{\exp(-ix_{j - 1} \mathcal{C}_{H_{r_{j - 1}}}t)} (\mathcal{C}_{H_{r_j}})\bigg),
    \end{align}
    where, for a superoperators $\mathcal{X}$ and $\mathcal{Y}$, we define $\text{ad}_{\mathcal{Y}}(\mathcal{X}) = [\mathcal{Y}, \mathcal{X}]$. We next recall that each of the Hamiltonians $H_{\alpha}$ are themselves sums of geometrically local terms with disjoint support: $H_\alpha = \sum_{\beta \in \mathcal{A}_\alpha} h_{\beta}$. We can then express
    \begin{align*}
        \text{ad}_{\mathcal{C}_{H_{r_1}}}^{n_1} \textnormal{Ad}_{\exp(-ix_{1} \mathcal{C}_{H_{r_{1}}}t)}\text{ad}_{\mathcal{C}_{H_{r_1}}}^{n_2}\textnormal{Ad}_{\exp(-ix_{2} \mathcal{C}_{H_{r_{2}}}t)} \dots \text{ad}^{n_{j - 1}}_{\mathcal{C}_{H_{r_{j-1}}}}\textnormal{Ad}_{\exp(-ix_{j - 1} \mathcal{C}_{H_{r_{j - 1}}}t)} (\mathcal{C}_{H_{r_j}}) = \sum_{\beta \in \mathcal{A}_{\alpha_j}} \mathcal{N}_{\beta},
    \end{align*}
    where
    \begin{align}\label{eq:def_n_beta}
    \mathcal{N}_{\beta} =  \text{ad}_{\mathcal{C}_{H_{r_1}}}^{n_1} \textnormal{Ad}_{\exp(-ix_{1} \mathcal{C}_{H_{r_{1}}}t)}\text{ad}_{\mathcal{C}_{H_{r_1}}}^{n_2}\textnormal{Ad}_{\exp(-ix_{2} \mathcal{C}_{H_{r_{2}}}t)} \dots \text{ad}^{n_{j - 1}}_{\mathcal{C}_{H_{r_{j-1}}}}\textnormal{Ad}_{\exp(-ix_{j - 1} \mathcal{C}_{H_{r_{j - 1}}}t)} (\mathcal{C}_{h_{\beta}}) 
    \end{align}
    Now, we observe that $\mathcal{N}_\beta$ is a super-operator whose image contains only traceless operators i.e.~for any operator $X$, $\textnormal{Tr}(\mathcal{N}_\beta (X)) = 0$ or equivalently, in the vectorized notation, $\vecbra{I}\mathcal{N}_\beta = 0$. To see this, we note that for any superoperator $\mathcal{N}$ and operator $V$, if $\vecbra{I}\mathcal{N} = 0$ then $\vecbra{I}\text{Ad}_{\exp(-i\mathcal{C}_V)}(\mathcal{N}) = \vecbra{I} \exp(-i\mathcal{C}_V) \mathcal{N} \exp(i\mathcal{C}_V) = \vecbra{I}\mathcal{N}\exp(i\mathcal{C}_V) = 0$. Furthermore, for any superoperator $\mathcal{N}$ and operator $V$, if $\vecbra{I}\mathcal{N} = 0$ then $\vecbra{I}\text{ad}_{\mathcal{C}_V}(\mathcal{N}) = \vecbra{I}\mathcal{C}_V \mathcal{N} - \vecbra{I}\mathcal{N}\mathcal{C}_V = 0$. Since $\mathcal{N}_{\beta}$ is obtained from $\mathcal{C}_{h_\beta}$, which satisfies $\mathcal{C}_{h_\beta}\vecket{I} = 0$, on applying a sequence of operations of the form $\text{Ad}_{\exp(-i \mathcal{C}_V)}$ or $\text{ad}_{\mathcal{C}_V}$, we obtain that $\vecbra{I}\mathcal{N}_\beta = 0$.
    
    Next, we observe that $\mathcal{N}_{\beta}$ is a super-operator that is geometrically local. To see this, note that $\mathcal{C}_{h_\beta}$ is a geometrically local super-operator. Furthermore, since $H_{\alpha_k}$ are themselves sum of geometrically local and commuting terms, $\text{Ad}_{\exp(-ix_k \mathcal{C}_{H_{\alpha_k}})}(\mathcal{X})$ and $\text{ad}_{\mathcal{C}_{H_{\alpha_k}}}(\mathcal{X})$ are both geometrically local super-operator if $\mathcal{X}$ is a geometrically local super-operator. In particular, if the diameter of the support of local terms in $H_{\alpha_k}$ is at-most $a_0$, then the diameter of the support of $\text{Ad}_{\exp(-ix_k \mathcal{C}_{H_{\alpha_k}})}(\mathcal{X})$ and $\text{ad}_{\mathcal{C}_{H_{\alpha_k}}}(\mathcal{X})$ is at-most $\text{diam}(X) + 2a_0$. Consequently, we obtain that
    \begin{align}\label{eq:diameter_recursion}
    &\text{diam}\bigg(\bigg(\prod_{q= k}^{j - 1}\text{ad}_{\mathcal{C}_{H_{\alpha_{q}}}}^{n_q} \text{Ad}_{\exp(-ix_q \mathcal{C}_{H_{\alpha_q}}t)}\bigg)\mathcal{C}_{h_\beta}\bigg) \leq (2(j - k) + 1)a_0 + 2a_0 \sum_{l = k}^{j - 1} n_l \leq (2p + 2j - 1)a_0,
    \end{align}
    where we have used the fact that, as per Eq.~\ref{eq:deriv_J_t_adj}, $n_1 + n_2 + \dots + n_{j - 1} = p$.
    Therefore, $\mathcal{N}_\beta$ defined in Eq.~\ref{eq:def_n_beta} is a geometrically local operator whose support satisfies $\textnormal{diam}(\mathcal{N}_\beta) \leq (2p +2j - 1)a_0$. Finally, we can also estimate the norm $\smallnorm{\mathcal{N_\beta}}_\diamond$: we note that an application of $\text{Ad}_{\exp(-ix_{k}\mathcal{C}_{H_{\alpha_k}}t)}$ on a super-operator does not change its diamond norm. Furthermore, since by assumption every local term in the Hamiltonian, $h_\alpha$, commutes with all but $\mathcal{Z}$ local terms and $\norm{h_\alpha}\leq J$, we obtain that for any geometrically local super-operator $\mathcal{X}$, $\smallnorm{\text{ad}_{\mathcal{C}_{H_{\alpha_k}}}(\mathcal{X})}_\diamond \leq 4J\mathcal{Z}\abs{\text{supp}(\mathcal{X})} \leq 4J\mathcal{Z}(\text{diam}(\mathcal{X}))^d \norm{\mathcal{X}}_\diamond$. Using this together with Eq.~\eqref{eq:diameter_recursion}, we obtain that
    \begin{align}\label{eq:n_beta_norm}
        \norm{\mathcal{N}_\beta}_\diamond \leq 2J(4J\mathcal{Z}a_0^d (2p + 2j - 1)^{d})^{n_1 + n_2 + \dots n_{j - 1}} \leq  2J(4J\mathcal{Z} a_0^d)^p (2p+ 2j - 1)^{pd}.
    \end{align}
    Next, from Eq.~\eqref{eq:observable_remainder_ub}, we obtain that 
    \begin{align}\label{eq:trotter_formula_manipulation}
        &\abs{\vecbra{O}\big(\exp(-i\mathcal{C}_H t) - \mathcal{S}^T(\varepsilon)\big) \vecket{\rho}}\nonumber\\
        &\qquad \leq \sum_{n = 1}^T \sum_{j = 2}^M \sum_{\substack{n_1, n_2 \dots n_{j - 1} \geq 0 \\ n_1 + n_2 \dots n_{j - 1} = p}} \sum_{\beta}\frac{(j - 1)!}{n_1 ! n_2 ! \dots n_{j - 1}!}\abs{x_j} \abs{x_1}^{n_1} \abs{x_2}^{n_2} \dots \abs{x_{n_{j - 1}}}^{n_{j - 1}} \times \nonumber\\
        &\qquad \qquad \qquad \qquad \qquad  \frac{1}{(p - 1)!} \int_0^\varepsilon \int_0^{\tau_1} (\tau_1 - \tau_2)^{p - 1} \norm{\vecbra{O} \exp(-i\mathcal{C}_H(n\varepsilon - \tau_1)) \mathcal{N}_\beta}d\tau_1 d\tau_2,\nonumber \\
        &\qquad\leq \frac{a_0^{(d - 1)p - 1}\abs{X}\norm{O_X}}{2e^p \mathcal{Z}^2 (p + 1)!}(c_\text{LR}\varepsilon)^{p + 1}\sum_{n = 1}^T \sum_{j = 2}^M \sum_{\substack{n_1, n_2 \dots n_{j - 1} \geq 0 \\ n_1 + n_2 \dots n_{j - 1} = p}}\frac{(j - 1)!}{n_1 ! n_2 ! \dots n_{j - 1}!}\abs{x_j} \abs{x_1}^{n_1} \abs{x_2}^{n_2} \dots \abs{x_{n_{j - 1}}}^{n_{j - 1}}(2p+ 2j - 1)^{pd} \times \nonumber\\
        &\qquad \qquad \qquad \qquad \qquad \qquad  \sum_{\beta \in \mathcal{A}_{r_j}}  \max\bigg(1, \exp\bigg(\frac{1}{a}(c_\text{LR} T \varepsilon - d(X, \text{supp}(\mathcal{N}_\beta)) \bigg)\bigg),
    \end{align}
    where we have used lemma \ref{lemma:lieb_robinson} and the fact that $n\varepsilon - \tau_1 \leq T\varepsilon $. We note that from the reverse triangle inequality that
    \begin{align}\label{eq:distance_relation}
    d(X, \text{supp}\big(\mathcal{N}_{\beta}\big)) \geq d(X, \text{supp}(h_\beta)) - (n_1 + n_2 \dots n_{j - 1} + j - 1) a_0 = d(X, \text{supp}(h_\beta)) - (p + j - 1) a_0,
    \end{align}
    and consequently
    \begin{align}\label{eq:lr_sum_trotter_first}
        \sum_{\beta \in \mathcal{A}_{r_j}}  \max\bigg(1, \exp\bigg(\frac{1}{a_0}(c_\text{LR} t - d(X, \text{supp}(\mathcal{N}_\beta)) \bigg)\bigg) &\leq e^{(p - j + 1)} {\sum_\beta  \max\bigg(1, \exp\bigg(\frac{1}{a_0}(c_\text{LR} T\varepsilon - d(X, \text{supp}(h_\beta)) \bigg)\bigg)},\nonumber \\
        &\leq e^{p -j + 1} \sum_{n = 0}^{\infty} \max(1, \exp(c_\text{LR} T\varepsilon - n))\abs{\{\beta: d(X, \text{supp}(h_\beta)) = n\}} .
    \end{align}
    Note that
    \begin{align}
        \abs{\{\beta : d(X, \text{supp}(h_\beta)) = n\}} \leq \abs{X} \abs{\text{supp}(h_\beta)} \abs{\{x : d(x, 0) = n\}} \leq a_0^d \abs{X} \frac{2^d}{(d - 1)!} (n  + d - 1)^{d - 1},
    \end{align}
    where, in the last step, we have used Eq.~\eqref{eq:sim_sum_3} and the fact that $\abs{\text{supp}(h_\beta)} \leq (\text{diam}(h_\beta))^d \leq a_0^d$. From Eq.~\eqref{eq:lr_sum_trotter_first}, we then obtain
    \begin{align}\label{eq:lr_sum_trotter}
        \sum_{\beta \in \mathcal{A}_{r_j}}  \max\bigg(1, \exp\bigg(\frac{1}{a_0}(c_\text{LR} t - d(X, \text{supp}(\mathcal{N}_\beta)) \bigg)\bigg) \leq e^{p - j + 1} a_0^d \abs{X}^2 \underbrace{\frac{2^d}{(d - 1)!} \sum_{n = 0}^\infty (n + d - 1)^{d - 1}\text{max}(1, \exp(c_\text{LR}T\varepsilon - n))}_{\nu_d(c_\text{LR}T\varepsilon)}
    \end{align}
    where it can be noted that $\nu_d(l) \leq O(l^d)$. Returning to Eq.~\eqref{eq:observable_remainder_ub} and using Eq.~\eqref{eq:lr_sum_trotter}, we then obtain that
    \begin{align}\label{eq:trotter_final}
        \abs{\vecbra{O}\big(\exp(-i\mathcal{C}_H t) - \mathcal{S}^T(\varepsilon)\big) \vecket{\rho}} \leq \frac{a_0^{(d - 1)(p - 1)}}{\mathcal{Z}^2} C^{(p)}\abs{X}^2\norm{O_X} \nu_d(c_\text{LR} T \varepsilon) (c_\text{LR}\varepsilon)^{p + 1} T,
    \end{align}
    where
    \begin{align*}
        C^{(p)} = \frac{e^{-p}}{(p + 1)!} \sum_{j = 2}^M \sum_{\substack{n_1, n_2 \dots n_{j - 1} \geq 0 \\ n_1 + n_2 \dots n_{j - 1} = p}}\frac{(j - 1)!}{n_1 ! n_2 ! \dots n_{j - 1}!}\abs{x_j} \abs{x_1}^{n_1} \abs{x_2}^{n_2} \dots \abs{x_{n_{j - 1}}}^{n_{j - 1}}(2p+ 2j - 1)^{pd} e^{p - j + 1},
    \end{align*}
    is a constant independent of system size $N$, the trotter time-step $\varepsilon$ or the total time $t$ that depends only on the trotter formula. Setting $T\varepsilon = t$ in Eq.~\eqref{eq:trotter_final}, we obtain the lemma statement.
\end{proof}
\noindent \emph{Proof of proposition 1}. Finally, we consider the noisy setting --- given target Hamiltonian $H_0$ to be simulated for a time $\tau$, we trotterize it to effectively obtain a time-dependent Hamiltonian $H(t)$ which is then implemented on the simulator. As described in the main text, in the presence of noise, we model the dynamics of the simulator with a master equation
\[
\frac{d}{dt} \rho(t) = \mathcal{L}_\gamma(t) \rho(t) \text{ where } \mathcal{L}_\gamma(t) = -i[H(t), \cdot ] + \delta \sum_\alpha \mathcal{N}_\alpha,
\]
where $\mathcal{N}_\alpha$ are geometrically local Lindbladians which are assumed to satisfy $\norm{\mathcal{N_\alpha}}_\diamond \leq 1$. The simulator will run for a time $t_\text{sim}$ --- when using the $p^\text{th}$ order Trotter formula with trotter step $\varepsilon$, the simulator time $t_\text{sim} = \Theta(\tau / \varepsilon)$ for any fixed $p$. We will consider a local observable $O$, and estimate the error between the simulated observable $\mathcal{O}_\text{sim} = \text{Tr}(O \mathcal{E}_\gamma(t_\text{sim}, 0)\rho(0))$, where $\mathcal{E}_\gamma(t, s)= \mathcal{T}\exp(\int_s^{t} \mathcal{L}_\gamma(s') ds')$ and the target observable $\mathcal{O}_\text{target} = \text{Tr}(O\exp(-iH\tau)\rho(0)\exp(iH\tau))$. We begin by using the triangle inequality to obtain
    \[
    \abs{\mathcal{O}_\text{sim}- \mathcal{O}_\text{target}} \leq \abs{\mathcal{O}_\text{sim}- \mathcal{O}_\text{sim}|_{\gamma = 0}} + \abs{\mathcal{O}_{\text{sim}}|_{\gamma = 0} - \mathcal{O}_\text{target}}.
    \]
    Note that $\abs{\mathcal{O}_\text{sim}|_{\gamma = 0}- \mathcal{O}_\text{target}}$ is simply the trotterization error, which has been upper bounded by $O(\tau^{d+1}\varepsilon^p)$ in lemma \ref{lemma:trotter_local_obs}. Since the simulator Hamiltonian $H(t)$ has an $O(1)$ Lieb-Robinson velocity by assumption, we use lemma \ref{lemma:lieb_robinson_perturbation_theory} to obtain $\abs{\mathcal{O}_\text{sim}|_{\gamma = 0} - \mathcal{O}(0)} \leq O(\gamma \tau^{d + 1} / \varepsilon^{d + 1})$. The total error $\abs{\mathcal{O}_\text{sim}|_{\gamma = 0} - \mathcal{O}_\text{target}}$ is then upper bounded by
    \[
    \abs{\mathcal{O}_\text{sim}- \mathcal{O}_\text{target}} \leq O(\tau^{d + 1}\varepsilon^p) + O\bigg(\tau^{d + 1}\frac{\gamma}{\varepsilon^{d + 1}}\bigg).
    \]
    Clearly, a choice of $\varepsilon = \Theta(\gamma^{1 / (p + d+ 1)})$ optimizes this upper bound as $\gamma, \varepsilon \to 0$, which yields
    \[
    \abs{\mathcal{O}_\text{sim} - \mathcal{O}_\text{target}} \leq O\big(\tau^{d + 1}\gamma^{p/(p + d + 1)}\big),
    \]
    proving the proposition.\(\hfill\) $\square$

\subsection{Lower fault-tolerance overhead for local observables}
Here, we show that local observables can be simulated on a fault tolerant quantum computer with an overhead that scales slower compared to when we want to simulate the full state. We begin with recalling some basics of fault tolerance with concatenated codes that can be found in standard quantum error correction literature \cite{gottesman2024surviving}.

We will restrict ourselves here to a special model of noise called stochastic noise \cite{gottesman2024surviving}. Consider a quantum circuit formed with single and two qubit gates, single qubit state preparation and measurements where we are finally considering the probability distribution of the measured bit-strings $q(x)$. To describe the action of stochastic noise, we will introduce a fault path, where certain locations in the circuit (single qubit gate, two qubit gate, single qubit state preparation and single qubit measurements) are faulty i.e., in the fault path, instead of the correct operation being performed in the faulty locations, an incorrect operation is being performed (with possible interaction with some set of environment qubits). The noisy probability distribution $q_\text{noisy}(x)$ is expressible as
\[
q_\text{noisy}(x) = \sum_F p(F) q(x | F),
\]
where $q(x|F)$ is the probability distribution of the output bits measured when the fault path $F$ acts and $p(F)$ is the probability of the fault path. We will assume that there is a small constant $\xi$, which will act as the gate failure probability, such that
\begin{align}\label{eq:noise_model}
\forall F: \sum_{F': F \subseteq F' } p(F') \leq \xi^\abs{F},
\end{align}
which essentially captures the physical expectation that larger fault paths are exponentially less likely. Suppose now that the quantum circuit under consideration is geometrically local of depth $T$ and we only perform measurement on qubits in a geometrically local set $X$. We denote by $\text{LC}(X)$ the set of locations in the lightcone of $X$. Then, for any fault path $F$ which lies entirely outside the light cone of the measured qubits, we obtain that $q(x|F) = q(x)$. Hence,
\begin{align}
    q_\text{noisy}(x) = \sum_{F: F\cap \text{LC}(X) =  \emptyset}p(F) q(x) + \sum_{F: F\cap \text{LC}(X) \neq  \emptyset}p(F) q(x|F),
\end{align}
and thus
\begin{align}
    \norm{q_\text{noisy} - q} = \sum_x\abs{q_\text{noisy}(x) - q(x)} \leq 2\bigg(\sum_{F: F\cap \text{LC}(X) \neq \emptyset} p(F)\bigg) \leq 2 \sum_{v \in \text{LC}(X)} \sum_{F: F \ni v} p(F) \leq \abs{\text{LC}(X)} \xi \leq O(T^{d + 1} \xi),
\end{align}
where we have assumed $\abs{X} \leq O(1)$ and $T$ is the depth of the circuit.

Now, to estimate the fault tolerance overhead, consider the circuit being simulated with concatenated codes with $L$ levels. We will assume that the base code of the concatenated code has distance $2 t + 1$. As is formalized in the modern proofs of the proof of thereshold theorem \cite{gottesman2024surviving, aliferis2005quantum}, given the encoded circuit with noise, a ``noisy" circuit on the logical qubits can be formally obtained by conjugating the encoded circuit with a unitary decoder for the concatenated code. A particularly elegant mathematical property of the stochastic noise model is that it remains true at any level of concatenation i.e., suppose that the physical qubits experience noise as per the stochastic adversarial noise model (Eq.~\eqref{eq:noise_model}) with $\xi=\xi_0$, then after one layer of cocatentation, the effective error model obtained on the logical qubits is also stochastic but with a reduced noise rate. Using the recursive structure of concatenated codes, it can be shown that after $L$ levels of concatenation, the noise model at the logical qubits is given by Eq.~\eqref{eq:noise_model} with $\xi = \xi_L$ given by \cite{gottesman2024surviving, aliferis2005quantum}
\[
\xi_L = \xi_\text{th}\bigg(\frac{\xi_0}{\xi_\text{th}}\bigg)^{(t + 1)^L},
\]
where $\xi_\text{th}$ is the noise threshold. Thus, if we are performing only local measurements, we obtain that the error in the measurement is
\[
\norm{q_\text{noisy} - q} \leq O\bigg(\xi_\text{th} T^{d + 1}\bigg(\frac{\xi_0}{\xi_\text{th}}\bigg)^{(t + 1)^L}\bigg).
\]
Furthermore, if the circuit was obtained from a $p^\text{th}$ order trotterization of a geometrically local Hamiltonian, we have the total error is given by
\[
O\bigg(\frac{\tau^{d + p + 1}}{T^p}\bigg) + O\bigg( T^{d + 1}\bigg(\frac{\xi_0}{\xi_\text{th}}\bigg)^{(t + 1)^L}\bigg).
\]
The optimal choice of $T = \Theta(\tau (\xi_0/\xi_\text{th})^{-(t+1)^L/(p+d+1)})$ which yields a total error
\[
O\bigg(\tau^{d+ 1}\bigg(\frac{\xi_0}{\xi_\text{th}}\bigg)^{(t + 1)^L p/(p + d+1)}\bigg).
\]
To make this error less than $\varepsilon$, we need a number of layers of concatenation $L = \Theta(\log((1+(d+1)/p) \log(\tau^{d + 1}/\varepsilon)/\log(\xi_\text{th}/\xi))/\log(t+1)) = \Theta(\log \log (\tau^{d + 1}/\varepsilon))$.

\section{Stability of stroboscopic protocols}
We will now consider the setting where we use the Floquet Magnus expansion to implement a target geometrically local Hamiltonian $H_0$ via the Floquet Magnus expansion of a periodic Hamiltonian $H(t)$. We will assume that $H(t)$ is given by
\[
H(t) = \sum_{\alpha} h_\alpha(t),
\]
where $h_\alpha(t) = h_\alpha(t+\uptau)$ is a time-dependent geometrically local operator supported on $\mathcal{S}^{(H)}_\alpha$ and that $\exists J >0$ such that $\norm{h_\alpha(t)}\leq J$. 

Reference~\cite{abanin2017effective} proposed a strategy to get a systematic expansion, as a power series of $\uptau$, of the effective stroboscopic Hamiltonian $H_\text{eff}$ satisfying $\exp(-iH_\text{eff}\uptau) = \mathcal{T}\exp(-\int_0^\uptau H(t) dt)$. For a given fixed order $p \in  \{0, 1, 2 \dots \}$, this expansion is constructed by the following formal procedure: Suppose $\Omega^{(q)}(t)$ are periodic functions of $t$ with period $\uptau$ which also satisfy $\Omega^{(q)}(0) = \Omega^{(q)}(n\uptau) = 0$ and are given recursively by
\begin{subequations}\label{eq:fm_recursion}
\begin{align}\label{eq:fm_recursion_a}
\Omega^{(q + 1)}(t) = -i\int_0^t \big(G^{(q)}(t') - \bar{G}^{(q)}\big)dt' \text{ where }\bar{G}^{(q)} = \frac{1}{\uptau}\int_0^\uptau G^{(q)}(t')dt',
\end{align}
and
\begin{align}\label{eq:fm_recursion_b}
    G^{(q)}(t) = \sum_{k = 1}^q \frac{(-1)^k}{k!} \sum_{\substack{1\leq i_1, i_2 \dots i_k \leq q \\ i_1 +i_2 \dots i_k= q}} \bigg(\prod_{\alpha = 1}^k \mathcal{C}_{\Omega^{(i_\alpha)}(t)} \bigg)H(t) + i\sum_{m = 1}^q \sum_{k=1}^{q+1-m}\frac{(-1)^{k + 1}}{(k+1)!} \sum_{\substack{1\leq i_1, i_2 \dots i_k \leq q + 1 -m \\ i_1 +i_2 \dots i_k= q+1-m}}  \bigg(\prod_{\alpha = 1}^k \mathcal{C}_{\Omega^{(i_\alpha)}(t)} \bigg)\frac{\partial}{\partial t}\Omega^{(m)}(t).
\end{align}
Then, the effective Hamiltonian $H_F^{(p)}$ at order $p$ is given by
\[
H_F^{(p)} = \sum_{q = 0}^p \frac{1}{\uptau} \int_0^\uptau {G^{(q)}(t)},
\]
\end{subequations}
where $\expect{\cdot}_t$ denotes time-averaging i.e., $\expect{f(t)}_t = \int_0^\uptau f(s) ds / \uptau$ if $f$ has period $\uptau$.
The effective Hamiltonian constructed this way is identical to the effective Hamiltonian obtained by the standard Floquet-Magnus expansion \cite{blanes2009magnus}.

We will assume that $h_\alpha(t)$ are chosen such that the $p^\text{th}$ order Floquet-Magnus expansion of $H(t)$, $H_F^{(p)}$ implements $H_0$ to $O(\uptau^p)$ i.e.
\begin{align}\label{eq:target_H}
H_F^{(p)} = \uptau^{p} H_0 + \uptau^{p + 1} V,
\end{align}
where $H_\text{target}$ is the target geometrically local Hamiltonian to be simulated and $V$ is another geometrically local Hamiltonian. Then, to simulate the dynamics of the target Hamiltonian $H_0$ for time $\tau$, $\exp(-iH_0 \tau)$, we would run the simulator Hamiltonian for time $t_\text{sim} = \tau / \uptau^p$. This would approximately yield a unitary $U_\text{sim} \approx \exp(-iH_F^{(p)} t) = \exp(-iH_0 \tau - i \uptau V \tau) $, which in the limit of $\uptau \to 0$ would yield the target dynamics.

% We note that satisfying Eq.~\ref{eq:target_H} might require us to rescale parts of the simulator Hamiltonian $H(t)$ also by $T$. Let us first consider some examples:
% \begin{enumerate}
%     \item[(1)] Suppose the target Hamiltonian $H_0$ was implemented as the first-order Floquet-Magnus expansion of $H(t)$. In this case,
%     \[
%     H_F^{(1)} =\frac{1}{T} \int_0^T H(t)dt = H_0,
%     \]
%     and $V = 0$.
%     \item[(2)] Next, suppose that we needed the second order Floquet Magnus expansion. As an example, suppose the simulator Hamiltonian $H(t)$ was of the form
%     \begin{align}\label{eq:sim_hamiltonian_xxz}
%     H(t) = \sum_{i}  g_i(t)\sigma^x_i \sigma_{i + 1}^x + \sum_i \Omega_i^z(t) \sigma_i^z,
%     \end{align}
%     and we wanted to simulate a target Hamiltonian
%     \[
%     H_0 = \sum_i\big( J_x \sigma_{i}^x \sigma_{i + 1}^x + J_y (\sigma_{i}^x \sigma_{i + 1}^y + \sigma_{i}^y \sigma_{i + 1}^x)\big).
%     \]
%     Since the simulator Hamiltonian $H(t)$ does not contain terms of the form $\sigma_i^y \sigma_{i  + 1}^x, \sigma_i^x \sigma_{i  + 1}^y$, we would need a higher than first-order Floquet-Magnus expansion. One possibility is to pick a stroboscopic period $T$ and pick 
%     \[
%     g_i(t) = T J_x+A_0 \cos\bigg(\frac{2\pi t}{T}\bigg), \Omega_i^z(t) = B_0 \cos\bigg(\frac{2\pi t}{T}\bigg).
%     \]
% \end{enumerate}

To perform the stabiliy analysis for stroboscopic problem-to-simulator mappings, we begin by recalling the key result shown in Ref.~\cite{abanin2017effective}, where they defined a Hamiltonian $\tilde{H}(t)$ via
\begin{align}\label{eq:transformed_floquet_H}
\tilde{H}(t)&= \exp(-\Omega(t)) H(t) \exp(\Omega(t)) -i \exp(-\Omega(t)) \frac{d}{dt} \exp(\Omega(t)), \nonumber \\
&=\exp(-\Omega(t)) H(t) \exp(\Omega(t)) - i\int_0^1 \exp(-s\Omega(t)) \frac{d \Omega(t)}{dt} \exp(s\Omega(t))ds, \nonumber \\
&=\sum_{k = 0}^\infty \frac{(-1)^k}{k!} \sum_{i_1, i_2\dots i_k \leq p} \bigg(\prod_{\alpha=1}^k \mathcal{C}_{\Omega^{(i_\alpha)}(t)}\bigg) H(t) - i\sum_{k =0}^\infty \frac{(-1)^k}{(k + 1)!} \sum_{i_1, i_2\dots i_k \leq p} \bigg(\prod_{\alpha=1}^k \mathcal{C}_{\Omega^{(i_\alpha)}(t)}\bigg) \frac{d\Omega(t)}{dt}.
\end{align}
where $\Omega(t)=\sum_{q=1}^p \Omega^{(q)}(t)$. Since $\Omega(t)$ is periodic with period $\uptau$ and $\Omega(n\uptau)=0$, it follows that the untiary generated by $\tilde{H}(t)$ and $H(t)$ conicide at the stroboscopic times $t=n\uptau$ i.e.,
\begin{align}\label{eq:floquet_equality}
\mathcal{T}\exp\bigg(-i\int_0^\uptau H(t) dt\bigg) = \mathcal{T}\exp\bigg(-i\int_0^\uptau  \tilde{H}(t) dt\bigg).
\end{align}
Furthermore, we can rewrite $\tilde{H}(t)$ as
\[
\tilde{H}(t)=H_F^{(p)} + R^{(p)}(t),
\]
where $R^{(p)}(t)$ can be considered to be the remainder for the Floquet-Magnus expansion. Using the explicit expression for $\tilde{H}(t)$ in Eq.~\ref{eq:transformed_floquet_H}, $R^{(p)}(t)$ can be expressed as 
\begin{align}\label{eq:fm_remainder}
R^{(p)}(t) = \sum_{k = 1}^\infty \frac{(-1)^k}{k!} \sum_{\substack{1 \leq i_1, i_2 \dots i_k \leq p \\ i_1 + i_2 + \dots i_k \geq p + 1 }} \bigg(\prod_{\alpha = 1}^k \mathcal{C}_{\Omega^{(i_\alpha)}(t)}\bigg) H(t) - i\sum_{m=1}^p \sum_{k = 1}^\infty \frac{(-1)^k}{(k + 1)!} \sum_{\substack{1 \leq i_1, i_2 \dots i_k \leq p \\ i_1 + i_2 + \dots i_k \geq p + 2 - m }} \bigg(\prod_{\alpha = 1}^k \mathcal{C}_{\Omega^{(i_\alpha)}(t)}\bigg) \frac{d\Omega^{(m)}(t)}{dt}.
\end{align}
Building on the analysis in Ref.~\cite{abanin2017effective}, we can establish the following useful lemma about the intensive norm of the remainder $R^{(p)}(t)$ in Eq.~\eqref{eq:fm_remainder}. 
\begin{lemma}[Local norm of $R^{(p)}(t)$, modification of Ref.~\cite{abanin2017effective}]\label{lemma:remainder_floquet}
The remainder $R^{(p)}(t)$ defined in Eq.~\ref{eq:fm_remainder} satisfies $\smallnorm{R^{(p)}(t)}_\star \leq C_{p} \uptau^{p + 1}$, for some constant $C_{p}$ that depends only on $p$ but is independent of the system size $N$ and the period $\uptau$.
\end{lemma}
\begin{proof}
We begin by bounding the local norm of $\Omega^{(q)}$ --- for this, we will use Eq.~\eqref{eq:fm_recursion}. Suppose that $\Gamma_q$ are constants such that $\smallnorm{G^{(q)}(t)}_\star = \Gamma_q \uptau^q$. Note from Eq.~\eqref{eq:fm_recursion_a}, it then follows that 
\[
\smallnorm{\bar{G}^{(q)}}_\star \leq \Gamma_q \uptau^q, \norm{\frac{d}{dt}\Omega^{(q + 1)}(t)}_\star  \leq 2\Gamma_q \uptau^q \text{ and } \smallnorm{\Omega^{(q + 1)}(t)}_\star \leq 2\Gamma_q \uptau^{q + 1}.
\]
Additionally $\Gamma_0 = \smallnorm{H}_\star \leq \mathcal{Z}$. Furthermore, from Eq.~\eqref{eq:fm_recursion_b}, we then obtain that
\[
\Gamma_q \leq \sum_{k = 1}^q \frac{2^k}{k!} \sum_{\substack{i_q, i_2 \dots i_k \geq 1 \\ i_1 + i_2  \dots + i_k = q}} \Gamma_{i_1 - 1}\Gamma_{i_2 - 1} \dots \Gamma_{i_k - 1}\Gamma_0 + \sum_{m = 1}^q \sum_{k = 1}^{q + 1 - m}\frac{2^{k + 1}}{(k + 1)!} \sum_{\substack{i_1, i_2 \dots i_k \geq 1 \\ i_1 + i_2  \dots +i_k = q+1 - m}}\Gamma_{i_1 - 1}\Gamma_{i_2 - 1} \dots \Gamma_{i_k - 1}\Gamma_{m - 1}.
\]
Together with $\Gamma_0 = \smallnorm{H}_\star$, this recursion can be solved to obtain $\Gamma_q$ which is independent of $\uptau$ and the system-size $N$. Next, we bound the intensive norm of the remainder $R^{(p)}(t)$ --- using Eq.~\ref{eq:fm_remainder}, assuming $\uptau < 1$ and defining $\tilde{\Gamma}_p = \max_{q \in \{0, 1 \dots p - 1\}} \Gamma_q$, we obtain that
\begin{align*}
\smallnorm{R^{(p)}(t)}_\star &\leq \sum_{k = 1}^\infty \frac{2^k}{k!}  \sum_{\substack{1 \leq i_1, i_2 \dots i_k \leq p \\ i_1 + i_2 + \dots i_k \geq p + 1 }} \bigg(\prod_{\alpha = 1}^k\smallnorm{\Omega^{(i_\alpha)}(t)}_\star\bigg) \smallnorm{H(t)}_\star + \sum_{m = 1}^p \sum_{k = 1}^\infty \frac{2^{k}}{(k + 1)!} \sum_{\substack{1 \leq i_1, i_2 \dots i_k \leq p \\ i_1 + i_2 + \dots i_k \geq p + 2 - m }} \bigg(\prod_{\alpha = 1}^k \smallnorm{\Omega^{(i_\alpha)}(t)}_\star\bigg)\norm{\frac{d}{dt}\Omega^{(m)}(t)}_\star, \nonumber \\
&\leq \uptau^{p + 1} \bigg[\sum_{k = 1}^\infty \frac{4^k}{k!} \Gamma_0 \sum_{\substack{1 \leq i_1, i_2 \dots i_k \leq p \\ i_1 + i_2 + \dots i_k \geq p + 1 }} \bigg(\prod_{\alpha = 1}^k \Gamma_{i_\alpha - 1}\bigg)+ \sum_{m =1}^p \sum_{k = 1}^\infty \frac{4^{k + 1}}{(k + 1)!} \sum_{\substack{1 \leq i_1, i_2 \dots i_k \leq p \\ i_1 + i_2 + \dots i_k \geq p + 1 }}  \bigg(\prod_{\alpha = 1}^k \Gamma_{i_\alpha - 1}\bigg)\Gamma_{m - 1} \bigg], \nonumber\\
&\leq \uptau^{p + 1} \bigg(\sum_{k = 1}^\infty \frac{4^k}{k!} \tilde{\Gamma}_p^{k + 1} p^k + \sum_{k = 1}^\infty \frac{4^{k + 1}}{(k + 1)!} \tilde{\Gamma}_p^{k + 1}p^{k + 1}\bigg)\nonumber\\
&\leq (\tilde{\Gamma}_p + 1) e^{4p\tilde{\Gamma}_p} \uptau^{p + 1} ,
\end{align*}
Since $\Gamma_0,\Gamma_1 \dots \Gamma_p$ are independent of $\uptau$ and the system size $N$, so is $\tilde{\Gamma}_p$ and we obtain the lemma statement.
\end{proof}
\emph{Proof of proposition 3}.  We now consider analyzing the simulator in the presence of errors. As before, we consider a model of errors described by a geometrically local Lindbladian $\mathcal{N} = \sum_\alpha \mathcal{N}_\alpha$ with $\norm{\mathcal{N}_\alpha}_\diamond \leq 1$:
\begin{align}
\frac{d}{dt}\rho(t)= \mathcal{L}_\gamma(t)\rho(t) = -i[H(t), \rho(t)] + \gamma\sum_\alpha \mathcal{N}_\alpha,
\end{align}
For a local observable $O$, we will analyze the deviation between its target expected value $\mathcal{O}_\text{target} = \textnormal{Tr}(O e^{-iH_0 \tau}\rho(0)e^{iH_0 \tau} )$ and the simulated expected value $\mathcal{O}_\textnormal{sim} =  \textnormal{Tr}[O \mathcal{E}_\gamma(t_\text{sim}, 0)\rho(0)]$ where $\mathcal{E}_\gamma(t, s) = \mathcal{T}\exp(\int_s^t \mathcal{L}_\gamma(s') ds')$ and $t_\text{sim} = \tau / \uptau^p$. We will assume that $t_\text{sim}$ is a stroboscopic time i.e., $t_\text{sim} / \uptau \in \mathbb{N} $. It will be convenient to introduce $\tilde{\mathcal{O}}^{(p)}$ given by
\[
\tilde{\mathcal{O}}^{(p)} = \text{Tr}[O_X \exp(-iH_F^{(p)}t_\text{sim})\rho(0)\exp(iH_F^{(p)}t_\text{sim})]
\]
Using the triangle inequality, we obtain that
\begin{align}\label{eq:prop_2_triangle_ineq}
\abs{\mathcal{O}_\text{target} - \mathcal{O}_\text{sim}|_{\gamma = 0}} \leq \abs{\mathcal{O}_\text{target} - \tilde{\mathcal{O}}^{(p)}} + \abs{\tilde{\mathcal{O}}^{(p)} -\mathcal{O}_\text{sim}|_{\gamma = 0}} + \abs{ \mathcal{O}_\text{sim}- \mathcal{O}_\text{sim}|_{\gamma = 0}}.
\end{align}
We first upper bound $\abs{ \mathcal{O}_\text{sim}- \mathcal{O}_\text{sim}|_{\gamma = 0}}$. Since the Lieb-Robinson velocity of $H(t)$ is $O(1)$, it follows from lemma \ref{lemma:lieb_robinson_perturbation_theory} that
\begin{align}\label{eq:prop_2_bound_1}
\abs{ \mathcal{O}_\text{sim}- \mathcal{O}_\text{sim}|_{\gamma = 0}} \leq O\bigg(\frac{\tau^{d + 1}}{T^{p(d + 1)}} \gamma\bigg).
\end{align}
Next, we consider upper bounding $\abs{\mathcal{O}_\text{target} - \tilde{\mathcal{O}}^{(p)}}$---$\mathcal{O}_\text{target}$ can be considered to be the observable obtained from the Hamiltonian $\uptau^p H_F^{(p)}$, whose Lieb-Robinson velocity is $O(\uptau^p)$ after evolution time $t_\text{sim} = \tau / \uptau^p$. Since, by design, $H_F^{(p)} = \uptau^p H_\text{target} + \uptau^{p + 1} V$ where $V$ is a geometrically local Hamiltonian, we obtain from lemma \ref{lemma:lieb_robinson_perturbation_theory} that
\begin{align}\label{eq:prop_2_bound_2}
     \abs{\mathcal{O}_\text{target} - \tilde{\mathcal{O}}^{(p)}} \leq O(\tau^{d + 1} \uptau).
\end{align}
Finally, we consider upper bounding $\abs{ \mathcal{O}_\text{sim}|_{\gamma = 0} - \tilde{\mathcal{O}}^{(p)}}$---from Eq.~\eqref{eq:floquet_equality}, it follows that
\[
\mathcal{O}_\text{sim}|_{\gamma = 0} = \text{Tr}[O \tilde{U}(t_\text{sim}, 0) \rho(0)\tilde{U}(0, t_\text{sim})] \text{ where }\tilde{U}(t, s) = \mathcal{T}\exp\bigg(-\int_s^t \tilde{H}(s')ds'\bigg),
\]
where $\tilde{H}(t)$ is the Hamiltonian defined in Eq.~\eqref{eq:transformed_floquet_H}. Since $\tilde{H}(t) = H_F^{(p)} + R^{(p)}(t)$ where $\smallnorm{R^{(p)}(t)}_\star \leq O(\uptau^{p + 1})$ (lemma \ref{lemma:remainder_floquet}) and the Lieb-Robinson velocity of $H_F^{(p)}$ is $O(\uptau^{p})$ by design, applying lemma \ref{lemma:lieb_robinson_perturbation_theory} we obtain that
\begin{align}\label{eq:prop_2_bound_3}
\abs{\mathcal{O}_\text{sim}|_{\gamma = 0}  - \tilde{\mathcal{O}}^{(p)}} \leq O(\tau^{d + 1}\uptau).
\end{align}
Using Eq.~\eqref{eq:prop_2_triangle_ineq} together with Eqs.~\eqref{eq:prop_2_bound_1}, \eqref{eq:prop_2_bound_2} and \eqref{eq:prop_2_bound_3}, we obtain that
\begin{align}
    \abs{\mathcal{O}_\text{target} - \mathcal{O}_\text{sim}}  \leq O\bigg(\frac{\tau^{d + 1}}{\uptau^{p(d + 1)}} \gamma\bigg) +  O(\tau^{d + 1}\uptau)
\end{align}
A choice of $\uptau = \Theta(\gamma^{1/(p + pd + 1)})$ asymptotically minimizes this upper bound as $\gamma, \uptau \to 0$ --- with this choice of $\uptau$, we obtain that
\[
\abs{\mathcal{O}_\text{target} - \mathcal{O}_\text{sim}} \leq O(\tau^{d + 1}\gamma^{1/(p + pd + 1)}),
\]
which establishes the proposition.

\section{Effective dynamics in the low-energy subspace}
Following Ref.~\cite{abanin2017rigorous}, we formulate a similar scheme in the low-energy subspace. Given a Hamiltonian $H_0$ and a commuting Gauge Hamiltonian $G$, we are interested in considering the Hamiltonian
\begin{align}\label{eq:pert_exp_sim_Hamil}
H = \uptau M + P,
\end{align}
where $\uptau$ is a small parameter. We will assume that $P = \sum_{\alpha} P_\alpha$, where $P_\alpha$ are  orthogonal projectors which also commute i.e., $[P_\alpha, P_\beta]= 0$. Given an operator $A$, we define
\begin{align}\label{eq:def_time_av_pert}
\mathbb{E}_t({A}) =\int_0^{1} e^{i 2\pi s P } A e^{-i 2\pi s P }ds.
\end{align}
Note that $\mathbb{E}_t({A})$ is block-diagonal on the eigen-spaces of $G$. Furthermore, if $A$ is a geometrically local operator, since $P_\alpha$ commute with each other, so is $\mathbb{E}_t(A)$. Given an order $p$, we now consider $\Omega = \Omega^{(1)} + \Omega^{(2)} + \dots + \Omega^{(2)}$, where $\Omega^{(q)} \leq O(\uptau^{q})$ and define $\tilde{H}_0$ via
\[
\uptau \tilde{M} +  P = e^{-\Omega}\big(\uptau M + P \big) e^{\Omega},
\]
To specify the choice of $\tilde{M}$ and $\Omega^{(q)}$, we first define $G^{(q)}$ via
\begin{align}\label{eq:pe_pert_exp}
G^{(0)} = M \text{ and } G^{(q)} = \sum_{k = 1}^{q} \frac{(-1)^k}{k!} \sum_{\substack{i_1, i_2 \dots i_k \geq 1 \\ i_1 + i_2 +\dots i_k = q}} \bigg(\prod_{\alpha = 1}^k \mathcal{C}_{\Omega^{(i_\alpha)}}\bigg) M + \frac{1}{\uptau}\sum_{k = 2}^{q}\frac{(-1)^k}{k!}\sum_{\substack{i_1, i_2 \dots i_k \geq 1 \\ i_1 +i_2 + \dots i_k=q+1}}\bigg(\prod_{\alpha = 1}^k \mathcal{C}_{\Omega^{(i_\alpha)}}\bigg) P.
\end{align}
We point out that $\uptau G^{(q)} - [\Omega^{(q + 1)}, P]$ can be formally considered to be the $O(\uptau^{q +1})$ term in the expression $e^{-\Omega}(\uptau H_0 + P)e^{\Omega}$. Now, we define
\[
\tilde{M}^{(p)} = \sum_{q = 0}^{p} \mathbb{E}_t(G^{(q)}),
\]
and we choose $\Omega^{(q)}$, for $q \geq 1$, via
\[
[\Omega^{(q)}, P] =\uptau (G^{(q-1)} - \mathbb{E}_t(G^{(q - 1)})). 
\]
We note that this equation can be solved for $\Omega^{(q)}$ to obtain
\begin{align}\label{eq:pe_omega}
\Omega^{(q)} =i\uptau \int_0^{1}\int_0^{s_1} e^{i2\pi s_2 P} \big(G^{(q-1)} - \mathbb{E}_t(G^{(q - 1)}) \big)e^{-i2\pi s_2 P} ds_2 ds_1.
\end{align}
We can now express $\tilde{H}_0$ as a sum of $\tilde{H}^{(p)}_0$ and a remainder $R^{(p)}$:
\[
\tilde{M} = \tilde{M}^{(p)} + R^{(p)},
\]
where $R^{(p)}$ can be expressed as the following infinite sum
\begin{align}\label{eq:remainder_time_ind}
R^{(p)} = \sum_{k = 1}^\infty \frac{(-1)^k}{k!} \sum_{\substack{i_1, i_2 \dots i_k \geq 1 \\ i_1 + i_2 + \dots i_k \geq p + 1 }}\bigg(\prod_{\alpha = 1}^k \mathcal{C}_{\Omega^{(i_\alpha)}}\bigg) M  + \frac{1}{\uptau} \sum_{k = 2}^\infty \frac{(-1)^k}{k!} \sum_{\substack{i_1, i_2 \dots i_k \geq 1 \\ i_1 + i_2 + \dots i_k \geq p + 2 }} \bigg(\prod_{\alpha = 1}^k \mathcal{C}_{\Omega^{(i_\alpha)}}\bigg) P.
\end{align}
We will now establish two useful results --- \emph{first}, we will show that $\Omega^{(q)}$ is a quasi-local operator with bounded norm we will establish that $R^{(p)}$ is a quasi-local operator with a local norm 
\begin{lemma}\label{lemma:time_ind_proof_Omega}
    For $q \in \{1, 2 \dots p\}$, $\smallnorm{\Omega^{(q)}}_\star \leq W_{p} \uptau^q$, where $W_{p}$ is a constant that is independent of $\uptau$ as well as the system size $N$.
\end{lemma}
\begin{proof}
 We introduce $\Gamma_q$ such that $\smallnorm{G^{(q)}}_\star = \Gamma_q \uptau^q$. It follows from lemma \ref{lemma:local_norm_prop}(c) together with Eqs.~\eqref{eq:def_time_av_pert} and \eqref{eq:pe_omega} that
 \begin{align}\label{eq:time_avg_upper_bounds}
 \smallnorm{\mathbb{E}_t(G^{(q - 1)})}_\star \leq \sqrt{w} \Gamma_{q - 1}\uptau^{q - 1} \text{ and }\smallnorm{\Omega^{(q)}}_\star \leq w\Gamma_{q - 1}\uptau^{q},
 \end{align}
 where $w = (2r_0)^{2d} $. Finally, we use Eq.~\eqref{eq:pe_pert_exp} to obtain that
 \begin{align}\label{eq:gamma_ineq}
 \Gamma_q \leq \sum_{k = 1}^q \frac{(2w)^k}{k!}\Gamma_0 \sum_{\substack{i_1, i_2 \dots i_k \geq 1\\ i_1+i_2 +\dots i_k= q}
 }\bigg(\prod_{\alpha=1}^k \Gamma_{i_\alpha - 1}\bigg) + \sum_{k=2}^q \frac{(2w)^k}{k!} \sum_{\substack{i_1, i_2 \dots i_k \geq 1\\ i_1+i_2 +\dots i_k= q+1}}\bigg(\prod_{\alpha = 1}^k \Gamma_{i_\alpha-1}\bigg) \smallnorm{P}_\star.
 \end{align}
 Together with $\Gamma_0  = \smallnorm{M}_\star$, these inequalities can be solved to obtain $\Gamma_1, \Gamma_2 \dots \Gamma_p$ --- note that since Eq.~\eqref{eq:gamma_ineq} is independent of the system size $N$ as well as the parameter $\uptau$, we conclude that $\Gamma_1, \Gamma_2 \dots \Gamma_p$ are independent of $N$ and $\uptau$. Finally, from Eq.~\eqref{eq:time_avg_upper_bounds}, we obtain that $\smallnorm{\Omega^{(q)}}_\star \leq W_{p}\uptau^q$, where $W_{p} = w\max_{q \in \{0, 1, 2 \dots p - 1\}} \Gamma_q$.
\end{proof}
\begin{lemma}\label{lemma:remainder_bound_local_pert}
    For $p \in \{1, 2 \dots \}$ and $\uptau < 1$, $\smallnorm{\Omega}_\star \leq \tilde{W}_{p} \uptau $ and $\smallnorm{R^{(p)}}_\star \leq C_{p} \uptau^{p + 1} $, where $\tilde{W}_{p},C_{p} > 0$ are constants independent of $\uptau$ and the system size $N$.
\end{lemma}
\begin{proof}
    Using lemma \ref{lemma:local_norm_prop}(a) together with lemma \ref{lemma:time_ind_proof_Omega}, we immediately obtain that
    \[
    \smallnorm{\Omega}_\star \leq \sum_{q = 1}^p \smallnorm{\Omega^{(q)}}_\star \leq \sum_{q = 1}^p W_{p}\uptau^q \leq \tilde{W}_{p} \uptau,
     \]
     where $\tilde{W}_{p} = \sum_{q = 1}^p W_{q}$ is a constant independent of both $N$ and $\uptau$. 

     Next, we consider bounding the local norm of $R^{(p)}$ defined in Eq.~\eqref{eq:remainder_time_ind}. Using lemma \ref{lemma:local_norm_prop} and defining $\tilde{\Gamma}_p = \max \{\smallnorm{P}_\star, \smallnorm{M_0}_{\star}, W_{1}, W_{2} \dots W_{p}\}$, we obtain
     \begin{align*}
         \smallnorm{R^{(p)}}_\star &\leq \sum_{k = 1}^\infty \frac{2^k}{k!}  \sum_{\substack{1\leq i_1, i_2 \dots i_k \leq p \\ i_1 + i_2 + \dots i_k \geq p + 1 }}\bigg(\prod_{\alpha = 1}^k  \smallnorm{\Omega^{(i_\alpha)}}_\star\bigg)  \smallnorm{M}_\star + \frac{1}{\uptau} \sum_{k = 2}^\infty \frac{2^k}{k!} \sum_{\substack{1\leq i_1, i_2 \dots i_k \leq p \\ i_1 + i_2 + \dots i_k \geq p + 2 }} \bigg(\prod_{\alpha = 1}^k \norm{\Omega^{(i_\alpha)}}_\star\bigg) \norm{P}_\star, \nonumber \\
         &\leq \uptau^{p + 1}\bigg(\sum_{k = 1}^\infty \frac{2^k}{k!} \tilde{\Gamma}_p^{k + 1} p^k + \sum_{k = 2}^\infty \frac{2^k}{k!} \tilde{\Gamma}_p^{k + 1} p^k\bigg), \nonumber \\
         &\leq 2 \tilde{\Gamma}_p e^{2p \tilde{\Gamma}_p} \uptau^{p + 1}.
     \end{align*}
     Since $\norm{P}_\star, \norm{M}_\star, W_{1}, W_{2} \dots W_{p}$ is independent of $N$ and $\uptau$, so is $\tilde{\Gamma}_p$. This establishes the bound on $\smallnorm{R^{(p)}}_\star$ in the lemma statement.
\end{proof}

\emph{Proof of proposition 4}. We consider the problem of quantum simulation of a geometrically local target Hamiltonian $H_0$ via a $p^\text{th}$ order perturbative expansion. As described in the main text, we will assume that the Hamiltonian $M$ and $P$, which determine the Hamiltonian implemented on the simulator (Eq.~\eqref{eq:pert_exp_sim_Hamil}), are designed in such a way that $\tilde{M}^{(p)}$ satisfy
\begin{align}\label{eq:pert_exp_design}
\tilde{M}^{(p)} = \uptau^p H_0 + \uptau^{p + 1}V,
\end{align}
where $V$ is another geometrically local Hamiltonian. To simulate the dynamics of ${H}_0$ for time $\tau$, we will need to run the simulator for time $t_\text{sim} = \tau / \uptau^{p + 1}$. As introduced in the main text, the dynamics of the simulator, in the presence of noise, will be described by the master equation
\[
\frac{d}{dt}\rho(t) = \mathcal{L}_\gamma \rho(t) \text{ where }\mathcal{L}_\gamma = -i[H, \cdot] + \gamma\sum_\alpha \mathcal{N}_\alpha,
\]
where $\mathcal{N}_\alpha$ is a geometrically local Lindbladian with $\norm{\mathcal{N}_\alpha}_\diamond \leq 1$. We will now consider a local observable $O$ and an initial state $\rho(0)$---we will bound the error between the simulated observable $\mathcal{O}_\text{sim} = \text{Tr}(O \mathcal{E}_\gamma(t_\text{sim}, 0) \rho(0))$, where $\mathcal{E}_\gamma(t, s) = \exp(\mathcal{L}_\gamma (t - s))$, and the target observable $O_\text{target} = \text{Tr}(O_X U_\text{target}(\tau, 0) \rho(0) U_0(0, \tau))$, where $U_0(\tau, \sigma) = \exp(-iH_0(\tau - \sigma))$. We will introduce $\tilde{O}^{(p)}$ and $\tilde{Q}^{(p)}$ via
\begin{align*}
    &\tilde{O}^{(p)} = \text{Tr}[O \exp(-i\uptau \tilde{M}^{(p)} t_\text{sim}) \rho(0)\exp(i\uptau \tilde{M}^{(p)} t_\text{sim})], \nonumber\\
    &\tilde{Q}^{(p)} = \text{Tr}[O \exp(-i(\uptau \tilde{M}^{(p)} + \uptau R^{(p)} + P)t_\text{sim})\rho(0) \exp(i(\uptau \tilde{M}^{(p)} + \uptau R^{(p)} + P)t_\text{sim})]
\end{align*}
We first use the triangle inequality to obtain that
    \begin{align}\label{eq:prop_3_triangle_ineq}
        \abs{\mathcal{O}_\text{sim} - \mathcal{O}_\text{target}} \leq \abs{\mathcal{O}_\text{sim} - \mathcal{O}_\text{sim} |_{\gamma = 0}} + \abs{\mathcal{O}_\text{sim}|_{\gamma = 0} - \tilde{O}^{(p)}}+\abs{\tilde{O}^{(p)} - \mathcal{O}_\text{target}}.
    \end{align}
Consider first upper bounding $\abs{\mathcal{O}_\text{sim} - \mathcal{O}_\text{sim}|_{\gamma = 0}}$ using lemma \ref{lemma:lieb_robinson_perturbation_theory}. We note that since $P$, by assumption, is a commuting geometrically local Hamiltonian, the Lieb-Robinson velocity of $H$ in Eq.~\eqref{eq:pert_exp_sim_Hamil} is $O(\uptau)$. Therefore, from lemma \ref{lemma:lieb_robinson_perturbation_theory}, we obtain that
\begin{align}\label{eq:prop_3_bound_1}
\abs{\mathcal{O}_\text{sim} - \mathcal{O}_\text{sim}|_{\gamma = 0}} \leq O(\gamma t_\text{sim} (\uptau t_\text{sim})^{d}) \leq O\bigg(\tau^{d + 1}\frac{\gamma}{\uptau^{p + 1 + pd}}\bigg). 
\end{align}
Another application of lemma \ref{lemma:lieb_robinson_perturbation_theory} allows us to upper bound $\abs{\tilde{\mathcal{O}}^{(p)} - \mathcal{O}_\text{target}}$. Note that $\mathcal{O}_\text{target}$ can be considered to be the observable obtained under evolution with the Hamiltonian $\uptau^{p + 1}H_0$, which has a Lieb-Robinson velocity of $O(\uptau^{p + 1})$, for time $t_\text{sim} = \tau / \uptau^{p + 1}$. Furthermore, from Eq.~\eqref{eq:pert_exp_design}, $\uptau\tilde{M}^{(p)} - \uptau^{p + 1}H_0 = \uptau^{p+ 2}V$ for a geometrically local Hamiltonian $V$ and thus we obtain that $\smallnorm{\uptau \tilde{H}_0^{(p)} - \uptau^{p + 1}H_\text{target}}_\star \leq O(\uptau^{p + 2})$. Therefore, from lemma \ref{lemma:lieb_robinson_perturbation_theory}, we obtain that
\begin{align}\label{eq:prop_3_bound_2}
\abs{\tilde{\mathcal{O}}^{(p)} - \mathcal{O}_\text{target}} \leq O(\tau^{d + 1} \uptau).
\end{align}
Finally, we consider upper bounding $\abs{\mathcal{O}_\text{sim}|_{\gamma = 0} - \tilde{\mathcal{O}}^{(p)}}$. We first note that, by construction, $[\tilde{M}^{(p)}, P] = 0$ and by assumption, $[O, P] = 0$. Thus, we obtain that
\[
\tilde{\mathcal{O}}^{(p)} = \text{Tr}\big[O \exp(-i(\uptau \tilde{M}^{(p)} + P) t_\text{sim}) \rho(0) \exp(i(\uptau \tilde{M}^{(p)} + P)t_\text{sim})\big].
\]
Recalling that $H = e^\Omega(\uptau \tilde{M}+P)e^{-\Omega} = e^\Omega (\uptau \tilde{M}^{(p)} + P + \uptau R^{(p)})e^{-\Omega}$, we use the triangle inequality to obtain:
\begin{align}\label{eq:prop_3_bound_3_triangle}
    \abs{\mathcal{O}_\text{sim}(0) - \tilde{O}^{(p)}} \leq \smallnorm{e^{-\Omega} O_X e^{\Omega} - O_X} + \abs{\tilde{O}^{(p)} - \tilde{Q}^{(p)}}.
\end{align}
To bound $\abs{\tilde{O}^{(p)} - \tilde{Q}^{(p)}}$, we note that the Lieb-Robinson velocity of $\uptau \tilde{M}^{(p)} + P$ is $O(\uptau^{p + 1})$ and, from lemma \ref{lemma:remainder_bound_local_pert}, $\smallnorm{\uptau R^{(p)}}_\star \leq O(\uptau^{p + 2})$. Using lemma \ref{lemma:lieb_robinson_perturbation_theory}, we then obtain that
\begin{align}\label{eq:eigenvalue_pert_error}
\abs{\tilde{\mathcal{O}}^{(p)} - \tilde{\mathcal{Q}}^{(p)}} \leq O(\tau^{d + 1} \uptau).
\end{align}
Furthermore, we note that 
\[
\smallnorm{e^{-\Omega} O e^{\Omega} - O} = \norm{\int_0^1 e^{-s\Omega}[\Omega, O]e^{s\Omega}ds} \leq \norm{[\Omega, O]},
\]
where we have used the fact that $\Omega$ is an anti-Hermitian operator. Expressing $\Omega$ as a sum of local terms, $\Omega = \sum_\alpha \Omega_\alpha$, we then obtain
\begin{align}\label{eq:eigenvector_rotation_error}
    \smallnorm{e^{-\Omega} O e^{\Omega} - O}  \leq 2\norm{O} \sum_{\alpha:\text{supp}(\Omega_\alpha)\cap \text{supp}(O) \neq \emptyset} \norm{\Omega_\alpha} \leq 2\norm{O} \sum_{x \in \text{supp}(O)} \sum_{\alpha: \text{supp}(\Omega_\alpha)\ni x}\norm{\Omega_\alpha} \leq 2 \norm{O}\abs{X} \norm{\Omega}_\star \leq O(\uptau),
\end{align}
where we have used lemma \ref{lemma:remainder_bound_local_pert} to set $\norm{\Omega}_\star\leq O(\uptau)$. Using Eqs.~\eqref{eq:prop_3_bound_3_triangle}, \eqref{eq:eigenvalue_pert_error} and \eqref{eq:eigenvector_rotation_error}, we obtain that 
\begin{align}\label{eq:prop_3_bound_3}
    \abs{\mathcal{O}_\text{sim} - \tilde{\mathcal{O}}^{(p)}} \leq O(\uptau) + O( \tau^{d + 1}\uptau) \leq O(\tau^{d + 1}\uptau)
\end{align}
From Eqs.~\eqref{eq:prop_3_bound_1}, \eqref{eq:prop_3_bound_2} and \eqref{eq:prop_3_bound_3} together with Eq.~\eqref{eq:prop_3_triangle_ineq}, we then obtain that
\[
\abs{\mathcal{O}_\text{sim} - \mathcal{O}_\text{target}} \leq O\bigg(\tau^{d + 1}\frac{\gamma}{\uptau^{p + 1 + pd}}\bigg) + O(\tau^{d + 1}\uptau).
\]
A choice of $\uptau = \Theta(\gamma^{1/(p + pd + 2)})$ asymptotically minimizes this upper bound as $\gamma, \uptau \to 0$ --- with this choice of $\uptau$, we obtain that
\[
\abs{\mathcal{O}_\text{sim}- \mathcal{O}_\text{target}} \leq O(\tau^{d + 1}\gamma^{1/(p + pd + 2)}),
\]
which establishes the proposition.
\end{document}